\documentclass[acmsmall]{acmart}

\acmJournal{PACMPL}
\acmNumber{CONF} 
\acmYear{2023}
\acmMonth{8}
\acmDOI{} 
\startPage{1}

\setcopyright{none}

\usepackage{amsmath,amsthm,amsbsy}
\usepackage{graphicx}
\usepackage{caption,wrapfig,picins}
\usepackage{subcaption}
\usepackage{longtable}
\usepackage{booktabs}
\usepackage{tabu}
\usepackage{tabularx}
\usepackage[referable]{threeparttablex}
\usepackage{multirow}
\usepackage{subcaption}
\usepackage{hyperref}
\usepackage[capitalise]{cleveref}
\usepackage{stmaryrd}
\usepackage{mathtools} 

\crefname{figure}{Figure}{Figure}
\crefformat{footnote}{#2\footnotemark[#1]#3}
\usepackage{tikz}
\usepackage{comment}
\usepackage{parskip}
\usepackage{paralist}
\usepackage{pgfplots}

\usepackage[ruled, linesnumbered, noend]{algorithm2e}
\usepackage{tikz}
\usetikzlibrary{shapes.misc,arrows.meta}
\usetikzlibrary{calc,decorations.pathmorphing,shapes}
\usetikzlibrary{decorations.markings}
\usetikzlibrary{decorations.pathreplacing}
\tikzset{snake it/.style={decorate, decoration=snake}}

\usepackage{xspace}
\usepackage{xcolor}
\usepackage[inline, shortlabels]{enumitem}

\usepackage{multicol}
\usepackage{bbm}

\usepackage{thm-restate}

\newcommand{\figlabel}[1]{\label{fig:#1}}
\newcommand{\figref}[1]{Fig.~\ref{fig:#1}}
\newcommand{\seclabel}[1]{\label{sec:#1}}
\newcommand{\secref}[1]{Section~\ref{sec:#1}}

\newcommand{\deflabel}[1]{\label{def:#1}}
\newcommand{\defref}[1]{Definition~\ref{def:#1}}
\newcommand{\problabel}[1]{\label{prob:#1}}
\newcommand{\probref}[1]{Problem~\ref{prob:#1}}
\newcommand{\thmlabel}[1]{\label{thm:#1}}
\newcommand{\thmref}[1]{Theorem~\ref{thm:#1}}
\newcommand{\proplabel}[1]{\label{prop:#1}}
\newcommand{\propref}[1]{Proposition~\ref{prop:#1}}
\newcommand{\lemlabel}[1]{\label{lem:#1}}
\newcommand{\lemref}[1]{Lemma~\ref{lem:#1}}

\newcommand{\claimlabel}[1]{\label{claim:#1}}
\newcommand{\claimref}[1]{Claim~\ref{claim:#1}}
\newcommand{\itmlabel}[1]{\label{itm:#1}}

\newcommand{\equlabel}[1]{\label{eq:#1}}
\newcommand{\equref}[1]{Equation~(\ref{eq:#1})}
\newcommand{\applabel}[1]{\label{app:#1}}
\newcommand{\appref}[1]{Appendix~\ref{app:#1}}
\newcommand{\algolabel}[1]{\label{algo:#1}}
\newcommand{\algoref}[1]{Algorithm~\ref{algo:#1}}

\makeatletter
\def\thm@space@setup{
\thm@postskip=0pt}
\makeatother

\theoremstyle{definition}
\newtheorem{example}{Example}
\numberwithin{example}{section}

\newtheorem{theorem}{Theorem}
\numberwithin{theorem}{section}

\newtheorem{definition}{Definition}
\numberwithin{definition}{section}

\newtheorem{problem}{Problem}
\numberwithin{problem}{section}

\newtheorem{lemma}{Lemma}
\numberwithin{lemma}{section}

\newtheorem{proposition}{Proposition}
\numberwithin{proposition}{section}

\numberwithin{corollary}{section}

\newtheorem{claim}{Claim}
\numberwithin{claim}{section}

\newtheorem{remark}{Remark}

\newenvironment{myparagraph}[1]{\vspace{0in}\noindent{\bf #1.}}{}

\DeclareMathAlphabet{\mathpzc}{OT1}{pzc}{m}{it}

\newcommand{\natsp}{\mathbb{N}_{>0}}
\newcommand{\set}[1]{\{#1\}}
\newcommand{\setpred}[2]{\{#1 \,|\, #2\}}

\renewcommand{\emptyset}{\varnothing}
\newcommand{\concat}{}
\newcommand{\Transducer}{\mathcal{M}}
\newcommand{\range}[1]{[1 .. #1]}  
\newcommand{\partialto}{\hookrightarrow}
\newcommand{\vect}[1]{\overline{#1}}

\newcommand{\powset}[1]{\mathcal{P}(#1)}

\newcommand{\evt}[2]{\langle #2,#1 \rangle}
\newcommand{\ev}[1]{\langle #1\rangle}
\newcommand{\vars}{\mathcal{X}}

\newcommand{\threads}{\mathcal{T}}

\newcommand{\ThreadOf}[1]{\mathsf{thr}(#1)}
\newcommand{\OpOf}[1]{\mathsf{op}(#1)}
\newcommand{\VariableOf}[1]{\mathsf{var}(#1)}

\newcommand{\events}[1]{\mathsf{Events}_{#1}}
\newcommand{\tr}{w}

\newcommand{\rdf}[1]{\mathsf{rf}_{#1}}
\newcommand{\po}[1]{\mathsf{po}_{#1}}

\newcommand{\reads}[1]{\mathsf{Reads}_{#1}}

\newcommand{\eqcl}[2]{[#1]_{#2}}

\newcommand{\maz}{\mathcal{M}}

\newcommand{\mazeq}{\equiv_{\maz}}

\newcommand{\indrel}{\mathbb{I}}
\newcommand{\deprel}{\mathbb{D}}
\newcommand{\alphabet}{\Sigma}

\newcommand{\noval}{\mathsf{Conc}}

\newcommand{\rfnovaleq}{\equiv_{\rdf{}}}
\newcommand{\checkOrderText}{\textsf{causallyOrdered}}
\newcommand{\checkOrder}[4]{\checkOrderText_{#1}(#2, #3, #4)}

\newcommand{\rfcl}[1]{\eqcl{#1}{\rdf{}}}
\newcommand{\focalEv}{\diamond}
\newcommand{\shasha}{\maz}

\newcommand{\grain}{g} 

\newcommand{\subsq}[1]{\mathbf{i}}

\newcommand{\grains}{G} 

\newcommand{\grainind}{\indrel_{\grains}}

\newcommand{\actv}[1]{\mathsf{Active}_{#1}}
\newcommand{\grainIDs}{\mathsf{gIDs}}

\newcommand{\contents}{\mathsf{C}}
\newcommand{\summaries}{\mathsf{SC}}
\newcommand{\pvars}{\mathsf{P}} 
\newcommand{\spvars}{\mathsf{SP}}

\newcommand{\deadpath}[1]{\rightsquigarrow_{#1}}
\newcommand{\splitGrains}{\mathsf{split}}

\newcommand{\gMarkAlphabet}{\widehat{\alphabet}}
\newcommand{\GGraph}[1]{\mathcal{G}_{#1}}
\newcommand{\SGGraph}[1]{\mathcal{SG}_{#1}}
\newcommand{\focalGrain}[1]{\grain^{\focalEv_{#1}}}
\newcommand{\mergeEdge}[2]{\mathit{mrg}({#1},{#2})}
\newcommand{\mergeSum}[3]{\mathit{mrgSm}({#1},{#2},{#3})}
\newcommand{\depEdge}[1]{\mathit{Dep}(#1)}
\newcommand{\stuple}[1]{\langle{#1}\rangle}
\newcommand{\cmpvar}[1]{\mathsf{complete}(#1)}
\newcommand{\depend}[1]{\mathit{depend}(#1)}

\newcommand{\hb}{\prec_{\maz}}

\newcommand{\meq}{\equiv_{\mathcal{M}}}

\newcommand{\after}[1]{\mathsf{After}}
\newcommand{\blkafter}[1]{\mathsf{BlocksAfter}}

\newcommand{\aut}{\mathcal{A}}
\newcommand{\autsup}[1]{\aut_{#1}}

\newcommand{\opfont}[1]{\texttt{#1}}

\newcommand{\rd}{\opfont{r}}
\newcommand{\wt}{\opfont{w}}

\newcommand{\bgn}{\rhd} 
\newcommand{\egn}{\lhd} 

\newcommand{\poly}{\operatorname{poly}}

\newcommand{\langeq}{L^=}

\newcommand{\acr}[1]{\textsf{#1}}
\newcommand{\acrtr}{\operatorname{\acr{tr}}}

\newcommand{\ord}[2]{\leq^{#1}_{\mathsf{#2}}}

\newcommand{\strictord}[2]{<^{#1}_{\mathsf{#2}}}
\newcommand{\notstrictord}[2]{\not<^{#1}_{\mathsf{#2}}}

\newcommand{\trord}[1]{\ord{#1}{\acrtr}}

\newcommand{\ctrf}[1]{\strictord{#1}{\rfnovaleq}}
\newcommand{\notctrf}[1]{\notstrictord{#1}{\rfnovaleq}}

\colorlet{RED}{red}

\newcommand{\execution}[2]{
\scalebox{1}{
  \begin{tikzpicture}%
    \foreach \x in {1,...,#1}
    \node[right] at (1.5*\x+0.2,0.25) {$t_{\x}$};
    \draw (1.2,0) -- (#1*1.5+1.3,0);%
    \pgfmathsetmacro{\y}{1};%
    #2%
    \draw (1.2,0) -- (1.2,-0.4*\y);%
    \draw (#1*1.5+1.3,0) -- (#1*1.5+1.3,-0.4*\y);%
    \foreach \x in {2,...,#1}
    \draw (1.2,-0.4*\y) -- (#1*1.5+1.3,-0.4*\y);%
  \end{tikzpicture}
}
}

\newcommand{\figev}[2]{
\pgfmathsetmacro{\y}{\y+1};
\pgfmathsetmacro{\y}{\y-1};
\node [left] at (1.25,-0.4*\y)  {\pgfmathprintnumber{\y}};%
\node at (#1*1.5 + 0.45,-0.4*\y) { #2 };%
\pgfmathsetmacro{\y}{\y+1};
}

\newcommand{\separatorlight}[4]{
  \draw[loosely dashed] (1.5 + 0.45 + #3, -0.4*#2 -0.2) -- (#1*1.5 + 0.45 + #4, -0.4*#2 - 0.2);
}

\newcommand{\separatordark}[4]{
  \draw[thick, dash dot] (1.5 + 0.45 + #3, -0.4*#2 -0.2) -- (#1*1.5 + 0.45 + #4, -0.4*#2 - 0.2);
  \draw[thick, dash dot] (1.5 + 0.45 + #3, -0.4*#2 -0.25) -- (#1*1.5 + 0.45 + #4, -0.4*#2 - 0.25);
}

\newcommand{\demarcate}[5]{
  \draw (#1*1.5 + 0.45 + #4, -0.4*#2) edge[-{Latex[length=2mm, width=2mm]},  <->,>=stealth]  (#1*1.5 + 0.45 + #4, -0.4*#3);
  \node[fill=white] (dummy) at (#1*1.5 + 0.45 + #4, -0.2*#2 -0.2*#3) {#5};
}

\newcommand{\executionfull}[9]{
\scalebox{#3}{
  \begin{tikzpicture}
    \foreach \x in {1,...,#1}
    \node[right] at (#4*\x+#6, #8) {$t_{\x}$}; 
    \draw (#7,0) -- (#1*#4+#7,0); 
    \pgfmathsetmacro{\y}{1};
    #2 
    \draw (#7,0) -- (#7,-#5*\y); 
    \draw (#1*#4+#7,0) -- (#1*#4+#7,-#5*\y); 
    \draw (#7,-#5*\y) -- (#1*#4+#7,-#5*\y); 
    \ifthenelse{#9 = 1}{
      \foreach \x in {2,...,#1}
      \draw[dashed] (#4*\x+#7-#4,0) -- (#4*\x+#7-#4,-#5*\y); 
    }{}
  \end{tikzpicture}
}
}

\newcommand{\figevfull}[9]{
\ifthenelse{#7 = 1}{
  \ifthenelse{#8 = -1}{
    \node [left] at (#5,(-#4*\y)  {\pgfmathprintnumber{\y}};%
  }{
    \node [left] at (#5,(-#4*\y)  {#9};%
  }
}{}
\node at (#1*#3 + #6,(-#4*\y) {$ #2 $};%
\pgfmathsetmacro{\y}{\y+1};
}

\begin{document}

\title{Coarser Equivalences for Causal Concurrency}

\author{Azadeh Farzan}
\email{azadeh@cs.toronto.edu}
\affiliation{%
  \institution{University of Toronto}
  \city{Toronto}
  \country{Canada}
}

\author{Umang Mathur}
\email{umathur@comp.nus.edu.sg}
\affiliation{%
  \institution{National University of Singapore}
  \city{Singapore}
  \country{Singapore}
}


\begin{abstract}
\emph{Trace theory} (formulated by Mazurkiewicz in 1987) 
is a principled framework for defining equivalence relations for concurrent program runs
based on a commutativity relation over the set of atomic 
steps taken by individual program threads. 
Its simplicity, elegance, and algorithmic efficiency makes it useful in many different contexts
including program verification and testing.
It is well-understood that the larger the equivalence classes are, 
the more benefits they would bring to the algorithms and applications that use them. 
In this paper, we study relaxations of trace equivalence with the goal of maintaining its algorithmic advantages. 

We first prove that the largest appropriate relaxation of trace equivalence,  an equivalence relation that preserves the order of steps taken by each thread \emph{and} what write operation each read operation observes, does not yield efficient algorithms. Specifically, we prove a \emph{linear space lower bound} for the problem of checking, in a streaming setting, if two arbitrary steps of a concurrent program run are \emph{causally concurrent} (i.e. they can be reordered in an equivalent run) or \emph{causally ordered }(i.e. they always appear in the same order in all equivalent runs). The same problem can be decided in \emph{constant space} for trace equivalence.
Next, we propose a new commutativity-based notion of equivalence called \emph{grain equivalence} that is 
strictly more relaxed than trace equivalence, and yet yields a constant space algorithm for the same problem. 
This notion of equivalence uses commutativity of \emph{grains}, which are sequences of atomic steps, 
 in addition to the standard commutativity from trace theory. We study the two distinct cases when the grains are contiguous subwords of the input program run and when they are not, formulate the precise definition of causal concurrency in each case, and show that they can be decided in {\em constant space}, despite being strict relaxations of the notion of causal concurrency based on trace equivalence.   

\end{abstract}

\maketitle

\section{Introduction}
\seclabel{intro}

In the last 50 years, several models have been introduced for concurrency and parallelism, of which Petri nets \cite{hack1976petri}, Hoare's CSP \cite{hoare1978communicating}, Milner's CCS \cite{milner1980calculus}, and event structures \cite{Winskel1987} are prominent examples. Trace theory \cite{diekert1995book} is a paradigm in the same spirit which enriches words (or sequences) by a very restricted yet widely applicable mechanism to model parallelism:  some pairs of {\em events} (atomic steps performed by individual threads) are determined statically to be {\em independent} (or commutative), and any two sequences that can be transformed to each other through swaps of consecutive independent events are identified as {\em trace equivalent}. In other words, it constructs a notion of equivalence based on {\em commutativity} of individual events. The simplicity of trace theory, first formulated by Mazurkiewicz in 1987~\cite{Mazurkiewicz87}, has made it highly popular in a number of areas in computer science, including programming languages, distributed computing, computer systems, and software engineering. 
The brilliance of trace theory lies in its simplicity, both conceptually and in yielding simple and efficient algorithms for several core problems in the context of concurrent and distributed programs. 
It has been widely used in both dynamic program analysis and in construction of program proofs. In dynamic program analysis, it has applications in predictive testing, for instance in data race prediction \cite{djit1999,Flanagan09,Elmas07,Smaragdakis12,Kini17} and in prediction of atomicity violations~\cite{Farzan09,Sorrentino10,Farzan2006} among others. In verification, it is used for the simplification of verification of concurrent and distributed programs \cite{lics2023,mycav19,mypopl20,mypldi22,Genest07,psynch,Desai14}.

The philosophy behind most applications of trace theory is that a single representative replaces an entire set of equivalent runs. Therefore, these applications would clearly benefit if larger sets of concurrent runs could soundly be considered equivalent. This motivates the key question in this paper: Can we retain all the benefits of classical trace theory while soundly enlarging the equivalence classes to improve the algorithms that use them? 
It is not difficult to come up with sound equivalence relations with larger classes. 
\citet{Abdulla2019} describe the sound equivalence 
\begin{wrapfigure}{r}{0.3\textwidth}
\vspace{-10pt}
\includegraphics[scale=0.75]{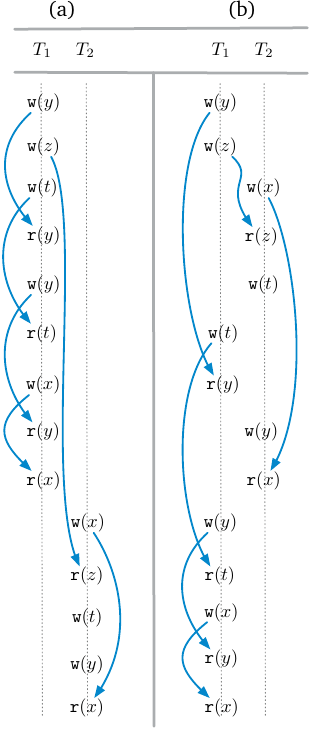}
\vspace{-10pt}
\caption{Read-From Equivalence}\label{fig:intro}
\vspace{-20pt}
\end{wrapfigure}
relation {\em reads-from equivalence} which has the largest classes that remain sound and is defined in a way to preserve two essential properties of 
concurrent program runs: (1) the total order of {\em events} in each thread must remain the same in every equivalent run, 
and (2) each read event must observe the value written from the same write event in every equivalent run.   
The former is commonly known as the preservation of {\em program order}, 
and the latter as the preservation of the {\em reads-from} relation. These conditions guarantee that any equivalent run is also a valid run of the same concurrent program, and all the observed values (by read events) remain the same, which implies that the outcome of the program run must remain the same. In other words, both control flow and data flow are preserved.

Consider the concurrent program run illustrated in \figref{intro}(a), and let us focus on the order of the read event $\rd(x)$ from thread $T_1$ and the write event $\wt(x)$ from thread $T_2$. In classical trace theory, these events are dependent (or non-commutatitive) and they must be appear in the same order in every member of the equivalence class of the run. This implies that the illustrated run belongs in an equivalence class of size 1. On the other hand, there exist other reads-from equivalent runs; one such run is illustrated in \figref{intro}(b), in which the two aforementioned events have been reordered. The arrows connect each write event to the read event that reads its value, which remain unchanged between the two runs.

We give a formal argument for why (the more relaxed) reads-from equivalence is not as useful as trace equivalence from an algorithmic standpoint. One of the most fundamental algorithmic questions in this context is: Given two events $e$ and $e'$ in a run $\rho$, do they always appear in the same order in every member of the equivalence class of $\rho$ or can they be reordered in an equivalent run? In the former case, we call $e$ and $e'$ {\em causally ordered}, and otherwise {\em causally concurrent}.  For example, the events $\rd(x)$ from thread $T_1$ and $\wt(x)$ from thread $T_2$ are causally ordered under trace equivalence but causally concurrent under reads-from equivalence. On the other hand, $\wt(z)$ event of thread $T_1$ and the $\rd(z)$ event of thread $T_2$ are causally ordered under both equivalence relations.

For trace equivalence, it is possible to decide if two events are causally ordered or concurrent, using a {\em single pass constant space} algorithm; if the length of the run is assumed to be the size of the input, and other program measures such as the number of threads and the number of shared variables are considered to be constants. 
In \secref{semantic-equivalence}, we prove that if {\em equivalence} is defined as broadly as reads-from equivalence, then this check can no longer be done in {\em constant space} by proving a {\em linear space} lower bound (\thmref{check-order-semantic-linear-space-lower-bound}). In particular, we prove this for two closely related variants of this problem: (1) the decision about the ordering of two specific events (as discussed in the example of \figref{intro}), and (2) the decision about the ordering of any two occurrences of a specific (atomic) action (e.g. are any two $\wt(x)$ actions unordered?). 
Both problems are closely related to predictive testing of violations of 
generic correctness properties for concurrent programs, 
such as data race freedom \cite{djit1999,Flanagan09,Elmas07,Smaragdakis12,Kini17,Huang14,Pavlogiannis2020}, deadlock freedom~\cite{Kalhauge2018,Tunc2023} and atomicity \cite{Farzan09,Sorrentino10,Farzan2006,Mathur2020},
and have applications in dynamic partial order reduction techniques~\cite{Flanagan2005,Abdulla2019,Kokologiannakis2022}
for model checking of concurrent programs. In all such contexts, having a monitor whose state space does not depend on the length of the input program run, which may include billions of events, is highly desirable. Thus far, the only existing instance of such a monitor has been based on trace equivalence. 

We propose a new notion of equivalence for concurrent program runs, which in terms of expressivity lies in between trace and reads-from equivalences, and retains the highly desirable algorithmic simplicity of traces. The idea is based on enriching the classical commutativity relation of trace theory to additionally account for commutativity of certain {\em sequences of events}, called {\em grains}. A grain can be an arbitrarily long sequence of operations,  which can belong to multiple threads. What motivates this definition is that in places where swapping a pair of individual events may not be possible, groups of operations (as grains) may still commute {\em soundly}, meaning without disturbing the program order or the reads-from relation. For two grains to be swappable, they must be adjacent and contiguous, at the time they are swapped. 

\begin{wrapfigure}{r}{0.35\textwidth}
\vspace{-15pt}
\includegraphics[scale=0.8]{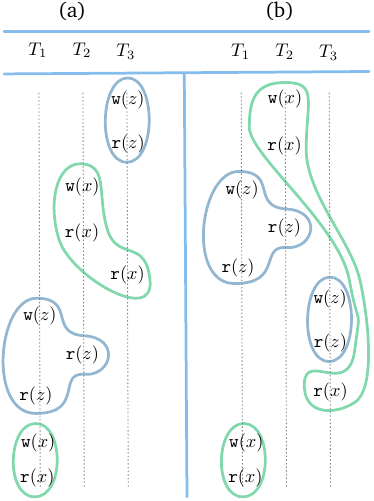}
\caption{Commuting Grains}\label{fig:intro2}
\vspace{-20pt}
\end{wrapfigure}
As an example, consider the concurrent program run in \figref{intro2}(a). Four grains are marked. Observe that the grains of the same colour soundly commute. Grains of different colour always share a thread and therefore commuting them would break program order. First, assume that only the two blue grains exist. One can then use a sequence of swaps that are either standard swaps permitted by trace equivalence, or a swap of the two blue grains, and transform the run in (a) to the one illustrated in \figref{intro2}(b), and hence reorder the two $\wt(z)$ events. We call the run in (b) to be {\em grain-equivalent} to the one in (a). Observe that this is not possible under trace equivalence. A similar observation can be made if we consider only the green grains, and the goal of reordering the two $\wt(x)$ operations. However, if all four grains are considered together, since the grains of different colour do not commute, nothing can be swapped in \figref{intro2}(a). Once something is decided to be part of a grain, it must always move together with the rest of the grain. This turns out to be a key to algorithmic simplicity. 

In Section \ref{sec:geq}, we formally define the set of runs that are soundly equivalent to a program run based on a choice of grains that, as in \figref{intro2}(a), appear as contiguous subwords of the program run. We observe that different choices of grains can imply the concurrency of different pairs of events. We define {\em grain concurrency} of two events formally as the existence of a choice of grains and a sound commutativity relation (over the grains) such that the two events can be (soundly) reordered in a {\em grain-equivalent run} up to those choices. 

The construction of a monitor based on trace equivalence vitally relies on the fact that the alphabet of concurrent programs is finite and consequently there are finitely many choices of commutativity relations over this alphabet. Note that, once any subword can become a new entity that participates in commutativity-based reasoning, the number of these subwords is no longer bounded. Neither is the set of possible commutativity relations over them. 
In Section \ref{sec:gmonitor}, we prove by construction that {\em grain-equivalence} can be monitored in constant-space. This construction relies on a key insight that the monitor can maintain {\em summaries} for these grains that belong to a finite universe, and correctly check the concurrency status of two events. 



Let us revisit the example runs in \figref{intro2}. The run in \figref{intro2}(b) is the result of commuting two blue grains in the run in \figref{intro2}(b). The first green grain, however, is no longer a contiguous subword and therefore cannot be seen as a potential grain in the solution we have outlined so far. In Section \ref{sec:sgrains}, we formally expand the definition of grains so that these so-called {\em scattered grains} can be considered as candidates as well. This requires a leap in the definition of the set of words that are soundly equivalent to the input run. In the case of contiguous grains, the set of equivalent runs maintains the characteristic of classic trace theory that one can transform the input run to any inferred equivalent run through a sequence of valid swaps. With {\em scattered grains}, we establish soundness by deferring to the more general definition of {\em reads-from equivalence}, and consequently forfeit the characteristic of transformation through swaps.

Surprisingly, however, this weaker definition still maintains the property that it can be monitored in constant-space. First, there is the additional complication that unlike contiguous subwords where at most one grain is {\em active} (open) at any given time, there may be an unbounded number of {\em active} scattered grains in a concurrent program run at any given time. Nevertheless, we prove that if an outcome can be decided based on an arbitrary choice of scattered grains, then it can also be decided based on a choice of scattered grains in which the number of active grains is bounded. With a bounded number of active grains (in contrast to just a single one), there is another complication where two active grains, which belong to the middle of a chain witnessing that $e$ and $e'$ are ordered, are only discovered to be ordered long after $e$ and $e'$ have been visited. In Section \ref{sec:beyond-contiguous-grains}, we construct a monitor that resolves these problems and soundly checks the concurrency of a pair of events based on all possible choices of scattered grains, which we call {\em scattered grain concurrency}.

Even though the {\em scattered grain concurrency} monitor subsumes the {\em grain concurrency} monitor in expressivity, the paper presents these monitors separately, since the former admits a characterization based on swaps and the latter does not. Moreover, this permits us to introduce the relevant ideas in tandem with the problems they help solve. In summary, the paper presents the following results:
\begin{itemize}
\item We prove that there is no constant-space algorithm that checks if two arbitrary events are causally ordered or concurrent under the reads-from equivalence relation 
by establishing a linear space lower bound, a quadratic time-space tradeoff bound, a conditional quadratic time lower bound as well as the non context-freeness of this problem. We complement these lower bounds by showing that the problem can nevertheless be solved in polynomial time as well as in deterministic linear space.  (\secref{semantic-equivalence}). 
\item We propose {\em grain equivalence} as the means of defining a set of runs that are soundly equivalent to a given program run and form a strictly larger set than the trace equivalence class of the run, and a strictly smaller set than the reads-from equivalence class of it. This notion of equivalence is constructed based on  commutativity of certain pairs of words over the alphabet of concurrent program operations that appear as contiguous subwords of the input program run (Section \ref{sec:geq}). 
\item We introduce the notion of {\em grain concurrency} that attempts to find a witness for causal concurrency of a pair of events based on all possible choices of grains. We further weaken the definition of {\em grain concurrency} by permitting {\em scattered grains} to be soundly used for reasoning in equivalence and call this {\em scattered grain concurrency} (Section \ref{sec:sgrains}).
\item We give a space efficient algorithm that soundly checks {\em grain concurrency} for a pair of events, i.e.  whether the two events are causally ordered or concurrent up to {\em grain equivalence} (Section \ref{sec:gmonitor}),  
and a space efficient algorithm that soundly checks {\em scattered grain concurrency} based on all possible choices of scattered grains (Section \ref{sec:beyond-contiguous-grains}).
\end{itemize}


\section{Preliminaries}
\seclabel{prelim}

A string over an alphabet $\alphabet$ is a finite sequence
of symbols from $\alphabet$.
We use $|w|$ to denote the length of the string $w$ and $w[i]$
to denote the $i^{\text{th}}$ symbol in $w$.
The concatenation of two strings $w, w'$ will be denoted by $w \concat w'$.

\subsection{Trace Equivalence}
Antoni Mazurkiewicz~ popularized the use of partially commutative monoids
for modelling executions of concurrent systems\cite{Mazurkiewicz87}. 
We discuss this formalism here.
An independence relation over $\alphabet$ is a symmetric
irreflexive binary relation $\indrel \subseteq \alphabet \times \alphabet$.
The Mazurkiewicz equivalence (or trace equivalence) relation induced by 
$\indrel$, denoted\footnote{The equivalence is parametric on the independence relation $\indrel$. In this view, 
a notation like $\equiv^\indrel_\maz$ would be more precise. In favor of readability, we skip this parametrization; the independence relation $\indrel$ will always be clear.} $\mazeq$
 is then the smallest equivalence
over $\alphabet^*$
such that for any two strings $w, w' \in \alphabet^*$
and for any two letters $a, b \in \alphabet$ with $(a, b) \in \indrel$,
we have $w \concat a \concat b \concat w' \mazeq w \concat b \concat a \concat w'$. 
A Mazurkiewicz trace is then an equivalence class of $\mazeq{}$.

\myparagraph{Mazurkiewicz partial order}{
	An equivalence class of $\mazeq{}$ can be succinctly represented using
	a partial order on the set of \emph{events} in a given string.
	Events are unique identifiers for the different occurrences of symbols in a string. 
	Formally, for a string $w$, the set of events of $w$, denoted $\events{w}$ is the set of pairs
	of the form $e = (a, i)$ such that $a \in \alphabet$ and there are at least $i$ occurrences
	of $a$ in $w$.
	Thus, an event uniquely identifies a position in the string --- the pair $(a, i)$
	corresponds to the unique position $j \leq |w|$ such that $w[j] = a$
	and there are exactly $i-1$ occurrences of $a$ before the index $j$ in $w$.
	Observe that if $w$ and $w'$ are permutations of each other, then $\events{w} = \events{w'}$.
	For an event $e = (a, i)$, we use the shorthand $w[e]$ to denote the label $a$.
	Often, we will use the position $j$ itself to denote the event 
	$(a, i)$ when the distinction is immaterial.
	Fix $w \in \alphabet^*$.
	The Mazurkiewicz (or trace) partial order for $w$, denoted $\hb$ is 
	then, the transitive closure of the relation 
	$\setpred{(e, f)}{e, f \in \events{w},\ e \text{ occurs before } f \text{ in }w\ \wedge\ (w[e], w[f]) \in \deprel}$.
}

For Mazurkiewicz traces, the corresponding partial order is a
sound and complete representation of an equivalence class ~\cite{Mazurkiewicz87}.
%

\subsection{Concurrent alphabet and dependence}
For modeling runs or executions of shared memory multi-threaded concurrent programs,
we will consider the alphabet consisting of reads and writes.
Let us fix \emph{finite} sets $\threads$ and $\vars$ of 
thread identifiers and memory location identifiers respectively.
The concurrent alphabet $\alphabet_\noval$ we consider in the rest of the paper is:
\begin{align*}
\alphabet_\noval = \setpred{\ev{t, o, x}}{t \in \threads, o \in \set{\rd, \wt}, x \in \vars}
\end{align*}
For a symbol $a = \ev{t, o, x} \in \alphabet_\noval $, we say
$\ThreadOf{a} = t$, $\OpOf{a} = o$ and $\VariableOf{a} = x$.
A concurrent program \emph{run} or \emph{execution} is a string over $\alphabet_\noval$.
For a run $w \in \alphabet_\noval^*$ and event $e \in \events{w}$,
we overload the notation and use $\ThreadOf{e}$, $\OpOf{e}$ and $\VariableOf{e}$
in place of $\ThreadOf{w[e]}$, $\OpOf{w[e]}$ and $\VariableOf{w[e]}$ respectively.
Since the focus of the rest of the article will be on concurrent program runs, 
we will omit the subscript $\noval$,
and unless $\alphabet$ is not explicitly defined, we will assume it is $\alphabet_\noval$.

We use the following 
independence (commutativity) relation:
\begin{align}
\equlabel{shasha}
\begin{array}{rcl}
\indrel_\shasha &= & \alphabet \times \alphabet - \setpred{(a, b)}{\ThreadOf{a} = \ThreadOf{b} \lor (\VariableOf{a} = \VariableOf{b} \land \wt \in \set{\OpOf{a}, \OpOf{b}})}
\end{array}
\end{align}
This defines an appropriate trace monoid for the alphabet of concurrent program actions. We refer to the equivalence classes a concurrent program run $\tr$ in this monoid as $[\tr]_\maz$. This trace monoid is provably sound in the following sense:
\begin{remark}\label{rem:1}
For a given concurrent program run $\tr \in \Sigma$, every member of $[\tr]_\maz$ preserves the {\em program order} and {\em reads-from} relations induced by $\tr$. 
\end{remark}

Next, for ease of notation, we will often denote labels as $\ev{t, o(x)}$ (for example $\ev{t, \wt(x)}$) 
in place of the expanded version $\ev{t, o, x}$.
We will also, at times, omit the thread identifier and use the shorthand $e = o(x)$ to denote that
$\OpOf{e} = o$ and $\VariableOf{e} = x$.



\subsection{Reads-From Equivalence}

A natural notion of equivalence of program runs in the context
of shared memory multi-threaded program is \emph{reads-from}
equivalence. We formalize this notion here. 

\myparagraph{Program Order and Reads-from mapping}{
The program order, or thread order induced by a concurrent program run $w \in \alphabet^*$
orders any two events belonging to the same thread.
Formally, $\po{w} = \setpred{(e, f)}{e, f \in \events{w}, \ThreadOf{e} = \ThreadOf{f}, e \text{ occurs before } f \text{ in } w}$.
The reads-from mapping induced by $w$
 maps each read event to the write event it observes.
In our context, this corresponds to the last conflicting write event before the read event.
Formally, the reads-from mapping is a partial function 
$\rdf{w} : \events{w} \partialto \events{w}$
such that $\rdf{w}(e)$ is defined iff $\OpOf{e} = \rd$.
Further, for a read event $e$ occurring at the $i^\text{th}$ index in $w$,
$\rdf{w}(e)$,
is the unique event $f$ occurring at index $j < i$ for which $\OpOf{f} = \wt$,
$\VariableOf{f} = \VariableOf{e}$ and there is no 
other write event on $\VariableOf{e}$ (occurring at index $j'$)
such that $j < j' < i$.
Here, and for the rest of the paper, we assume that every read event is preceded by some write event
on the same variable. 
}

\myparagraph{Reads-from equivalence}{
Reads-from equivalence is a semantic notion of equivalence
on concurrent program runs, that distinguishes between
two runs based on whether or not they might produce different outcomes.
We say two runs $w, w' \in \alphabet^*$ are
\emph{reads-from equivalent}, denoted $w \rfnovaleq w'$
if $\events{w} = \events{w'}$, $\po{w} = \po{w'}$ and
$\rdf{w} = \rdf{w'}$.
That is, for $w$ and $w'$ to be reads-from equivalent,
they should be permutations of each other and must follow the same
program order, and further, every read event $e$ must read from
the same write event in both $w$ and $w'$.
Reads-from equivalence is a strictly coarser equivalence than
trace equivalence for concurrent program runs.
That is, whenever $w \mazeq{} w'$, we must have $w \rfnovaleq w'$;
but the converse is not true.
}

\begin{example}
Consider the two runs (denoted $\sigma$ and $\sigma'$) shown in \figref{intro}(a) and 
\figref{intro}(b) respectively.
Observe that $\sigma$ and $\sigma'$ have the same set of events,
and program order ($\po{\sigma} = \po{\sigma'}$).
Also, for each read event $e$, $\rdf{\sigma}(e) = \rdf{\sigma'}(e)$.
This means $\sigma \rfnovaleq \sigma'$.
Consider now the permutation of $\sigma$ corresponding to the sequence
$\sigma'' = e_1\ldots e_8e_{10}e_9e_{11}\ldots e_{14}$, where $e_i$ denotes the $i^\text{th}$ event of $\sigma$ from the top.
$\sigma''$ does not have the same reads-from mapping as $\sigma$
since $\rdf{\sigma''}(e_9) = e_{10} \neq e_7 = \rdf{\sigma}(e_9)$,
and thus $\sigma'' \not\rfnovaleq \sigma$.
\end{example}

\myparagraph{Soundness of Equivalence}{
The focus of this work is to develop equivalences that are coarser than Mazurkiewicz equivalence
and are also \emph{sound},
where soundness will be with respect to reads-from equivalence.
Given a run $\tr \in \alphabet^*$ and a set $S_\tr \subseteq \alphabet^*$,
we say that $S_\tr$ is sound for $\tr$ if $S_\tr \subseteq \setpred{\tr'}{\tr \rfnovaleq \tr}$.
Likewise, an equivalence relation $\sim$ over $\alphabet^*$ is said to be sound
if for every $\tr\in \alphabet^*$, the equivalence class $\eqcl{\tr}{\sim}$
is sound for $\tr$. Hence, Remark \ref{rem:1} precisely states that trace equivalence defined based on the independence relation $\indrel_\maz$ is sound in this sense. 
}


\section{Causal Concurrency under Reads-From Equivalence}
\seclabel{semantic-equivalence}

In scenarios like dynamic partial order reduction in stateless model checking,
or runtime predictive monitoring, one is often interested in
the \emph{causal relationship} between actions.
Understanding causality at the level of program runs often amounts to
answering whether there is an equivalent run that witnesses
the inversion of order between two  events.

The efficiency of determining causal concurrency 
is then key in designing efficient techniques in the aforementioned contexts. 
When deploying such techniques for monitoring 
large scale software
artifacts exhibiting executions with billions of events,
a desirable goal is to design monitoring algorithms that 
can be efficiently implemented in an online `incremental' fashion
that store only a constant amount of information, independent of 
the size of the execution being monitored~\cite{RV03RTA}.
In other words, an ideal algorithm would observe events in an execution
in a single forward pass, using a small amount of memory.
This motivates our investigation of the efficiency of checking causal ordering.

We first formally define the relevant algorithmic questions in the context of any equivalence relation on concurrent program runs, and then study their complexity under reads-from equivalence.

\subsection{Causal Concurrency and Ordering}
\seclabel{check-order-semantic}

Formally, let $\sim$ be an equivalence over $\alphabet^*$ such that when two runs
are equivalent under $\sim$, they are also permutations of each other.
Let $w \in \alphabet^*$ be a concurrent program run, and let $e, f \in \events{w}$ 
be two events occurring at indices $i$ and $j$ (with $i < j$).
We say that $e$ and $f$ in $\tr$ are \emph{causally ordered} under $\sim$  
if for every $w' \sim w$, $e$ and $f$ appear in the same order as they do in $w$.
If $e$ and $f$ are not causally ordered under an equivalence $\sim$,
we say that they are causally concurrent under $\sim$.

The following are the key algorithmic questions we investigate in this paper.
\begin{problem}[Checking Causal Concurrency Between Events]
\problabel{conc-events}
Let $\sim$ be an equivalence relation over $\alphabet^*$.
Given a program run $w \in \alphabet^*$, and two events
$e, f \in \events{w}$, the problem of checking causal concurrency between events 
asks if $e$ and $f$ are causally concurrent under $\sim$. 
\end{problem}

In the context of many applications in testing and verification
of concurrent programs, one often asks the following more coarse grained question.
\begin{problem}[Checking Causal Concurrency Between Symbols]
\problabel{cc-symb}
Let $\sim$ be an equivalence relation over $\alphabet^*$.
Given a program run $w \in \alphabet^*$, and two symbols
(or letters) $c, d \in \alphabet$, 
the problem of checking causal concurrency between symbols (or letters)  
asks to determine if 
there are events $e, f \in \events{w}$ such that $w[e] = c$, $w[f] = d$
and $e$ and $f$ are causally concurrent under $\sim$.
\end{problem}

If one has an oracle for deciding causal concurrency between symbols, then one can use it to check causal concurrency between events. 
In particular, assume  $c = w[e]$ and $d = w[f]$ and 
 consider the alphabet $\Delta = \alphabet \uplus \set{c^{\focalEv_1}, d^{\focalEv_2}}$,
where $c^{\focalEv_1}$ and $d^{\focalEv_2}$ are distinct marked copies of the letters $c, d \in \alphabet$.
Consider the string $w' \in \Delta^*$ with $|w'| = |w|$
such that $w'[g] = w[g]$ for every $g \not\in \set{e, f}$,
$w[e] = c^{\focalEv_1}$ and $w[f] = d^{\focalEv_2}$.
Then, causal concurrency (respectively orderedness) of symbols $c^{\focalEv_1}$ and ${d^{\focalEv_2}}$ under $\sim'$ implies the causal concurrency (respectively orderedness) of events $e$  and $f$ under $\sim$, where
the equivalence $\sim'$ is the same as $\sim$, modulo renaming letters $c^{\focalEv_1}$ and $d^{\focalEv_2}$
to $c$ and $d$ respectively.


\subsection{Computational Hardness in Checking Concurrency}

For the case of trace equivalence, more specifically under $\equiv_\maz$, 
causal concurrency can be determined in a constant space
streaming fashion. This result is somewhat known amongst the experts in the field, but it does not appear in the following specific way anywhere in the literature. 

\begin{proposition}[Causal concurrency for trace equivalence]
Given an input $w \in \alphabet^*$ and symbols $c, d \in \alphabet$,
the causal concurrency between $c$ and $d$ under
$\mazeq{}$ can be determined
using a single pass constant space
streaming algorithm.
\end{proposition}

In \secref{gmonitor}, we give a constructive proof to this lemma by presenting a constant space monitoring algorithm. Next, we show that the same is not achievable for the case of reads-from equivalence --- we show that any algorithm
for checking causal ordering (for the semantic notion of equivalence $\rfnovaleq$)
must use linear space in a streaming setting.
\begin{theorem}[Linear space hardness]
\thmlabel{check-order-semantic-linear-space-lower-bound}
Any streaming algorithm that checks the causal concurrency of a pair of symbols under $\rfnovaleq$
in a streaming fashion uses linear space, even for program runs containing
just $2$ threads and $6$ variables.
\end{theorem}

The key idea behind the proof of \thmref{check-order-semantic-linear-space-lower-bound}
is to exploit and generalize the intricate pattern in the runs in \figref{intro}.
The idea is that determining if two specific events are causally concurret
in such a pattern, relies on successively inferring the concurrency status
of linearly many pairs of events, placed arbitrarily far away in the past and/or in the future,
which is impossible for a one pass streaming algorithm
that only uses sub-linear space.
The formal proof
is presented in~\appref{hardness}.

In fact, we use similar ideas as in the proof of the above statement
to also establish a lower bound on the time-space tradeoff for
the problem of determining causal concurrency, even when we do not bound the number of passes;
formal proof is deferred to \appref{hardness}:
\begin{theorem}[Quadratic time space lower bound]
\thmlabel{check-order-semantic-tradeoff}
For any algorithm (streaming or not) that checks if a pair of symbols are causally concurrent under $\rfnovaleq$ 
in $S(n)$ space and $T(n)$ time
on an input run of length $n$, we must have $S(n) \cdot T(n) \in \Omega(n^2)$.
\end{theorem}


The linear lower bound in \thmref{check-order-semantic-linear-space-lower-bound} establishes non-regularity of any monitor for causal concurrency of symbols. We can further refine this result and show that the problem of determining causal concurrency of symbols is also not context-free. 
We establish the following result by invoking a pumping lemma argument for context free languages
on sets of runs that observe intricate patterns akin to the one in~\figref{intro}.

\begin{theorem}[Non Context-freeness]
\thmlabel{check-order-non-context-free}
Let $c, d \in \alphabet$. There exists no nondeterministic pushdown automaton that accepts exactly the
runs $\tr$ such that both $c$  and $d$ appear in $\tr$ and are causally concurrent in it under $\rfnovaleq$.
\end{theorem}


Finally, conditioned on 
the widely believed Strong Exponential Time Hypothesis (SETH)~\cite{ImpagliazzoP01}, we establish a quadratic time lower bound for Problems \ref{prob:conc-events} and \ref{prob:cc-symb}.

\begin{theorem}
\thmlabel{quadratic-time-hardness}
Assume SETH. The problem of determining the causal concurrency of a pair of events in $\tr$ under $\rfnovaleq$ cannot be solved in time $O({|\tr|}^{2-\epsilon})$ for all $\epsilon>0$, even when $\tr$ has $2$ threads.
\end{theorem}

The formal proof of \thmref{quadratic-time-hardness} is presented in \appref{ov-hardness}. 
It is established via a fine-grained reduction from the Orthogonal Vector Conjecture,
which also admits a quadratic lower bound under SETH~\cite{Williams05}.
The input to the Orthogonal Vectors (OV) problem is a pair of sequences $A, B \subseteq \set{0,1}^d$ of $d$-dimensional vectors, each of length $n$ (i.e., $|A| = |B| = n$), and the output is YES iff
there are two vectors $a \in A, b \in B$ such that their inner product is $0$ (i.e., $a \cdot b = 0$). 
The OV hypothesis states that for every $\epsilon > 0$, 
there is no algorithm that solves the OV problem 
(with input instances of length $n$ and dimension $d$, with $d \in \Omega(\log n)$) in $O(n^{2-\epsilon}\poly(d))$ time.

\subsection{Upper Bounds for Checking Concurrency}

In this section, we study the precise complexity of reasoning about concurrency under $\rfnovaleq$
and establish time and space upper bounds.
First, we observe that both Problems \ref{prob:conc-events} and \ref{prob:cc-symb} can be solved using an algorithm whose running time is a polynomial expression whose degree varies with the number of threads:

\begin{theorem}[Polynomial time algorithm]
\thmlabel{check-order-rf-time-upper-bound}
Let $w \in \alphabet$ be a run with $|w| = n$ and let 
$e$ and $f$ be two events in $w$.
The problem of determining if $e$ and $f$ are causally concurrent under $\rfnovaleq$
can be solved in time $O(|\threads|\cdot n^{|\threads|+1})$.
Similarly, given $c, d \in \alphabet$, the 
problem of determining if $c$ and $d$ are causally concurrent under $\rfnovaleq$ can be solved in time
$O(|\threads|\cdot n^{|\threads|+2})$.
\end{theorem}

The proof (see \appref{poly-algo}) is based on constructing a 
`frontier graph'~\cite{Gibbons1994}, whose vertices represent subsets of $\events{w}$
which are downward closed with respect to $\po{w}$ and $\rdf{w}$,
while edges represent valid extensions obtained by adding single events to the subsets.

The problem can also be solved using a linearly bounded Turing machine,
which also implies that
the language of runs that exhibit concurrency of two given events
is a context sensitive language.

\begin{theorem}[Linear Space Upper Bound]
\thmlabel{check-order-rf-space-upper-bound}
Let $w \in \alphabet$ be a run with $|w| = n$ and let 
$e$ and $f$ be two events in $w$.
The problem of determining if $e$ and $f$ are causally concurrent under $\rfnovaleq$ 
can be solved in deterministic space $O(n)$.
Similarly, given $c, d \in \alphabet$, the 
problem of determining if $c$ and $d$ are causally concurrent under $\rfnovaleq$
can be solved in time deterministic space
$O(n)$.
\end{theorem}

We present the formal proof of \thmref{check-order-rf-space-upper-bound} in
 \appref{linear-space-upper-bound}. It operates by successively generating permutations
 of the given run and checking if they are equivalent to the input run
 and also invert the order of the given events, all in deterministic linear space.


\section{Grain Commutativity}\label{sec:geq}
\seclabel{syntactic-weakenings}

In this section, we present a new stronger and more syntactic definition of equivalence for concurrent runs that can overcome the hardness results previously discussed. We build on the theory of traces, where equivalence is defined based on a commutativity relation on the alphabet of program actions. The new equivalence relation is defined based on an extended  commutativity relation that additionally allows commutating some pairs of {\em words} over the same underlying alphabet, and strictly weakens trace equivalence. First, let us briefly discuss that unchecked generalization in this direction can very quickly result in hardness. 

\begin{theorem}\label{thm:wc-hardness}
Let $\alphabet = \{a,b,c\}$ and let $\indrel = \{(a, bc), (bc,a), (b,c), (c,b)\}$. 
There exists no constant space monitor that given an input word $w \in \alphabet^*$ can decide whether the first occurrence of $a$ and the last occurrence of $c$ in $w$ are ordered. 
\end{theorem}

The idea of the proof is to focus on words of the form $ab^nc^m$, and a causal concurrency query that 
involves the $a$ and the last occurrence of $c$. It can be argued that the two are causally concurrent if and only if $n \ge m$. A detailed proof is presented in \appref{geq}.

Note that here, $\indrel$ is a generalization of commutativity relation in trace equivalence ($\indrel_\shasha$) by what seems to be the smallest increment to a classic commutativity relation over letters: the commutativity of one word of length $2$ (the shortest possible word that is not a letter) against one single letter. Yet, we immediately loose the constant-space checkability of concurrency/orderedness of pairs of events. Therefore, to maintain the property of constant-space checkability, one has to be careful with the generalization of $\indrel_\shasha$.

\subsection{Partially-Commutative Grain Monoids}

We define an equivalence relation based on commutativity of words. 

\begin{definition}[Grains]\label{def:grain} 
A {\em grain} is simply a non-empty word in $\alphabet^*$.
\end{definition}

\def\letters{\mathit{letters}}

For a grain $g$, let $\letters(g)$ be the set of letters (from $\alphabet$) that appear in $g$. We can define a partially commutative monoid in the same style as {\em trace monoids} \cite{Mazurkiewicz87} as follows:

\begin{definition}[Grain  Monoids]\label{def:grain-monoid} 
Given a trace monoid $(\alphabet, \indrel)$, a set of grains $G$ induces the (partially commutative) {\em grain}  monoid $(\alphabet_G \cup \alphabet, \widehat{\indrel}_G)$ where $\alphabet_G = \{a_g|\ g \in G\} $ and 
\begin{align*}
\widehat{\indrel}_G &= \indrel \cup \indrel_G \cup \{(a, a_g),(a_g,a)|\ a \in \alphabet \land a_g \in \alphabet_G \land \{a\} \times \letters(g) \subseteq \indrel \}\\
    &  \cup \{(a_g, a_{g'}), (a_{g'}, a_g)|\ a_g, a_{g'} \in \alphabet_G \land \letters(g) \times \letters(g') \subseteq \indrel \}
\end{align*}
where $\indrel_G \subseteq \alphabet_G \times \alphabet_G$, the {\em grain commutativity} relation, is an arbitrary symmetric independence relation defined on the grains.  
\end{definition}

If $G = \alphabet$ and $\indrel_G = \emptyset$, then the induced grain monoid coincides with the trace monoid $(\alphabet,\indrel)$. Otherwise, it can be viewed as a classic trace monoid on a new alphabet $\alphabet_G \cup \alphabet$. For this reason, it induces an equivalence relation $\equiv_{\,\widehat{\indrel}_G}$ on the set of words in $(\alphabet_G \cup \alphabet)^*$. 


%
%

Recall that $\indrel_\shasha$ is defined in the context of the alphabet $\alphabet$ to be {\em sound}, 
precisely in the sense that the induced $\meq$ preserves $\rdf{}$-equivalence. 
We need to determine when $\equiv_G$ is considered sound.

\begin{definition}[Strict Soundness]\label{def:sscom}
We call a grain commutativity relation $\indrel_G$ {\em strictly sound} if, $(g,g') \in \indrel_G$ iff for all $\alpha, \beta \in \alphabet^*$, we have $\alpha g g' \beta \rfnovaleq \alpha g' g \beta$. 
\end{definition}

It is straightforward to see that if we let $G = \alphabet$, then the independence relation that defines trace equivalence is strictly sound according to this definition. The less straightforward fact is that trace equivalence defines the largest such sound relation. To be precise, 

\begin{proposition}
If $g$ and $g'$ strictly soundly commute, then $gg' \meq g'g$.
\end{proposition} 

On the one hand, strict soundness seems reasonable because it decouples commutativity from its {\em context}. On the other hand, it makes the grains seem pointless in the sense that they do not offer any additional commutativity compared to classical trace monoids. The motivation behind forcing the soundness to be independent of the {\em context} (i.e. the choice of $\alpha$ and $\beta$ in Definition~\ref{def:sscom}) is that as one swaps two commutating grains, the {\em context} may change, and it would complicate reasoning substantially if the commutativity status of two other grains were to change as a result. In \cite{gtrace}, a generalized version of trace monoids are formulated that account for context. These have a sophisticated set of coherence and consistency conditions and have not been studied in any algorithmic contexts beyond being defined.

Fundamentally, we would like a commutativity relation that is sound in the sense that it maintains $\rdf{}$-equivalence and defines an equivalence class in which the commutativity relation does not change from one member to another to keep things simple for formulating algorithms.

\parpic[r]{\includegraphics[scale=0.9]{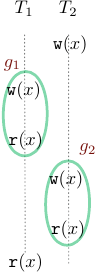}}
\begin{example}\label{ex:dreads}
Consider the program run on the right, with 2 threads. Two grains $g_1$ and $g_2$ have been marked. In isolation (as in if everything else from this run is ignored), these two grains soundly commute. But, the run illustrates that in the presence of the last $\rd(x)$, the two grains do not soundly commute, and therefore they do not strictly soundly commute. Observe that the left context is irrelevant to the commutativity status of these two grains. Only the right context can violate soundness. If unsoundness is the result of the program order being broken, one would observe it by looking only at the two grains.  Therefore, any violations to soundness related to the context have to be related to the reads-from relation. Specifically, a write event $\wt$ belongs to a grain, but there is at least one read event $\rd$, which {\em reads} from it (i.e., $\wt = \rdf{}(\rd)$), that does not belong to the same grain. By formally disallowing any such bad right contexts for a pair of grains, one can define a more permissive commutativity relation.
\qed
\end{example}

Example \ref{ex:dreads} illustrates why we put the fault in the definition of strict soundness mainly with the {\em right} context. To rule these scenarios out, one can restrict the right context from all possible contexts to those that cannot adversarially affect the commutativity status of $g$ and $g'$.

\begin{definition}[Sound Grain Commutativity]\label{def:scom}
We call a grain commutativity relation $\indrel_G$ {\em sound} if for all $g,g' \in G$, for all $x \in \VariableOf{g} \cap \VariableOf{g'}$ where at least one $\wt(x)$ appears in $gg'$,  and for all $\alpha,\beta \in \alphabet^*$ such that  $\beta|_x \in L(\wt(x)\alphabet^*)$, we have:
$(g,g') \in \indrel_G  \iff \alpha g g' \beta \rfnovaleq \alpha g' g \beta$
\end{definition}
Definition \ref{def:scom} strictly weakens Definition \ref{def:sscom} by limiting the (right) contexts in which commuting the actions must be sound. In other words, every strictly sound grain commutativity relation is sound, but not all sound grain commutativity relations are strictly sound. In particular, if $g$ and $g'$ do not strictly soundly commute, it is fairly straightforward to construct a right context $\beta$ in which they do not soundly commute.


%

\begin{remark}
Given a run $w$ and another run $u$ that can be acquired from $w$ through a sequence of swaps defined by a {\em strictly} sound commutativity relation $\widehat{\indrel}_G$ (Definition \ref{def:sscom}), we have $w \rfnovaleq v$. This is not true for a sound commutativity relation (Definition \ref{def:scom}).
\end{remark}

\parpic[r]{\includegraphics[scale=0.85]{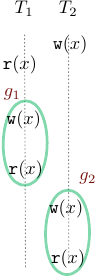}}
\begin{example}\label{ex:sound}
Recall the run illustrated in Example \ref{ex:dreads}. According to Definition \ref{def:scom} and as discussed in Example \ref{ex:dreads}, the two grains $g_1$ and $g_2$ soundly commute. But, clearly, the equivalence class of the run from Example \ref{ex:dreads} up to this commutativity relation is not sound. In contrast, the same two grains appear in the run illustrated on the right and the equivalence runs inferred by their commutativity are sound. The difference is that the run on the right does not violate the condition about right contexts (from Definition \ref{def:scom}) but the run from Example \ref{ex:dreads} does. 
\qed
\end{example}

It feels like we took one step forward by weakening the definition of soundness and then one step backward since it is not guaranteed to provide soundness in all contexts. 
There are, however, two key observations that make this definition of soundness a good fit in the context of our main goal, that is checking the status of concurrency of two given events.
First, we are solely interested in the set of words equivalent to a single reference program run.
Second, we may not have control over the choice of right contexts, but we do have control over the choice of grains and the commutativity relation $\indrel_G$. We can limit these choices based on the input run so that the equivalence class of the input run induced by the corresponding grain monoid is indeed sound.
Therefore, we next focus on the equivalence class $[w]_G$ of a given word $w$ and the choices for $G$ (the set of grains) and the grain commutativity relation $\indrel_G$ that make $[w]_G$ sound. 


\subsection{Sound Grain Equivalence}

First, observe that the same exact grain (which is a word) can appear several times as a subword in a particular word. In our formalism so far, we had no need of distinguishing the multiple instances, but since their right contexts may differ, we must do so now to attach different commutativity attributes to them.

\begin{definition}
A {\em valid} set of grains for a run $\tr \in \alphabet^*$ is set of indexed words of the form $g@i$ where $g$ is a contiguous subword of $\tr$ that appears at position $i$ and no two grains overlap in $\tr$. 
\end{definition}

Consider a valid set of (indexed) grains $\grains$ in a program run $\tr$. 
Since the indexed set $\grains$ may only be a valid set for the run $\tr$ in consideration,
we cannot cleanly define the equivalence induced by $\grains$ on the set of all runs $\alphabet^*$.
However, we can still precisely define the class of runs that can be inferred by successively
swapping adjacent grains from $\grains$.
For ease of presentation, we will abuse the notation $\eqcl{\tr}{\grains}$ to denote this
set, and will take the liberty to call it the \emph{equivalence class of $\tr$}, when $\grains$
is known from the context.

We can formalize $\eqcl{\tr}{\grains}$ as follows.
Define $h_G: \alphabet^* \to (\alphabet_G \cup \alphabet)^*$ as the homomorphism that maps each indexed grain $g@i$ to the corresponding letter $a_g \in \alphabet_G$ and each letter in $\alphabet$, that does not belong to a grain, to itself. Let $h_G^{-1}$ be the inverse homomorphism that replaces letters in $\alphabet_G$ with the corresponding grain words. Then,
\[u \in [w]_G \iff \exists u' \in (\alphabet_G \cup \alphabet)^*:\ \left( u = h^{-1}_G(u') \wedge u' \equiv_{\widehat{\indrel}_G} h_G(w) \right)\]
In short, the set of words that are considered equivalent to $w$ are determined by those that are equivalent to its corresponding word in the grain monoid, for the specific choice of valid grains $G$. 
In prose, we say $u$ is equivalent to $w$ when $u \in [w]_G$ for some valid choice of grains $G$.


We lift the commutativity relation $\indrel_G$ to relate indexed words of the form $g@i$ as well. This will enable us to say that $(g@i, g'@j) \in \indrel_G$ while $(g@k, g'@j) \not \in \indrel_G$.  Note that corresponding grain monoid is defined as before, each new grain $g@i$ is mapped to a designated letter $a_{g@i}$. 
\begin{definition}\label{def:se}
For a word $w$ and a set of valid grains $G$ in $w$,  $\indrel_G$ is sound if $[w]_G$ is sound (i.e. $[w]_G \subseteq \rfcl{\tr}$).  
\end{definition}

Next we give necessary and sufficient conditions for soundness of $\indrel_G$. For an event $e = \wt(x)$, let $\mathit{reads}(e)$ denote all events $e'$ of the form $\rd(x)$ that read from $e$. 
To lighten the notation whenever possible, we may refer to a grain only by its word $g$ whenever the position is not of importance, or implied from the context. We only specifically mention the position $g@i$ when it really matters.
Define $\OpOf{g,x}$, for a grain $g$ to be the set of operations ($\{\rd,\wt\}$) of variable $x$ that appear in $g$. 

\begin{theorem}\label{thm:sgrain}
In the context of a word $w$ and a {\em valid} set of grains $G$,  a commutativity relation $\indrel_G$ is sound iff for all $(g,g') \in \indrel_G$, $gg' \rfnovaleq g'g$ and for all $x \in \VariableOf{g} \cap \VariableOf{g'}$ the following holds:
\[  
\big(\OpOf{g, x} \cup \OpOf{g', x} = \{\rd\}\big) 
\lor 
\big( \forall e, f \cdot (e = \wt(x), f = \rd(x), f = \rdf{\tr}(e)) \implies (e \in g {\iff} f \in g )\big)
\]
\end{theorem}
It is clear that if $gg' \rfnovaleq g'g$ does not hold, the resulting commutativity relation is unsound. Example \ref{ex:dreads} captures the idea of why the violation of the additional conditions also leads to unsoundness. Intuitively, these conditions express that if the grains share a variable and at least one writes to that shared variable, then once a write action is included in one grain then all reads that read from it also must be included in that grain. The proof of the other direction is a tedious case analysis and given in Appendix \ref{app:geq}

Recall the program runs in Examples \ref{ex:dreads} and \ref{ex:sound}. Observe that even though the pairs of grains are identical, if we assume the commutativity of the pair of grains, then the choice in Example \ref{ex:sound} satisfies the conditions of the above theorem, and the one in Example \ref{ex:dreads} does not; in particular, they violate the part that demands every read that reads from the same $\wt(x)$ operation must belong to the grain.

\begin{figure}[t]
\begin{center}
\includegraphics[scale=0.9]{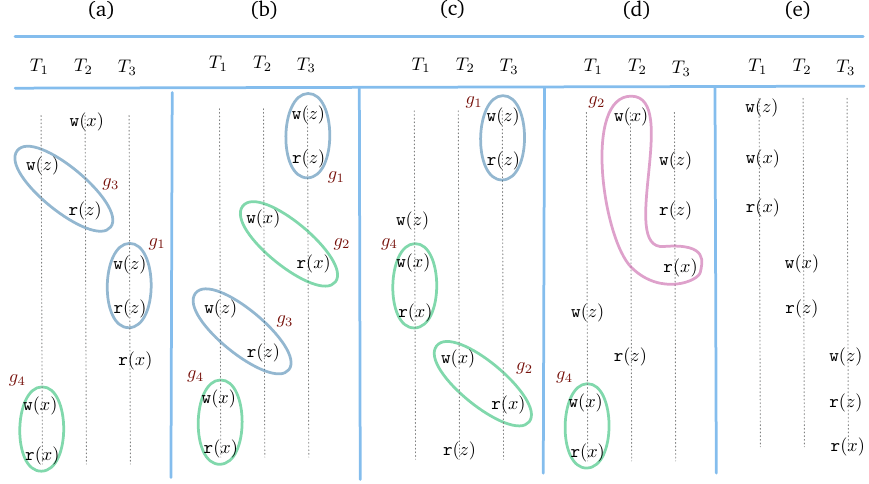}
\caption{Examples \ref{ex:converse}, \ref{ex:scattered}, and \ref{ex:limit}.}
\label{fig:n}
\end{center}
\end{figure}

\subsection{The Expressive Power of $[w]_G$}

We use an extended example to highlight how $[w]_G$ soundly enlarges the trace equivalence class of $w$,  induced by the trace equivalence relation $\meq$.  

\begin{example}\label{ex:converse}
Consider the program runs illustrated in \figref{n}. They are all $\rdf{}$-equivalent. We observe that different choices of grains  witness the equivalence of different pairs of the run from the figure. Independent of which grain is present in which subfigure, the only sound commutativity between grains, in addition to classic trace theory commutativity, is $\indrel_G = \{(g_1,g_3),(g_3,g_1),(g_2,g_4),(g_4,g_2)\}$. 

First, focus on \figref{n}(a), and observe that $g_1$ and $g_3$ soundly commute (in the sense of Theorem \ref{thm:sgrain}). Yet $g_4$ does not  commute with anything that it would not otherwise under the classic trace monoid through the commutativity of its individual events. For example, taking the single $\wt(x)$ of thread $T_2$ or the single $\rd(x)$ of thread $T_3$ as grains, one would violate the conditions of Theorem \ref{thm:sgrain}) if one were to declare either event commutative against $g_4$. With the grains marked in \figref{n}(a), the run illustrated in \figref{n}(a) is equivalent to the one in \figref{n}(b) ---  starting from (a), we can swap $g_1$ and $g_3$ first, and then $g_1$ against the $\wt(x)$ of thread $T_2$ and $g_3$ against the $\rd(x)$ of thread $T_3$.
 
Now consider the set of grains in \figref{n}(b). It is a superset of grains marked in \figref{n}(a), and yet, in this configuration we have less freedom of movement.  
$\indrel_G$ is still sound, but since $(g_1,g_2)$, $(g_2,g_3)$, and $(g_3,g_4)$ do not commute, this run effectively belongs to an equivalence class of size 1. Specifically, we cannot conclude that it is equivalent to the run illustrated in \figref{n}(a). If we remove grain $g_2$, then the equivalence class gets larger, and includes the run illustrated in \figref{n}(a). Observe that, therefore: 
\begin{quote}
\em
Having more grains {\em does not necessarily} imply having a larger equivalence classes.
\end{quote}

Similarly, if we take the set of grains in \figref{n}(c), then the run in \figref{n}(c) is equivalent to the one in \figref{n}(b). But, note that there is no possible choice of grains in \figref{n}(c) that would make it equivalent to the run in \figref{n}(a), and vice versa. 
The two runs are clearly $\rdf{}$-equivalent. They are grain-equivalent to the run in \figref{n}(b). Yet, for no choice of grains they can be made grain-equivalent to each other.

\begin{quote}
\em
If $v \in [u]_G$ and $w \in [v]_{G'}$, there may not exist a sound $G''$ such that $w \in [u]_{G''}$. 
\end{quote}


\end{example}

Example \ref{ex:converse} illustrates that depending on the input run $w$, or the choice of events for a query of concurrency, a different choice of grains $G$ may be suitable to define the appropriate $[w]_G$. 

\begin{definition}[Grain Concurrency]\label{def:gc}
Consider a program run $w$ and two events $e$ and $e'$ that appear in $w$. We call the pair of events $e$ and $e'$ to be {\em grain concurrent} iff there exists a sound set of grains $G$ and a commutativity relation $\indrel_G$ in the context of $w$ such that, there exists $u \in [w]_G$ in which $e$ and $e'$ are reordered with respect to their order of appearance in $w$. 
\end{definition} 

Note that for any given program run $w$ and any choice of valid grains $G$ and a sound commutativity relation $\indrel_G$,  $[w]_G$ is always, by definition, a superset of the trace equivalence class of $w$. Moreover, if we let $G = \Sigma$, then $[w]_G$ coincides with the trace equivalence class of $w$, since the $\indrel_\maz$ is by definition sound. As such, any two events that are concurrent according to trace equivalent are also grain concurrent. 

In Section \ref{sec:gmonitor}, we formally argue that grain concurrency can be checked in constant space by giving a construction for a monitor that is strictly more expressive than a constant-space monitor based on $\meq$. For example, our monitor would declare that both the pair of $\wt(x)$ and the pair of $\wt(z)$ operations in the run illustrated in \figref{n}(b) can be soundly reordered, while they are strictly ordered according to $\meq$. 

\renewcommand{\summaries}{E}

\newcommand{\operations}{\mathsf{C}}

\section{Grain Concurrency Monitor}
\seclabel{gmonitor}

{\em Grain monoids} are closely related to trace monoids. Therefore, we begin by defining a monitor that in constant space checks whether two events in an input run are concurrent according to $\meq$, and then present the grain concurrency monitor as an extension of this monitor. 

To have a simple setup, we augment our alphabet $\Sigma$ with two new symbols $\diamond_1$ and $\diamond_2$ that are assumed to appear precisely once each in any input word, marking the two events that are meant to be checked for concurrency; these would be the events that immediately succeed $\diamond_1$ and $\diamond_2$. The regular expression $\Sigma^* \diamond_1 \Sigma^+ \diamond_2 \Sigma^+$  captures that the the two $\diamond$'s are properly placed in an input run. 
Therefore, in the description of our main monitor, we assume that the input run is already well-formed in this sense. 

The high level idea behind the monitor is simple. 
The monitor idles until it sees the first marker $\diamond_1$ and the event $e_1$ that it marks. Afterwards, it maintains a summary of the set of operations, reads or writes to specific variables by specific threads, seen so far that are {\em ordered} wrt to $e_1$. When the monitor comes across $\diamond_2$ and therefore identifies the second event $e_2$, it can use the summary to determine if $e_1$ and $e_2$ are causally concurrent or ordered. 

The key to constant-space implementability is that the monitor does not have to remember all individual events, but rather what variables and threads are involved. This is based on the observation that to determine if two events commute, it suffices to know what variables are being accessed, what the nature of the access is, and to what threads the two events belong. The summary can be maintained in the most compact manner if the list of operations $\opfont{op}(x)$ and the list of thread identifies are kept separately. But, we present the less compact version that maintains the summary as a set of events, since it is easier to generalize this version of the monitor to {\em grains}.

Formally, the monitor's state is pair $\stuple{d, \operations}$ where  $\operations \subseteq \alphabet$ 
is a set of labels. 
The first element of the pair $d$ is used to track whether the monitor has seen events $e_1$ and $e_2$ yet. It encodes the six distinct stages: $-$: before the first diamond, $\diamond_1$: right after the first diamond, $\diamond -$: after $e_1$ has been recorded, $\diamond_1 - \diamond_2$: right after the second diamond is seen, and $\mathit{true}$/$\mathit{false}$: depending on the monitors decision to accept/reject after reading $e_2$. \figref{mcm} lists the transitions of the monitor and the functions used. Since the result is folklore, we forgo giving a proof for the correctness of this monitor.

\begin{figure}[t]
\begin{center}
\framebox{
\begin{minipage}{0.4\textwidth}
\begin{tabular}{lcl}
State & Event & State Update \\ \hline
$\langle-, \operations \rangle$ & $e \in \alphabet$ & $\langle -, \operations \rangle$ \\
$\langle -, \operations \rangle$ & $\diamond_1$ & $\langle \diamond_1, \operations \rangle$ \\
$\langle \diamond_1, \operations \rangle$ & $e \in \alphabet$ & $\langle \diamond_1 -, \{e\} \rangle$ \\
$\langle \diamond_1 - , \operations \rangle$ & $e \in \alphabet$ & $\langle \diamond_1 - , \operations \odot e \rangle$ \\
$\langle\diamond_1 -, \operations \rangle$ & $\diamond_2$ & $\langle \diamond_1 - \diamond_2, \operations \rangle$ \\
$\langle \diamond_1 - \diamond_2, \operations \rangle$ & $e \in \Sigma$ & $\langle \mathit{Ord}(\operations,e), \operations \rangle$ \\
$\langle \mathit{false}, \operations \rangle$ & $e \in \Sigma$ &  $\langle \mathit{false}, \operations\rangle$\\
\end{tabular}
\end{minipage}
\begin{minipage}{0.4\textwidth}
\begin{align*}
\operations \odot e &= \left\{ \begin{array}{ll}
                                     \operations \cup \{e\}  & \text{if }  \exists e' \in \operations:\ (e,e') \in \deprel_\shasha \\
                                     \operations & \text{owise} 
                                    \end{array}\right.
\end{align*}
\begin{align*}
\mathit{Ord}(\operations, e) &\iff  \exists e' \in \operations:\ (e,e') \in \deprel_\shasha 
\end{align*}
\end{minipage}}\vspace{-5pt}
\caption{Trace Concurrency Monitor: The monitor starts in state $\langle -, \emptyset\rangle$ and accepts if in a state $\langle \mathit{false}, \operations\rangle$ (for any $\operations$) once the input is read. 
The operation $\odot$ updates $\operations$ based on a new event. }
\label{fig:mcm}\vspace{-10pt}
\end{center}
\end{figure}

\subsection{A Monitor for a Fixed Set of Grains $G$}
\seclabel{fixed-grains-monitor}
We introduce our monitor based on grains in two stages, for the simplicity of presentation. First, we assume a set of grains is pre-decided and pre-marked in an input word, and present the core idea behind monitoring causal concurrency in this setup.  Then, in Section \ref{sec:gm}, we build the final grain concurrency monitor as a generalization of this one. 

Assume that $\Sigma$ is further extended with a pair of symbols $\rhd$ and $\lhd$, which are used as delimiters to mark grain boundaries. Any letter that appears outside the range of these delimiters is treated as a standalone event. For example, the program run form \figref{n}(a) with the marked grains becomes:
\[\evt{\wt(x)}{T_2} \rhd  \evt{\wt(z)}{T_1} \evt{\rd(z)}{T_2} \lhd \rhd \evt{\wt(z)}{T_3}  \evt{\rd(z)}{T_3} \lhd \evt{\rd(x)}{T_3} \rhd \evt{\wt(x)}{T_1} \evt{\rd(x)}{T_1} \lhd\]

The two events of interest are marked with $\diamond$'s, as before. Except that if either event belongs to a grain, then the diamond must mark the entire grain. For example, if we want to determine whether the two $\wt(x)$ events are ordered in the program run above, the diamonds would mark the first one as before, and the second one behind the left grain delimiter like this:
\[
\diamond_1 \evt{\wt(x)}{T_2} \rhd  \evt{\wt(z)}{T_1} \evt{\rd(z)}{T_2} \lhd \rhd \evt{\wt(z)}{T_3} \evt{\rd(z)}{T_3} \lhd \evt{\rd(x)}{T_3} \diamond_2 \rhd \evt{\wt(x)}{T_1} \evt{\rd(x)}{T_1} \lhd \]
Note that the events in a grain always move together. Therefore, one cannot have a verdict that an event $e$ (e.g. the first $\wt(x)$ above) is concurrent with some arbitrary event $f$ of a grain (e.g. the second $\wt(x)$ above), while it is ordered wrt another event $f'$ of the same grain (e.g. the second $\rd(x)$ above). The following (revised) regular expression captures all the input runs in which grains and $\diamond$'s are marked properly, and therefore, we do not make it part of the design of the monitor:
\begin{align}
 \big((\rhd \alphabet^+ \lhd) + \alphabet\big)^* \diamond_1 \big((\rhd \alphabet^+ \lhd) + \alphabet\big)^+ \diamond_2 \big((\rhd \alphabet^+ \lhd) + \alphabet\big)^+ \tag{WF}  \label{eq:wf}
\end{align}

To generalize the monitor in \figref{mcm} to work with grains, we face two sources of complications. First, grains are arbitrary words in $\Sigma^+$, and even though $\Sigma$ is finite, $\Sigma^+$ is infinite in size. Second, there are (potentially) infinitely many sound commutativity relations that can be inferred over the unbounded set of potential grains.  

A simple observation helps overcome the first problem. Fundamentally, we are interested in grains because we are interested in the commutativity properties of these grains. \thmref{sgrain} captures what information is relevant to make a sound decision about the commutativity of two grains. 
According to \thmref{sgrain}, $g$ and $g'$ soundly commute if $gg' \rfnovaleq g'g$  {\bf and } for every variable $x$ accessed in $g$ (respectively $g'$), where $x$ is written in at least one of the two grains, the following predicate is true:
\begin{align}
\cmpvar{g, x} \stackrel{\text{\tiny def}}{=} \big( \forall e, f \cdot (e = \wt(x), f = \rd(x), e = \rdf{\tr}(f)) \implies (e \in g {\iff} f \in g )\big)
\label{eq:full}
\end{align}
In other words, two grains commute, if they commute according to $\rfnovaleq$ in isolation and if the share a variable that is written by at least one of them, then both grains must be {\em complete} wrt to that variable. 
We now introduce the {\em signature} of a grain as a pair of a set of letters that appear in the grain and a set of  variables wrt which the grain is complete:
\[
\forall w \in \alphabet^+:\ 
\mathit{grain}(w) \stackrel{\text{\tiny def}}{=} 
\stuple{ 
\letters(w), \setpred{x \in \vars}{\cmpvar{w, x}} 
} 
\]
which contains all the information required for deciding commutativity of two given grains. More importantly, observe that there are boundedly many different grain signatures, $2^{|\alphabet| + |\vars|}$ to be exact. 
Therefore, in the summary ($\operations$ in \figref{mcm}), rather than keep track of the grain as an arbitrary size word, we maintain a set of grain signatures, which is very much bounded. 

Let us turn our attention to the second problem. For any pair of grains in a word, one can look at their signatures and soundly decide whether they commute or not. Yet, there are unboundedly many such grains, and therefore, unboundedly many such commutativity relations to enumerate for the definition of {\em grain concurrency} (Definition \ref{def:gc}). 


Observe that commutativity and causal concurrency are monotonically related: the larger the grain commutativity relation, the larger the set of pairs of grain concurrent events. Therefore, rather than enumerate all possible commutativity relations, one can conservatively choose the largest one. In this largest relation, any two grains, that can soundly commute, are assumed to be commuting. This is determined based on their signatures alone, and therefore, the number of possible choices for the {\em largest} commutativity relation is bounded because the number of distinct signatures is bounded. 

\paragraph{\bfseries The Monitor}
One can conceptually think about the monitor combining the following two passes into a single pass through nondeterminism:
\begin{itemize}
\item Pass 1: From right to left, analyze the grains and replace each with a fresh letter (corresponding to its signature), and learn the bounded maximal commutativity relations for these new letters. Intuitively, this pass finds out which grains are complete wrt which variables, constructs their signature, and replaces the grains with a new letter that encodes the information from the signature. Constructing signatures is straightforward in a left-to-right or right-to-left pass, but it is more straightforward to see why completeness (condition \ref{eq:full}) can be checked and encoded for each grain in a right-to-left pass: a violation of this condition manifests as a pending read's matching write appearing in a grain.

\item Pass 2: In the style of
the trace concurrency monitor (\figref{mcm})), in a left to right pass, decide the causal concurrency of the two entities marked by $\diamond$'s based on the original letters and the new letters computed during pass 1. 
\end{itemize}

Classic ideas from automata theory provide the recipe to combine these passes, through nondeterminism, into a single-pass (left to right) constant space monitor that decides causal concurrency between any two events based on the largest sound grain commutativity relation:

\begin{theorem}\label{thm:monitor-sound}
For a fixed set of valid grains $G$ and a largest sound commutativity relation $\indrel_G$, the monitor sketched above accepts a program run $u \diamond_1 e\ v \diamond_2 e' w$ iff $e$ and $e'$ are reordered in some member of $[u \diamond_1 e\ v \diamond_2 e' w]_G$.   
\end{theorem}
 
The proof together with the details of the monitor are presented in Appendix \ref{sec:app-transitions-grains}.

\subsection{Grain Concurrency Monitor}\label{sec:gm}

We are now ready to construct the monitor that precisely captures Definition \ref{def:gc}. This monitor nondeterministically guesses the $\rhd$ and $\lhd$ symbols and therefore the grain boundaries, and for each guess runs the monitor in \figref{gcm}. It has to maintain a state to make sure that the grains are well-formed (non-overlapping) and non-empty. Therefore it effectively makes a guess and checks that its guess belongs to the language of the regular expression \ref{eq:wf}. Note that this guessing must account for the $\diamond$'s. The monitor makes a guess that an event of interest will be in a grain that it just nondeterministically opened, and therefore marks it with a diamond. Naturally, all wrong guesses are refuted later when the grain closes without seeing an event of interest.

\begin{theorem}\label{thm:monitor-sound2}
There exists a monitor that can decide, in constant space,  the (causal) grain concurrency (respectively orderedness) of two events in a given program run.
\end{theorem}

\renewcommand{\summaries}{\mathsf{SC}}
\section{Scattered Grains}\label{sec:sgrains}

So far, we have formally defined grains as subwords of a concurrent program run. Let us revisit our examples from \figref{n} to motivate expanding the definition to include grains that do not appear as contiguous subwords; we call these {\em scattered grains}. 

\begin{example}\label{ex:scattered}
To argue for the equivalence of the run illustrated in \figref{n}(d) to the one in \figref{n}(c), we need the $\wt(x)$ of $T_2$ to first commute as an individual event (from (d) to (b)), and then move as part of the grain that is marked in \figref{n}(c). If we are permitted to  consider the {\em scattered grain} $g_2$ as a grain in \figref{n}(d) (marked in pink), then we can argue that $\wt(x)$ of thread $T_2$ can be reordered against $\wt(x)$ of thread $T_1$ by (eventually) swapping the corresponding grains $g_2$ and $g_4$.

The grain monitor we present in Section \ref{sec:gmonitor} cannot keep track of these dual roles. Starting from the run illustrated in \figref{n}(d), it cannot see the potential of the grain including the $\wt(x)$ of $T_2$ and $\rd(x)$ of $T_3$ forming after a few sound swaps, and therefore cannot reason that $\wt(x)$ of $T_2$ can ultimately be soundly reordered against $\wt(x)$ of $T_1$. \qed
\end{example}


Formally, we say $\mathbf{i}$ is subsequence of the sequence of the range $[1 .. n]$ if it is strictly increasing, and all its elements belong to $[1..n]$. 
For convenience, we treat subsequences as sets of their elements (without order) when appropriate. 

\begin{definition}[Scattered Grains]
A {\em scattered grain} of program run $\tr$ is a subsequence of $w$. To distinguish identical scattered grains from each other, we denote them as $g@\mathbf{i}$, where $\mathbf{i}$ is a subsequence of $[1 .. |\tr|]$ and identifies the position of $g$. A set of scattered grains $\grains$ for a word $\tr \in \alphabet^*$
is said to be \emph{valid} if no two distinct grains overlap, that is, $\grain_1@\subsq{i}_1 \neq g_2@\subsq{i}_2 \in \grains \implies \subsq{i}_1 \cap \subsq{i}_2 = \emptyset$.
\end{definition}

Observe that scattered grains generalize the definition of grains when the subsequences happen to be contiguous. We may refer to a scattered grain simply as $g$ rather than $g@\mathbf{i}$ whenever the position is unimportant or clear from the context. Sound commutativity relations over scattered grains are defined identically to contiguous grains, therefore all definitions and theorems from Section \ref{sec:syntactic-weakenings} hold.
Moreover, we assume that single events form grains of size one and therefore every event belongs to some grain in what follows.

\begin{definition}[Grain Graph of a Run]
\deflabel{grain-graph}
Let $\grains$ be a valid set of scattered grains for a program run $\tr$ and $\indrel_\grains$ be the largest 
sound  commutativity relation over $\grains$ in the context of $\tr$. 
The grain graph $\GGraph{\tr, \grains} = (V, E)$ is a directed graph 
defined with the set of nodes $V = \grains$ and the set of edges 
\[E = V \times V - \setpred{(v_1, v_2)}{v_1 = v_2 \lor (v_1, v_2) \in \widehat{\indrel}_\grains \lor \tr|_{v_1,v_2} \meq v_2v_1}\]
where $\tr_{v_1,v2}$ is the projection of $\tr$ to the content of the grains $v_1$ and $v_2$.
\end{definition}
The first two sets of excluded edges correspond to the classic notions of anti-reflexivity and independence. The third condition above determines when there is an edge between two scattered grains that are {\em entangled}. We want a directed edge only if the second grain cannot be safely commuted to before the first grain.

Grain graphs can be used to define a notion of concurrency based on a set of scattered grains in the following sense:
\begin{definition}[Grain Graph Concurrency]\label{def:wgc}
Let $w$ be a run, $\grains$ be a valid set of scattered grains in $\tr$ and $\GGraph{\tr, \grains}$ be the corresponding grain graph. Let $e_1$ and $e_2$ be events in $\tr$
such that  $e_1$ appears before $e_2$ in $\tr$.
We say that the events $e_1$ and $e_2$ 
are {\em grain graph concurrent} under $\grains$ 
if there is no path in $\GGraph{\tr, \grains}$ 
from the node containing $e_1$ to the node containing $e_2$.   
\end{definition}

We call a valid set of scattered grains $\grains$ and the corresponding commutativity relation $\indrel_\grains$ {\em sound} in the context of a run $\tr$ if the same conditions listed in \defref{scom} hold. 

Observe that for contiguous grains, soundness of grain concurrency was baked into the definition that $[w]_G$ is sound. With scattered grains, this is no longer the case, and hence we need the following theorem:

\begin{theorem}(Soundness of Grain Graph  Concurrency)\label{thm:swgc}
Let $w$ be a run, $G$ be a valid set of scattered grains in $w$. If a pair of events $e_1$ and $e_2$ are grain graph concurrent under $\grains$, 
then they appear in a different order in some run $u$ such that $u \rfnovaleq \tr$. 
\end{theorem}

\begin{wrapfigure}{r}{0.2\textwidth}
\vspace{-10pt}
 \begin{center}
    \includegraphics[scale=0.9]{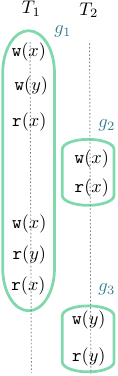}
  \end{center}
  \caption{\scriptsize Entangled Grains}
  \figlabel{rf}\vspace{-20pt}
 \end{wrapfigure} 
The proof (given in Appendix \ref{app:sgrains}) relies on the construction of the {\em condensation} of $\GGraph{\tr, \grains}$; that is, the directed {\em acyclic} graph acquired from $\GGraph{\tr, \grains}$ by contracting all its maximal strongly connected components. The proof argues that any linearization of this {\em condensed} graph is $\rdf{}$-equivalent to $\tr$. This, in turn, means that the condensed graph is analogous to a partial order representing a an equivalence class of $\mazeq{}$ induced by the trace commutativity relation $\indrel_\shasha$.  However, it differs from it in two important ways: (1) even though every linearization is $\rdf{}$-equivalent to $w$, it is not guaranteed to be equivalent to $w$ up to a sequence of valid grain and letter swaps, and (2) the set of linearizations does not necessarily include everything that is (grain) equivalent to $\tr$, even possibly $\tr$ itself. 

Consider the grains illustrated in \figref{rf}. They can be used to argue that the first $\wt(x)$ and the last $\wt(y)$ are grain graph concurrent. The grain graph only has one edge between grains $g_2$ and $g_3$, since all other pairs commute. However, observe that because $g_2$ is somewhat {\em entangled} with $g_1$, there exists no swap sequence to witness this concurrency. Moreover, the illustrated run itself does not belong to any linearization of the (condensed) grain graph, since no such linearization can reproduce the {\em entanglement} of the grains of the illustrated run. 

Even though contiguous grains are a special case of scattered grains, the way we define {\em concurrency} in the two cases are fundamentally different in Definitions \ref{def:gc} and \ref{def:wgc}. Yet, for a set of contiguous grains, {\em grain graph concurrency} coincides with {\em grain concurrency}. 

\begin{theorem}\label{thm:collapse}
For a program run $w$ and a sound valid set of contiguous grains $G$, a pair of events $e$ and $e'$ are {\em grain concurrent} under $G$ iff they are {\em grain graph concurrent} under $G$.  
\end{theorem}

This is the consequence of the fact that $\GGraph{\tr, \grains}$ is an acyclic graph for a valid set of {\em contiguous} grains $G$, and as such the condensed graph and $\GGraph{\tr, \grains}$ coincide, and are identical to the partial order describing the same equivalence class in the corresponding grain monoid, for which a valid swap sequence can be constructed. 

As with \defref{gc}, two events may be graph grain concurrent under one choice of scattered grains $\grains$
but not under another choice $\grains'$.
Consequently we define the following more permissive notion of concurrency under scattered grains.

\begin{definition}[Scattered Grain Concurrency]\deflabel{lgc}
Consider a program run $\tr$ and two events $e$ and $e'$ that appear in $\tr$. 
We call the pair of events $e$ and $e'$ to be {\em scattered grain concurrent} if there exists a valid set of scattered grains $\grains$ for $w$ such that
$e$ and $e'$ are grain graph concurrent under $\grains$.
\end{definition} 

The \thmref{collapse} also implies that {\em scattered grain concurrency} properly subsumes {\em grain concurrency}. In Section \ref{sec:beyond-contiguous-grains}, we demonstrate how {\em scattered grain concurrency} can be monitored in constant space. With the following example, we make the observation that even though {\em scattered grain concurrency} is strictly weaker than {\em grain concurrency}, it strictly under-approximates {\em sound concurrency} defined based on $\rdf{}$-equivalence. 

\begin{example}\label{ex:limit}
There remains a fundamental gap between the notion of concurrency defined based on $\rdf{}$-equivalence and scattered grain concurrency: in the run in \figref{n}(b), no choice of grains would witness the fact that the $\wt(z)$ operation of thread $T_3$ can be soundly reordered against the $\wt(x)$ operation of thread $T_1$. If all 4 grains are present, then the events are ordered. If we take either  $g_1$ or $g_3$ out, then they become ordered through the conflict dependencies between the $x$ operations. If we take either $g_2$ or $g_4$ out, then they become ordered through the conflict dependencies between the $z$ variables.   
Yet, the  $\rdf{}$-equivalent run in \figref{n}(e) witnesses that they are soundly concurrent. 

Interestingly, if we focus on the run in \figref{n}(e), and assume all grains $g_1$, $g_2$, $g_3$, and $g_4$ are present in it as scattered grains, then we can reason using the induced grain graph that the run in \figref{n}(b) is linearization of its (condensed) grain graph and as such  $\wt(x)$ of thread $T_1$ is (scattered) grain concurrent with $\wt(z)$ of $T_3$. Therefore, in the non-swap-based notion of {\em scattered grain concurrency}, one can reason about the implied equivalence and the corresponding notion of the runs in Figures \ref{fig:n}(b,e) in one way but the inverse.
\qed
\end{example}


\section{Monitoring with Scattered Grains}
\seclabel{beyond-contiguous-grains}

In this section, we develop a monitor for checking concurrency of two events
in the presence of scattered grains. 
When grains are scattered, they can be interleaved
in the run $\tr$, and this poses fundamental challenges towards the design
of a constant space monitor.

We call a grain {\em active} in a prefix of a concurrent run, if part of the grain has appeared in the prefix, but it is has not appeared in its entirety. With contiguous grains, at most one grain can be active at any given time. In sharp contrast, the number of scattered grains that
may be  \emph{active} simultaneously, can be unbounded (i.e., not constant).
Consider the run and the scattered grains $G = \set{g_1, \ldots, g_n}$ 
marked in \figref{challenges}(a).
Observe that all grains $g_1, g_2, \ldots g_n$
are active in the prefix $\pi$.
The unboundedness of the number of active grains 
is problematic because one expects that a monitoring algorithm for checking
concurrency would need to at least keep track of all active grains. This is the first challenge that has to be overcome in designing a constant space monitor for scattered grain concurrency.

\begin{figure}[t]
\begin{center}
\includegraphics[scale=0.9]{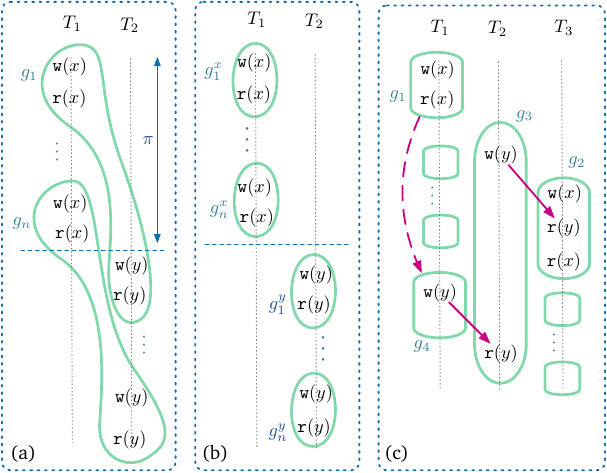}\vspace{-5pt}
\caption{The challenges of monitoring with scattered grains: (a) Unboundedly many active grains, (b) Minimal grains, and (c) Retroactive paths.}
\label{fig:challenges}\vspace{-10pt}
\end{center}
\end{figure}

The second challenge arises because a single active grain may overlap with many
other active grains through its lifetime, even if at a given point
only boundedly many grains are active. 
This means that two grains can be observed as ordered witnessed by a path 
in the grain graph (\defref{grain-graph}),
but this path is completed by a grain that appears long after the lifespan of both the grains have ended. 
Consider, for example, the run in \figref{challenges}(c).
Before the grain $g_4$ appears, there is a path, namely the direct edge, from $g_3$ to $g_2$, but no path from $g_1$ to $g_2$.
In fact, right before $g_4$ appears, $g_1$ and $g_2$ have both become inactive. 
When $g_4$ appears, a path is formed from $g_1$ to $g_4$. Once the $\rd(y)$ event in $g_3$ appears, an edge is formed between grain $g_4$ and grain $g_3$, which now completes the path from $g_1$ to $g_2$ \emph{retroactively}. Accounting for such retroactive paths is necessary for soundness.

We present solutions to these two challenges in Sections \ref{sec:bounding-active} and \ref{sec:retroactive-paths}. We introduce the notion of a \emph{minimal} grain to deal with the problem of unboundedly meany active grains and then demonstrate how a monitor can track retroactive paths
by \emph{summarizing} the paths between active grains when an intermediate grain has finished.


\subsection{Bounding the number of Active Scattered Grains}
\seclabel{bounding-active}


We start by defining a class of grains that are
\emph{minimal}, with the idea that if all (scattered) grains are minimal, then the number of active grains in any prefix of the program run is bounded. 

\begin{definition}[Minimal Grain]
\deflabel{minimal-grain}
A (scattered) grain $\grain$ is said to be minimal 
in program run $\tr$ if for all $\rho, \sigma \neq \epsilon$ such that $\grain = \rho \sigma$, 
there is a read event $\rd \in \sigma$ where 
$\rdf{\tr}(\rd) \in \rho$.
We call a set of grains minimal in $\tr$
if all grains in it are minimal in $\tr$.
\end{definition}


Intuitively, a grain is minimal if it cannot be broken into simpler grains,
without worsening its commutativity status with respect to other grains.
In~\figref{challenges}(a), none of the grains $g_1, \ldots, g_n$
are minimal, because each of them have a prefix (of size $2$) where there
is no read event whose corresponding write event is not in the prefix grain. In contrast, the following regular expression can produce arbitrarily long minimal grains:
\[\wt(x) \wt(y) \rd(x) \Big(\wt(x) \rd(y) \wt(y) \rd(x)\Big)^*\]
For a valid set $\grains$ of grains in $\tr$ and a prefix $\pi$ of $\tr$,
we use the notation $\actv{\pi, \grains}$ to denote the set of grains in $\grains$
that are active in $\pi$. If all grains are minimal, then the set of active grains is bounded in size by the total number of variables:

\begin{lemma}
\lemlabel{minimal-grains-active-bound}
Let $\tr$ be a run and let $\grains$ be a set of minimal grains in $\tr$.
For a prefix $\pi$ of $\tr$, we have $|\actv{\pi, \grains}| \leq |\vars|$.
\end{lemma}

This is implied by the following crucial observation. Assume two minimal grains $g$ and $g'$ are active at given prefix $\pi$ of $\tr$. The minimality $g$ (respectively $g'$) implies that there is a variable $x$ (respectively $x'$) that is written in $\pi$ and has a corresponding read in the remainder of $\tr$. The variables $x$ and $x'$ cannot coincide. Therefore one needs a distinct variable that only belongs to one of the many active grains at any given point, which puts a bound of $|\vars|$ on the maximum number of grains that can be active at any given time. 


Ideally, we want all grains to be minimal, since this resolves the problem of having unboundedly many active grains. Let us argue why this can be achieved without any compromises to the result of checking causal concurrency. 
Consider \figref{challenges}(b) which depicts the same run as in \figref{challenges}(a)
but this time with a different set of grains which are all minimal. Observe that the set of minimal grains in \figref{challenges}(b) witness more causal concurrency than the set of non-minimal grains in \figref{challenges}(a). 
In general, one can argue that for any set of (non-minimal) grains that witnesses the causal concurrency of two events $e$ and $f$, there exists a set of minimal events that does the same. We give a constructive argument for this claim. 

One can argue that any grain $g$ is the concatenation of a sequence of minimal grains. Recall \defref{minimal-grain}. For any $\rho$ and $\sigma$ that witness non-minimality of $g = \rho \sigma$ according to \defref{minimal-grain}, add a split point between $\rho$ and $\sigma$. Let $g = g_1 \dots g_n$ where $g_i$'s are precisely marked by these split points. By definition, $g_i$'s are all minimal. Define $\splitGrains(\grains) = \{g_1, \dots, g_n\}$.

Given a valid set of grains $\grains$,
we use $\splitGrains(\grains) = \bigcup_{g \in \grains} \splitGrains(g)$
to denote the set of grains obtained by splitting individual grains in $\grains$.
We need to argue that splitting all grains into minimal grains would result in declaring at least as many pairs of events causally concurrent as before.

Given a commutativity relation $\indrel_\grains$ on 
$\grains$, define $\splitGrains(\indrel_\grains) \subseteq \splitGrains(G)  \times \splitGrains(G)$ as 
\[\splitGrains(\indrel_\grains) = \setpred{(\grain'_1, \grain'_2)}{\exists \grain_1, \grain_2 \in \grains, 
(g_1, g_2) \in \indrel_\grains \text{ and } \grain'_1 \in \splitGrains(\grain_1), \grain'_2 \in \splitGrains(\grain_2)}\]
One can prove that $\splitGrains(\indrel_\grains)$ is sound. This in turn implies that splitting grains
does not add spurious paths in the new grain graph, which did not exist in the original one. 

\begin{lemma}
\lemlabel{minimal-grains-sufficient}
Let $\tr$ be a run, $\grains$ be a valid set of scattered grains,  
and  $\grainind$ be a sound independence relation (\defref{se}). $\splitGrains(\indrel_\grains)$ is sound, and 
for every pair of events $(e_1, e_2)$, if $e_1$ and $e_2$ are grain graph concurrent
under $\grains$ (using commutativity relation $\indrel_\grains$), 
then they are grain graph concurrent under the grains $\splitGrains(\grains)$
(using commutativity relation $\splitGrains(\indrel_\grains))$.
\end{lemma}

Therefore, when checking scattered grain concurrency, it is safe to ignore all sets of grains that include any non-minimal grains, since the same causal concurrency verdicts can be declared by other sets of minimal grains.  


\subsection{Tracking Retroactive Paths}
\seclabel{retroactive-paths}

Let us address our second challenge, that is how to keep track of {\em retroactive} paths in the grain graph.
We fix the set of minimal grains. 
Since grains containing only a single event are by definition minimal, we can assume, without loss of generality, that every event is part of a grain; standalone events are grains of size one.
Like Section \ref{sec:gmonitor}, the causal concurrency question between events
$e_1$ and $e_2$ is posed as the causal concurrency between the grains $\focalGrain{1}$ and $\focalGrain{2}$ containing these events. 

It is clear that to design a constant space monitor, we cannot store the entire grain graph in the memory of the monitor. The idea is to forget all grains that are no longer active and {\em summarize} their effect instead. 
Under the assumption that all grains are minimal, the number of active grains are bounded by $|\vars|$.
This guarantees that the monitor keeps track of constantly many grains. 
We can annotate events with a finite set of grain identifiers $\set{1, 2, \ldots, |\vars|}$. Identifiers are reused for grains that do not overlap.

The monitor maintains a graph where the nodes are precisely the set of active grains. We call this graph the {\em summarized grain graph}. Recall that the key information a monitor wants from a grain graph is whether two grains are connected by a directed path in the grain graph. 
Paths from the grain graph are represented as edges in the summarized grain graph. 
In particular, the edges in the summarized grain graph capture
paths in the original grain graph whose intermediate grains have become inactive.

More formally, let $\pi$ be a prefix of run $\tr$ and let $\grains$
be a valid set of grains.
For grains (active or otherwise) $\grain, \grain'$ 
that overlap with $\pi$, 
use the notation $\grain \deadpath{\pi, \grains} \grain'$ to say that
there is a path \emph{through} inactive grains from $\grain$ to $\grain'$,
i.e., there are grains $\grain_1, \grain_2 \ldots, \grain_{k} \in \grains$ 
(with $k>1$, $g = g_1$ and $g' = g_{k}$), such that
$g_2, \ldots, g_{k-1} \subseteq \pi$ are inactive (i.e., have started and completed) in $\pi$,
and $(g_i, g_{i+1})$ is an edge of the original grain graph, for each $1 \leq i \leq k-1$.
In essence, a path in the grain graph can be split into successive 
paths (through inactive grains) between the active grains.
 
The summarized grain graph is maintained by the monitor for answering a specific causal concurrency query between two grains $\focalGrain{1}$ and $\focalGrain{2}$. Formally:

\begin{definition}[Summarized Grain Graph]
Let $\tr$ be a run, $\grains$ be a valid set of scattered grains and 
let $\pi$ be a prefix of $\tr$.
The summarized conflict graph of $\pi$ is 
$\SGGraph{\pi, G} = (V_\pi, E_\pi)$, where
\begin{enumerate}
	\item $V_\pi = \actv{\pi, \grains} \cup \set{\focalGrain{1}, \focalGrain{2}}$ 
	is the set containing the active grains at the end of $\pi$ as well as the focal grains,
	\item $(\grain, \grain') \in E_\pi$ if $\grain \deadpath{\pi, \grains} \grain'$,
\end{enumerate}
\end{definition}

Summarized grain graphs sufficiently capture the reachability information between the focal grains:
\begin{proposition}
\proplabel{summarized-reachability}
Let $\tr$ be a run and let $\grains$ be a valid set of scattered grains in $\tr$.
There is a path from $\focalGrain{1}$ to $\focalGrain{2}$ in $\GGraph{\tr, G}$
iff there is a prefix $\pi$ of $\tr$ such that there is a path from $\focalGrain{1}$ to $\focalGrain{2}$
in $\SGGraph{\pi, G}$.
\end{proposition}

It remains to argue that a constant-space streaming algorithm (that reads an input run in one pass from left to right) can successfully construct the summarized grain graph for the run. Intuitively, every time an active grain $g$ is about to end, and as such become inactive and disappear, 
its summary is added to all its predecessors $g'$; that is, all nodes in the summarized grain graph that have a incoming edge to $g$. This way, a future active grain $g''$, that would be a successor of $g$ in the grain graph, will now become the successor of all $g'$'s, since the summary information stored in them will trigger the formation of an edge from them to $g''$. We make this idea formal in the detailed description of the monitors in the next sections.

For example, recall \figref{challenges}(c) and assume we want to query causal concurrency between $g_1$ and $g_2$. As such, both $g_1$ and $g_2$ have dedicated nodes in the summarized grain graph independent of their activeness status. The dashed edge appears in the summarized grain graph in place of the path from $g_1$ to $g_4$, while $g_4$ is active. Once $g_4$ is finished, as a predecessor of $g_4$, $g_1$ would remember the $\wt(y)$ access so that it can add an edge to $g_3$ when its $\rd(y)$ appears.


\subsection{Monitoring for a Fixed set of Minimal Scattered Grains}
\seclabel{scattered-monitor-fixed}


The description of the monitor that checks grain graph concurrency
(under a given set of minimal grains) is now straightforward as it essentially tracks
and updates the summarized grain graph after each prefix of the run that is seen.
We assume that the run contains symbols $\bgn$ and $\egn$ denoting
start and end of grains. 
For ease of presentation, let us assume that the run labels the \emph{focal grains}
$\focalGrain{1}$ and $\focalGrain{2}$ with fresh identifiers
identifiers $\focalEv_1$ and $\focalEv_2$.
Thus, the set of grain identifiers is thus
$\grainIDs = \set{1, 2, \ldots, |\vars|} \uplus \set{\focalEv_1, \focalEv_2}$.

The commutativity status of a grain depends on the
\emph{pending} variables of a grain,
i.e., for each grain $\grain$, we identify the set of variables
$x$ such that $x$ is read at some event $e \in \grain$ but not written to in $\grain$
(i.e., $\rdf{\tr}(e) \not\in \grain$),
or $x$ is written at some event $e \in \grain$ to but is read outside of $\grain$
i.e., $\exists e', \rdf{\tr}(e') = e \land e' \not\in \grain$).
While this information can be guessed non-deterministically
and checked later on, for simplifying the presentation of our monitor, 
we will also assume that the
alphabet encodes this information as part of the $\bgn$ marker of grain.
Thus, the alphabet of runs can be assumed to be 
\[
\gMarkAlphabet = \grainIDs \times (\alphabet \uplus \set{\bgn} \times \powset{\vars} \uplus \set{\egn}).
\]
For the purpose of this section, we will assume that the runs we consider are valid
strings over the alphabet $\gMarkAlphabet$ that are minimal, and use the language
$L_\textsf{VMG} = \{\tr \in \gMarkAlphabet^* \, | \, \tr \text{ represents a \underline{\bf v}alid \underline{\bf m}inimal \underline{\bf g}rain}$ 
annotation with focal grains$\}$ to denote the set of all such annotated runs.
Below, we present the annotated version of the run in \figref{n}(a) with all grains $g_1, g_2, g_3, g_4$,
assuming the two focal grains are $g_2$ and $g_4$ is presented below.
\begin{align*}
\begin{array}{rcl}
\tr & = & (\bgn, \emptyset)^{\focalEv_1} \, \ev{T_2, \wt(x)}^{\focalEv_1} \, 
(\bgn, \emptyset)^{1} \, \ev{T_1, \wt(z)}^{1} \, \ev{T_2, \rd(z)}^{1} \egn^{1} \, 
(\bgn, \emptyset)^{1} \, \ev{T_3, \wt(z)}^{1} \, \ev{T_3, \rd(z)}^{1} \egn^1 \\
&& \, \ev{T_3, \rd(x)}^{\focalEv_1} \, \egn^{\focalEv_1} \,
(\bgn, \emptyset)^{\focalEv_2} \, \ev{T_1, \wt(x)}^{\focalEv_2} \, \ev{T_1, \rd(x)}^{\focalEv_2} \, \egn^{\focalEv_2}
\end{array}
\end{align*}
In the above, we use the notation $a^i$ as shorthand for $(i, a) \in \gMarkAlphabet$.
Observe that the (unique) grains with grain identifier $\focalEv_1$ (namely grain $g_2$)
overlaps with both the
grains with identifier $1$ (i.e., grains $g_1$ and $g_3$).
Also observe that $L_\textsf{VMG}$ is a regular language.
Thus, a monitor that works correctly for runs in this language also works for non-annotated runs
because the language of a constant space monitor is regular
and thus closed under projection.
As discussed in Section \ref{sec:gmonitor}, it suffices to consider the largest sound
commutativity relation on a given set of grains.
In fact, the largest commutativity relation also has a succinct representation --- the dependence
between any two scattered grains in this relation can be checked only using
the \emph{signatures} $g_1 = \langle E_1, V_1 \rangle$ and $g_2 = \langle E_2, V_2 \rangle$ of these grains:
\begin{align*}
\depend{g_1,g_2} &\iff  
\exists e_1 \in E_1, e_2 \in E_2 \cdot
\bigg(
\begin{aligned} 
\begin{array}{ll}
& \ThreadOf{e_1} = \ThreadOf{e_2} \\
\lor & \VariableOf{e_1} = \VariableOf{e_2} \not\in V_1 \cap V_2 \land \wt \in \set{\OpOf{e_1}, \OpOf{e_2}}
\end{array}
\end{aligned}
\bigg)
\end{align*}

Recall that the signatures are bounded sized, and so is their dependence relationship.
Thus, in order to maintain the summarized graph inductively, states 
stores additional information to correctly infer commutativity with other grains,
including those that have finished.

Our monitor essentially tracks grain signatures of interest to infer edges between 
relevant grains.
Let $\pi$ be a prefix of $\tr$ and let $\grain$ be some grain
that is active in $\pi$.
Then, we use the notation  
$\contents_\pi(\grain) = \setpred{\tr[e]}{e \in \grain \cap \range{|\pi|}}$
to denote the \emph{\underline{\sf C}ontents} of $\grains$ in $\pi$.
Likewise, we denote the set of \emph{\underline{\sf P}ending}
variables of $\grain$ by 
$\pvars(\grain) = \setpred{x \in \vars}{\exists e \in \grain, \VariableOf{e} = x, \neg \cmpvar{\grain, x}}$.
For each active grain $\grain$ that we track in our monitor,
we will also maintain summarized information about those grains 
that are no longer active but can be reached from $\grain$,
to accurately infer edges in the summarized graph (alternatively \emph{retroactive} paths in the grain graph).
The first such information is the \emph{\underline{\sf S}ummarized \underline{\sf C}ontents}
of the inactive but reachable grains:
$\summaries_\pi(\grain) = \bigcup \setpred{\contents_\pi(\grain')}{\exists \grain' \subseteq \pi, \grain \deadpath{\pi, \grains} \grain'}$.
Likewise, we use
$\spvars_\pi(\grain) = \bigcup \setpred{\pvars_\pi(\grain')}{\exists \grain' \subseteq \pi, \grain \deadpath{\pi, \grains} \grain'}$ to denote the
\emph{\underline{\sf S}ummarized \underline{\sf P}ending} variables that $\grain$
must keep track of.

We are now ready to formally describe our 
\underline{\bf G}rain \underline{\bf G}raph 
concurrency monitor
$\autsup{{\sf GG}} = (Q_{\sf GG}, q_0, \delta_{\sf GG}, F_{\sf GG})$.
The states $Q_{\sf GG}$ of $\autsup{{\sf GG}}$ are tuples of the form $\stuple{V, E, \contents, \pvars, \summaries, \spvars}$ where
\begin{itemize}
	\item $V \subseteq \grainIDs$ represents the set of active and focal grains
	of the summarized graph.

	\item $E \subseteq V \times V$ represents the edges of the summarized graph.

	\item $\contents: V \to \powset{\alphabet}$ maintains the \emph{contents} of active grains.
	For each grain $\grain$ tracked as a vertex $u$, we will have 
	$\contents(u) = \contents_\pi(\grain)$
	after having processed the prefix $\pi$.

	\item $\pvars: V \to \powset{\vars}$ to track the set $\pvars_\pi(\grain)$ 
   for each grain $\grain$ (at the end of prefix $\pi$) tracked as some vertex in $V$.

	\item $\summaries: V \to \powset{\alphabet}$ is such that $\summaries_\pi(u)$ 
	tracks the set $\summaries_\pi(\grain)$ at the end of prefix $\pi$, where $u$ represents the grain $\grain$.

	\item $\spvars: V \to \powset{\vars}$ to track the set $\spvars_\pi(\grain)$
   for each grain $\grain$ (at te end of prefix $\pi$) tracked as some vertex in $V$.
\end{itemize}

The start state of the monitor is 
$q_0 = \stuple{\emptyset, \emptyset, \lambda i \cdot \emptyset, \lambda i \cdot \emptyset, \lambda i \cdot \emptyset, \lambda i \cdot \emptyset}$.
All states $\stuple{V, E, \contents, \pvars, \summaries, \spvars}$ in which
$\focalEv_1, \focalEv_2 \in V$ and further $\focalEv_2$
is reachable from $\focalEv_1$ via a path using edges in $E$ are marked rejecting,
and others are accepting (i.e., belong to $F_{\sf GG}$).
The transitions of the monitor are described in \figref{sgcm}.
When we see an event $e = (i, (\bgn, Y))$ that demarcates the beginning of a grain
with identifier $i$, we also know upfront the set of pending variables
in the grain. 
At this point, we create a new node labeled $i$ in the graph and also track this set $Y$.
When we see the end of a grain (marked $\egn$) with identifier $i$, then we \emph{garbage collect}
the node $i$ from the graph.
As part of this garbage collection, we add an edge from the immediate
predecessors of $i$ to its immediate successors;
this is captured by the operation $\mergeEdge{\cdot}{\cdot}$.
Further, the maps $\contents$ and $\pvars$ reset the entry corresponding to $i$,
and $\summaries$ and $\spvars$
entries of the predecessors of $i$ are updated to include 
$\contents(i)$, $\summaries(i)$, $\pvars(i)$ and $\spvars(i)$;
this is captured using $\mergeSum{\cdot, \cdot}{\cdot}{\cdot}$.
When we see a read or a write event $a = \ev{T, o, x}$ in grain corresponding to
the node $i$, we add edges from all nodes $j$ to $i$ such that
$j$ conflicts with $a$, i.e.,
either the contents or the summary of $j$ contains a letter that conflicts with $a$,
or one of $j$ or an inactive grain reachable from $j$ has $x$ as a pending variable.

One can prove that the monitor in \figref{sgcm} correctly maintains $\deadpath{\pi, \grains}$
between each pair of active/focal grains
after every prefix $\pi$ of $\tr$, and as such, it can correctly decide whether the two focal grains are scattered grain concurrent or not:

\begin{theorem}
\thmlabel{weak-concurrency-fixed-grains-monitor}
Given a run $\tr \in L_\textsf{VMG}$ annotated with grains $\grains$, 
$\tr$ is accepted by $\autsup{{\sf GG}}$ iff 
the two focal grains are grain graph concurrent under $\grains$.
\end{theorem}



\begin{figure}[t]
\fbox{
\parbox{\textwidth}{
\begin{center}
\begin{tabular}{lll}
State & Event & State Update \\ \hline \hline
$\stuple{V, E, \contents, \pvars, \summaries, \spvars}$ & $e = (i, (\bgn, Y))$ & $\stuple{V \uplus \set{i}, E, \contents, \pvars[i \mapsto Y], \summaries, \spvars}$\\
\hline
$\stuple{V, E, \contents, \pvars, \summaries, \spvars}$ & $e = (i, \egn)$, $i\in\set{\focalEv_1, \focalEv_2}$ & $\stuple{V, E, \contents, \pvars, \summaries, \spvars}$ \\
\hline
$\stuple{V, E, \contents, \pvars, \summaries, \spvars}$ & $e = (i, \egn)$, $i\not\in\set{\focalEv_1, \focalEv_2}$ & 
$\begin{aligned}
&\stuple{V', E', \contents', \pvars', \summaries', \spvars'}, \text{ where,}\\
&\;\; V' = V - \set{i}, E' = \mergeEdge{E}{i} \\
&\;\; \contents' = \contents[i \mapsto \emptyset], \pvars' = \pvars[i \mapsto \emptyset], \\
&\;\; \summaries' = \mergeSum{\summaries, \contents}{E}{i}, \\
&\;\; \spvars' = \mergeSum{\spvars, \pvars}{E}{i}
\end{aligned}$\\
\hline
$\stuple{V, E, \contents, \pvars, \summaries, \spvars}$ & $e = (i, a)$, $a\in \alphabet$ &
$\begin{aligned}
&\stuple{V, E', \contents \cup \set{a}, \summaries, \spvars}, \text{ where,}\\
&\;\; E' = E \cup \setpred{(j, i)}{ j \neq i \text{ and }\\
& \quad\quad\quad\depEdge{\contents(j){\cup}\summaries(j),\; a,\; \pvars(j){\cup}\spvars(j){\cup}\pvars(i)} \;}
\end{aligned}$\\
\hline
\hline
\end{tabular}%
\begin{align*}
\mergeEdge{E}{i} &= (E - \setpred{(i, j), (j, i)}{j \neq i}) \cup \setpred{(j, k)}{(j, i) \in E, (i, k) \in E}\\
\mergeSum{SM, M}{E}{i} &= \lambda j \cdot   \begin{cases} 
                                            \emptyset & \text{ if } j = i \\ 
                                            SM(j) \cup M(i) \cup SM(i) & \text{ if } j \neq i, (j, i) \in E \\
                                            SM(j) & \text{ owise }
                                        \end{cases}\\
\depEdge{S, a, Z} \iff& 
\exists b \in S \cdot \bigg(
\ThreadOf{a} = \ThreadOf{b} \lor 
\big(\VariableOf{a} = \VariableOf{b} \in Z \land \wt \in \set{\OpOf{a}, \OpOf{b}}\big)
\bigg)
\end{align*}%
\end{center}}}\vspace{-5pt}
\caption{Grain Graph Concurrency Monitor for Annotated Runs $\autsup{{\sf GG}}$: The monitor rejects if it is in a state $(V, E, \contents, \pvars, \summaries, \spvars)$ such that $(\focalEv_1, \focalEv_2) \in E^*$ at the end of the run, and accepts otherwise.}.
\figlabel{sgcm}\vspace{-20pt}
\end{figure}
\vspace{-0.1in}

\subsection{Monitoring Scattered Grain Concurrency}
\seclabel{weak-concurrency-monitor}

Similar to the case for grain concurrency, the monitor for scattered grain concurrency  (\defref{lgc}) can non-deterministically guess the choice of scattered grains and check, using the monitor in \figref{sgcm}, if any of these guesses is valid and declares the two focal grains to causally concurrent.

\begin{theorem}
\thmlabel{weak-grain-concurrency-regular}
There exists a monitor $\autsup{}$ that uses $2^{O(|\vars| \cdot |\alphabet|)}$ space and
accepts the word $\tr \in \alphabet^* \focalEv_1 \alphabet^+ \focalEv_2 \alphabet^+$ iff 
events that appear immediately after $\focalEv_1$ and $\focalEv_2$  are scattered grain concurrent in $\tr$. 
Consequently, scattered grain concurrency can be checked in constant space.
\end{theorem}


%

The proof of the above theorem relies on \thmref{weak-concurrency-fixed-grains-monitor} 
and the observations that given a run $\tr \in \alphabet$ (with $\focalEv_1$ and $\focalEv_2$),
we can guess the grains on $\tr$ in constant space, validate whether the guessed grains
are minimal and valid, and finally if the resulting annotated run is accepted
 by $\autsup{\sf GG}$.
The number of states in $\autsup{\sf GG}$ is $2^{|\vars|^2} \cdot (2^{|\alphabet|})^{|\vars|} \cdot (2^{|\vars|})^{|\vars|} \cdot (2^{|\alphabet|})^{|\vars|} \in 2^{O(|\vars| \cdot |\alphabet|)}$
can can be obtained by counting the different ways to construct graphs on $|\vars|$ vertices
and deciding on the contents, pending variables and summaries and summarized pending variables for
these vertices.
Since the scattered grain concurrency monitors essentially adds non-determinism on top of this DFA,
its deterministic version has exponentially many states, and each state thus has size $2^{O(|\vars| \cdot |\alphabet|)}$.


\section{Related Work}
\seclabel{related}

There is little work in the literature in the general area of generalizing commutativity-based analysis of concurrent programs. On the theoretical side, some generalizations of Mazurkiewicz traces have been studied before \cite{bauget,gtrace,gtrace-types,KLEIJN199898}. The focus has mostly been on incorporating the concept of {\em contextual} commutativity, that is, when two events can be considered commutative in some contexts but not in all. This is motivated, among other things, by send/receive or producer/consumer models in distributed systems where send and receive actions commute in contexts with non-empty message buffers. 

On the practical side, similar ideas were used in \cite{mypopl20,mypldi22,Genest07,Desai14} to reason about equivalence classes of concurrent and distributed programs under contextual commutativity. There is a close connection between this notion of contextual commutativity and the concept of {\em conditional independence} in partial order reduction \cite{KatzP92,GodefroidP93} which is used as a weakening of the independence (commutativity) relation by making it parametric on the current {\em state} to increase the potential for reduction.  


This work points out that, in general, the coarsest notion of equivalence
may not yield monitoring-style algorithms for analyzing runs of concurrent programs, and puts forward an alternative notion of equivalence, coarser than trace equivalence,  that can be efficiently used in monitoring causal concurrency.
Below, we briefly survey some application domains relevant to programming languages research where our proposed equivalences can have an immediate positive impact.

\paragraph{Concurrency Bug Prediction}
Dynamic analysis techniques for detecting concurrency bugs such as data races~\cite{Flanagan09},
deadlocks~\cite{Samak2014} and atomicity violations~\cite{Flanagan2008,FM08CAV,Sorrentino10}
suffer from poor coverage since their bug detection capability is determined
by the precise thread scheduling observed during testing.
Predictive techniques~\cite{Koushik05,Said11,Wang09} such as those for
detecting data races~\cite{Smaragdakis12,Huang14,Kini17,Pavlogiannis2020,Roemer18,Mathur21}
or deadlocks~\cite{Kalhauge2018,Tunc2023deadlock}
enhance coverage by exploring equivalent runs that might yield
a bug.
The core algorithmic problems involved in such approaches are akin to checking
causal concurrency.
Coarser yet tractable equivalence relations can yield better analysis techniques with more accurate predictions.
Recent work~\cite{Kulkarni2021} explores hardness results for data race prediction, including
Mazurkiewicz-style reasoning, when the alphabet is not assumed to be of constant size.

\paragraph{Dynamic Partial Order Reduction for Stateless Model Checking}
There has been a rising interest in developing dynamic partial order based~\cite{Flanagan2005} 
stateless model checking techniques~\cite{Godefroid1997}
that are \emph{optimal}, in that they explore as few program runs as possible from the underlying program. 
 Coming up with increasingly coarser equivalences is the go-to approach for this.
The notion of reads-from equivalence we study has received
a lot of attention in this context~\cite{Abdulla2019,Chalupa2017,Kokologiannakis2022,Kokologiannakis2019}.
Recent works also consider an even coarser reads-value-from~\cite{ChatterjeePT19,Agarwal2021} equivalence. Events encode variable values and two runs are equivalent if every read observes the
same value across the two runs. The problem of verifying sequential
consistency, which is one of the key algorithmic questions underlying such approaches, was proved to be intractable in general by Gibbons and Korach~\cite{Gibbons1997}. 

\section{Conclusion and Future Work}
This paper demonstrates that reads-from equivalence, 
the most relaxed sound notion of equivalence on concurrent program runs, does not share the nice algorithmic properties of (Mazurkiewicz) trace equivalence. This poses the following research questions: ``Are there other notions of equivalence, which remain sound, relax trace equivalence, and yet maintain its key desirable algorithmic properties? And, what design principles bring about algorithmic simplicity?''. 
We propose two new notions of equivalence in this paper under which causal concurrency can be decided by a streaming algorithm in constant space.   
Equivalence based on contiguous grains shares the characteristic with trace equivalence, in that it is definable purely in terms of commutativity. Remarkably, while this characteristic is lost with scattered grains, the algorithmic simplicity remains. The point of commonality between the two notions is that each individual event can only {\em move} in one role: either individually, or as part of a grain. As we demonstrate in \thmref{wc-hardness}, algorithmic hardness kicks in when this rule is broken in the simplest of syntactic settings. It would be interesting to investigate whether one can further relax the notion of {\em grain equivalence} by breaking this barrier. 

As we note in \secref{beyond-contiguous-grains},  the straightforward determinization of the (scattered) grain concurrency monitor can result in an exponential blowup on the number of threads and shared variables. We treat these parameters as constants, in a manner similar to trace theory's reliance on a finite (constant-sized) alphabet of actions.  There seems to be a tradeoff between how coarse the equivalence relation is and how efficiently causal concurrency can be monitored. The research question of how to devise a practical monitor when the number of threads or variables is not very small remains an interesting direction for future research.

Finally, this paper studies coarser equivalences in the context of a {\em causal concurrency query}. In several application domains, the standard oracle for causal concurrency based on trace equivalence can be replaced with the  (scattered) grain concurrency from this paper. Yet, there remain other application domains in which trace equivalence is used in completely different ways: for example, proof simplification \cite{lics2023} by verifying a commutativity-based reduction of a concurrent program. The notion of soundness used in this paper suffices when the object of study is a single program run. In proof simplification, however, a set of program runs must be considered together, and, as we argued, different grain commutativity relations may be sound in different program runs. Moreover, the alphabet of actions for proof simplification typically includes atomic program statements rather than shared variable reads and writes. As such, the theoretical results presented in this paper do not immediately offer a solution in such domains. It will be interesting to explore how similar coarse equivalences, based on commutativity of words (rather than symbols), can be designed to be exploited for proof simplification.   
\bibliographystyle{ACM-Reference-Format}
\bibliography{references}

\newpage
\appendix

\section{Proofs from \secref{semantic-equivalence}}\applabel{semantic-equivalence}

In the following we will use the notation $\checkOrder{\rfnovaleq}{w}{i}{j}$ to denote that the events
$i$ and $j$ are causally ordered in $w$ under $\rfnovaleq$


\subsection{Proof of \thmref{check-order-semantic-linear-space-lower-bound} and \thmref{check-order-semantic-tradeoff}}
\applabel{hardness}

We first state some properties of reads-from equivalence.
We first define additional notation.
For an execution $w \in \alphabet^*$, we define
the relation $\ctrf{w}$ defined as follows:
\[
\ctrf{w} = \setpred{(i, j)}{i, j \in \range{|w|}, \checkOrder{\rfnovaleq}{w}{i}{j}}
\]

\begin{proposition}
\proplabel{properties-check-order-semantic}
Let $w \in \alphabet^*$ be an execution.
The relation $\ctrf{w}$ defined above satisfies the following properties.
\begin{description}
	\item[(Partial order).] $\ctrf{w}$ is a partial order. That is, $\ctrf{w}$ is irreflexive and transitive.
	\item[(Intra-thread order).] $\ctrf{w}$ orders events of $w$ in the same thread. 
	That is, for every $i < j \in \range{|w|}$, if
	$\ThreadOf{w[i]} = \ThreadOf{w[j]}$, then $i \ctrf{w} j$.
	\item[(Reads-from).] $\ctrf{w}w$ orders events if there is a reads-from dependency between them.
	That is, for every $i < j \in \range{|w|}$, if $\rdf{w}(j) = i$, then $i \ctrf{w} j$.
	\item[(Implied orders).] 
	\itmlabel{implied-order-semantic}
	Let $i, j, k \in \range{|w|}$ be distinct indices such that
	such that $\VariableOf{i} = \VariableOf{j} = \VariableOf{k}$,
	$\OpOf{i} = \OpOf{k} = \wt$ and $\OpOf{j} = \rd$
	and $\rdf{w}(j) = i$.
	\begin{itemize}
		\item If $i \ctrf{w} k$, then $j \ctrf{w} k$.
		\item If $k \ctrf{w} j$, then $k \ctrf{w} i$.
	\end{itemize}
\end{description}
\end{proposition}

\begin{proof}
Follows from the definition of $\ctrf{w}$.
\end{proof}

The proof of \thmref{check-order-semantic-linear-space-lower-bound}
relies on a reduction from the following
language, parametrized by $n \in \natsp$:
\begin{align*}
\langeq_n = \setpred{\vect{a}\#\vect{b}}{\vect{a}, \vect{b} \in \set{0,1}^n \text{ and } \vect{a} = \vect{b}}
\end{align*}

We first observe that there is a linear space lowerbound 
for the problem of recognition of this language.

\begin{lemma}
Any streaming algorithm that recognizes $\langeq_{n}$ uses $\Omega(n)$ space.
\end{lemma}

\begin{proof}
Assume towards contradiction otherwise, i.e., 
there is a streaming algorithm that uses $o(n)$ space.
Hence the state space of the algorithm is $o(2^n)$.
Then, there exist two distinct $n$-bit strings $a \neq a'$, 
such that the streaming algorithm is in the same 
state after parsing $\vect{a}$ and $\vect{a'}$.
Hence, for any $n$-bit string $\vect{b}$, the algorithm gives 
the same answer on inputs $\vect{a}\#\vect{b}$ and $\vect{a'}\#\vect{b}$.
Since the algorithm is correct, it reports that $a\#a$ belongs to $L_n$.
But then the algorithm reports that $\vect{a'}\#\vect{a}$ also belongs to $L_n$, a contradiction.
The desired result follows.
\end{proof}

\begin{proof}[Proof of \thmref{check-order-semantic-linear-space-lower-bound}]
We will now show that there is a linear-space lower bound for
\probref{cc-symb} 
by showing a reduction from $\set{0,1}^n\#\set{0,1}^n$
to $\alphabet^*$ in one pass streaming fashion using constant space.
The constructed string will be such that there are unique indices
$i$ and $j$ whose labels correspond to the symbols $c$ an $d$ in the causal concurrency problem.

\newcommand{\trbit}{v}
Consider the language $L_n$ for some $n$.
We describe a transducer $\Transducer_n$ such that, on input a string 
$\trbit=\vect{a}\#\vect{b}$, the output $\Transducer_n(\trbit)$ is an execution 
$\tr$ with $2$ threads $t_1$ and $t_2$, $O(n)$ events and $6$ variables such
that  $\tr$ is of the form 
\[
\tr = \pi \concat \kappa \concat \eta \concat \delta
\]
We fix the letters $c$ and $d$ for which we want to check
for causl concurrency to be $c = \ev{t_1, \rd, u}$
and $d = \ev{t_2, \wt, u}$, where $u \in \vars$.
The fragments $\pi$ and $\eta$ 
do not contain any access to the variable $u$.
The fragments $\kappa$ and $\delta$ contain accesses to $u$.
\figref{check-order-lb-trace} describes our construction for $n=4$.


\begin{figure}[t]
   \centering
   \begin{subfigure}[b]{0.45\textwidth}
      \centering
      \scalebox{1.0}{
      \input{figures/trace-construction-linear-space}
      }
      \caption{Run constructed for $n=4$.}
      \figlabel{check-order-lb-trace}
   \end{subfigure}
   \hfill
   \begin{subfigure}[b]{0.45\textwidth}
      \centering
      \scalebox{1.0}{
      \input{figures/reordering-unequal-strings}
      }
      \caption{Reordering for $a= 11\textcolor{red}{\underline{0}}0$ and $b = 11\textcolor{red}{\underline{1}}0$.}
      \figlabel{check-order-lb-reordering}
   \end{subfigure}
   \caption{Reduction from $\langeq_4 = \setpred{a_1a_2a_3a_4\#b_1b_2b_3b_4}{ \forall i, a_i = b_i \in \set{0,1} }.$
Trace construction (on left) and example of equivalent trace when $\vect{a} \neq \vect{b}$.}
\figlabel{check-order-lb}
\end{figure}

Our reduction ensures the following:
\begin{enumerate}
\item If $\trbit \in \langeq_n$, then there are $\theta_1, \theta_2 \in \range{|\tr|}$ such that 
$\ThreadOf{\tr[\theta_1]} = t_1, \ThreadOf{\tr[\theta_2]} = t_2$, 
$\VariableOf{\tr[\theta_1]} = \VariableOf{\tr[\theta_2]} = u$ 
and $\checkOrder{\rfnovaleq}{\tr}{\theta_1}{\theta_2}$.
\item If $\trbit \not\in \langeq_n$, then there are no such $\theta_1$ and $\theta_2$.
\end{enumerate}
Moreover, $\Transducer_n$ will use $O(1)$ working space.

Let us now describe the construction of $\tr$.
At a high level, the sub-execution $\pi$ 
corresponds to the prefix $\vect{a} = a_1\concat a_2 \cdots \concat a_n$ of $\trbit$ 
and the sub-execution $\eta$ corresponds to the suffix 
$\vect{b} = b_1\concat b_2 \cdots \concat b_n$ of $\trbit$.
Both these sub-executions have $O(n)$ events.
The sub-executions $\kappa$ and $\delta$ have $O(1)$ (respectively $4$ and $1$)
events.
The sequences of events in the two sub-executions $\pi$ and $\eta$ ensure that
there is a `chain' of conditional dependencies, that
are incrementally met until the point when the two substrings $\vect{a}$ and $\vect{b}$ match.
If $\vect{a} = \vect{b}$ (i.e., full match), then the dependency chain further
ensures that there are events, at indices $\theta_1$ and $\theta_2$,
both in $\kappa$ (accessing the variable $u$ in threads $t_3$ and $t_2$ respectively)
such that $\checkOrder{\rfnovaleq}{\tr}{\theta_1}{\theta_2}$.

The execution fragment $\pi$ is of the form 
\begin{align*}
\pi = \pi_1 \concat \pi_2 \cdots \concat \pi_n.
\end{align*}
Here, the fragment $\pi_i$ corresponds to the $i^\text{th}$
bit $a_i$ of $\vect{a}$,
and only contains events performed by the thread $t_1$.
The fragment $\eta$ is of the form
\[
\eta = \eta_1 \concat \eta_2 \cdots \concat \eta_n 
\]
Here, $\eta_i$ corresponds to $b_i$, the  $i^\text{th}$ bit in $\vect{b}$
and only contains events of thread $t_2$.
The variables used in the construction are $\set{c, x_0, x_1, y_0, y_1, u}$,
where $u$ is the special variable whose events will be ordered or not based on the input.
In the rest of the construction, we will use  
the notation $\ev{t, o, v, \gamma}$ 
to denote the unique index $\alpha$ occuring within the fragment $\gamma$ 
for which $\tr[\alpha] = \ev{t, o, v}$.
We next describe each of the fragments 
$\pi_1, \ldots, \pi_n$, $\eta_1, \ldots, \eta_n$
and the fragments $\kappa$ and $\delta$.

\myparagraph{Fragment $\pi_1$}{
	The first fragment $\pi_1$ of $\pi$ is as follows:
	\begin{align*}
	\pi_1 = \ev{t_1, \wt, x_{\neg a_1}} \concat \ev{t_1, \wt, c} \concat \ev{t_1, \wt, x_{a_1}}
	\end{align*}
	That is, the last (resp. first) event writes to $x_0$ (resp. $x_1$) if $a_1 = 0$, 
	otherwise it writes to the variable $x_1$ (resp. $x_0$).
}

\myparagraph{Fragment $\eta_1$}{
	The first fragment $\eta_1$ of $\eta$ is as follows:
	\begin{align*}
	\eta_1 = \ev{t_2, \rd, x_{b_1}} \concat \ev{t_2, \wt, c}
	\end{align*}

	In the entire construction, the variables $x_0$ and $x_1$ are being written-to only
	in fragment $\pi_1$, and (potentially) read only in $\eta_1$.
	This means that, either $\rdf{\tr}(\ev{t_2, \rd, x_{b_1}, \eta_1}) = \ev{t_1, \rd, x_{a_1}, \pi_1}$
	(if $a_1 = b_1$)
	or $\rdf{\tr}(\ev{t_2, \rd, x_{b_1}, \eta_1}) = \ev{t_1, \rd, x_{\neg a_1}, \pi_1}$
	(if $a_1 \neq b_1$). In summary,
	\begin{align}
	\equlabel{pi1_c_to_eta1_c}
	a_1 = b_1 \implies \checkOrder{\rfnovaleq}{\tr}{\ev{t_1, \wt, c, \pi_1}}{\ev{t_2, \wt, c, \eta_1}} 
	\end{align}
}

\myparagraph{Fragment $\pi_i$ ($i \geq 2$)}{
For each $i\geq 2$, the fragment $\pi_i$ is the following
\begin{align*}
	\pi_i = \ev{t_1, \wt, y_{a_i}}  \concat \ev{t_1, \rd, c} \concat \ev{t_1, \wt, c} \concat \ev{t_1, \wt, r_{a_i}}.
\end{align*}	
Let us list some reads-from dependencies introduced due to $\pi_i$ ($i \geq 2$):
\begin{align*}
	\begin{array}{ccc}
		\rdf{\tr}(\ev{t_1, \rd, y_{a_2}, \pi_i}) = \ev{t_1, \wt, y_{a_2}, \pi_i} & \quad\text{ and }\quad
		\rdf{\tr}(\ev{t_3, \rd, c, \pi_i}) = \ev{t_3, \wt, c, \pi_{i-1}} \\
	\end{array}
\end{align*}
The reads-from mapping 
$\rdf{\tr}(\ev{t_1, \rd, c, \pi_2}) = \ev{t_1, \wt, c, \pi_1}$,
together with \equref{pi1_c_to_eta1_c}, implies that
$\ev{t_1, \rd, c, \pi_2} \ctrf{\tr} \ev{t_2, \wt, c, \eta_1}$ (see \propref{properties-check-order-semantic}),
in the case $a_1 = b_1$.
Finally, the read-from dependency from $t_1$ to $t_3$ due to $f$ gives us:
\begin{align}
\equlabel{pi2_a2_to_eta1_c}
a_1 = b_1 \implies \checkOrder{\rfnovaleq}{\tr}{\ev{t_1, \wt, y_{a_2}, \pi_2}}{\ev{t_2, \wt, c, \eta_1}}
\end{align}
}

\myparagraph{Fragment $\eta_i$ ($i \geq 2$)}{
For each $i\geq 2$, the fragment $\eta_i$ is the following
\begin{align*}
	\eta_i = \ev{t_2, \wt, y_{b_i}} \concat \ev{t_2, \wt, c}.
\end{align*}
It is easy to see from \equref{pi2_a2_to_eta1_c} that if $a_1 = b_1$
then the intra-thread dependency between $\eta_1$ and $\eta_2$
further implies that
$\ev{t_1, \wt, y_{a_2}, \pi_2} \ctrf{\tr} \ev{t_2, \wt, y_{b_2}, \eta_2}$.
Now, if additionally $a_2 = b_2$, we get:
\begin{align}
\equlabel{pi2_a2_to_eta2_b2}
(\forall i \leq 2, a_i = b_i) \implies \checkOrder{\rfnovaleq}{\tr}{\ev{t_1, \rd, y_{a_2}, \pi_2}}{\ev{t_2, \wt, y_{b_2}, \eta_2}}
\end{align}.

In fact, the same reasoning can be inductively extended for any $2 \leq k \leq n$:
\begin{align}
\equlabel{pii_ai_to_etai_bi}
(\forall i \leq k, a_i = b_i) \implies \checkOrder{\rfnovaleq}{\tr}{\ev{t_1, \rd, y_{a_k}, \pi_k}}{\ev{t_2, \wt, y_{b_k}, \eta_k}}.
\end{align}
The base case of $k=2$ follows from \equref{pi2_a2_to_eta2_b2}.
For the inductive case, assume that the statement holds for some $k < n$,
and that $(\forall 2 \leq i \leq k, a_i = b_i)$.
Together with intra-thread dependencies, we have
$\checkOrder{\rfnovaleq}{\tr}{\ev{t_1, \wt, c, \pi_k}}{\ev{t_2, \wt, c, \eta_k}}$.
It follows from \propref{properties-check-order-semantic} that 
$\ev{t_1, \wt, c, \pi_{k+1}} \ctrf{\tr} \ev{t_2, \wt, c, \eta_k}$
This, in turn implies that $\ev{t_1, \wt, y_{a_{k+1}}, \pi_{k+1}}$
is causally ordered before $\ev{t_2, \wt, y_{b_{k+1}}, \eta_{k+1}}$.
Thus,
$\checkOrder{\rfnovaleq}{\tr}{\ev{t_1, \rd, y_{a_{k+1}}, \pi_{k+1}}}{\ev{t_2, \wt, y_{b_{k+1}}, \eta_{k+1}}}$
if $a_{k+1} = b_{k+1}$.
}

\myparagraph{Fragments $\kappa$ and $\delta$}{
	The sequence $\kappa$ and $\delta$ are:
	\begin{align*}
	\begin{array}{ccc}
		\kappa = \ev{t_1, \wt, u} \concat \ev{t_1, \rd, c} \concat \ev{t_1, \rd, u} \concat \ev{t_2, \wt, u} &\text{ and }& \delta = \ev{t_2, \rd, u}
	\end{array}
	\end{align*}

	The reads-from dependencies induces due to $\kappa$ and $\delta$ are:
	\begin{align*}
	\begin{array}{ccc}
		\rdf{\tr}(\ev{t_1, \rd, c, \kappa}) = \ev{t_1, \wt, c, \pi_n}, &
		\rdf{\tr}(\ev{t_1, \rd, u, \kappa}) = \ev{t_1, \wt, u, \kappa}, &
		\rdf{\tr}(\ev{t_2, \rd, u, \delta}) = \ev{t_2, \wt, u, \kappa} \\
	\end{array}
	\end{align*}
}

\myparagraph{Correctness}{
	Let us make some simple observations.
	First, every read event has a write event on the same variable prior to it,
	and thus, $\rdf{w}(i)$ is well defined for every $i \in \reads{w}$
	Second, it is easy to see that the construction can be performed by a transducer
	$\Transducer_n$ in $O(1)$ space.

	($\Rightarrow$)
	Let us first prove that if $\forall i \in [n], a_i = b_i$, then 
	$\checkOrder{\rfnovaleq}{\tr}{\theta_1}{\theta_2}$,
	where $\theta_1$ is the index of the $\ev{t_1, \rd, u}$ event in $\kappa$
	and $\theta_2$ is the index of the $\ev{t_2, \wt, u}$ event in $\kappa$.
	Recall that if $a_i = b_i$ for every $i\leq n$, then 
	$\checkOrder{\rfnovaleq}{\tr}{\ev{t_3, \rd, y_{a_n}, \pi_n}}{\ev{t_2, \wt, y_{b_n}, \eta_n}}$
	(see \equref{pii_ai_to_etai_bi}).
	As a result, $\ev{t_1, \wt, c, \pi_n} \ctrf{\tr} \ev{t_2, \wt, c, \eta_n}$.
	Next, due to \propref{properties-check-order-semantic}, we get
	$\ev{t_1, \rd, c, \kappa} \ctrf{\tr}  \ev{t_2, \wt, c, \eta_n}$.
	This, together with intra-thread dependency further gives
	$\ev{t_1, \wt, u, \kappa} \ctrf{\tr}  \ev{t_2, \rd, u, \delta}$.
	If we next apply \propref{properties-check-order-semantic}, we get
	$\ev{t_1, \wt, u, \kappa} \ctrf{\tr}  \ev{t_2, \wt, u, \kappa}$.
	Applying \propref{properties-check-order-semantic} once again, we conclude that
	$\checkOrder{\rfnovaleq}{\tr}{\theta_1}{\theta_2}$

	($\Leftarrow$)
	Let us now prove that if there is an index $i$ for which $a_i \neq b_i$, 
	then 
	$\theta_1 \notctrf{\tr} \theta_2$ and $\theta_2 \notctrf{\tr} \theta_1$.
	To show this, we will construct an execution $\tr'$ such that
	$\tr' \rfnovaleq \tr$ and $\theta_2 \trord{\tr'} \theta_1$.
	In the rest, we assume $i$ is the least such index.

	First, consider the case when $i = 1$. 
	In this case, $\neg a_1 = b_1$.
	Then, $\tr'$ is the following concatenated sequence ($\kappa'$ is the second largest prefix of $\kappa$):
	\[
		\tr' = \ev{t_2, \wt, u} \concat \ev{t_1, \wt, x_{\neg a_1}} \concat \eta \concat \delta \concat \ev{t_1, \wt, x} \concat \ev{t_1, \wt, x_{a_1}} \concat \pi_2 \cdots \concat \pi_n \concat \kappa'
	\]
	First, observe that, the thread-wise projections of $\tr$ and $\tr'$ are the same.
	Next, we note that the reads-from dependencies of each read access to either $y_0, y_1, c$ or $u$
	are the same in both $\tr$ and $\tr'$:
	\begin{itemize}
		\item The only read accesses to $y_0$ or $y_1$ are in thread $t_1$, and so are their corresponding write events (as per $\tr$). In the new sequence $\tr'$, we ensure that
		all write access to $y_0$ or $y_1$ in $t_2$ occur before any access to $y_0$ or $y_1$
		in thread $t_1$.
		\item The same reasoning as above applies to the accesses to $c$.
		\item The read access to $u$ in $t_2$ appears before the write access to $u$ in $t_1$ in the new sequencd $\tr'$.
	\end{itemize}
	Finally, the reads-from dependency of $\ev{t_2, \rd, x_{b_1}}$ is also preserved. 
	Hence, $\tr'\rfnovaleq\tr$.
	Finally, observe that the corresponding events $\theta_2 = \ev{t_2, \rd, u, \kappa}$
	and $\theta_1 = \ev{t_1, \wt, u, \kappa}$ appear in inverted order.

	Next, consider the case of $i > 1$.
	Then, we construct $\tr'$ as the following interleaved sequence:

	\[
		\gamma_1 \concat \gamma_2 \cdots \concat \gamma_n \concat \kappa'
	\]
	where:
	\begin{align*}
	\begin{array}{rclr}
	\gamma_1 &=& \ev{t_2, \wt, u} \concat \ev{t_1, \wt, x_{\neg a_1}} \concat \ev{t_1, \wt, c} \concat \ev{t_1, \wt, x_{a_1}} \concat \ev{t_2, \rd, x_{b_1}} & \\
	\gamma_j &=& \ev{t_1, \wt, y_{a_j}} \concat \ev{t_1, \rd, c} \concat \ev{t_2, \wt, c} \concat \ev{t_1, \wt, c} \concat \ev{t_1, \rd, y_{a_j}} \concat \ev{t_2, \wt, y_{b_j}} & (2 \leq j < i) \\
	\gamma_i &=& \ev{t_1, \wt, y_{a_i}} \concat \ev{t_1, \rd, c} \concat \ev{t_2, \wt, c} \concat \ev{t_2, \wt, y_{b_i}} \concat \ev{t_2, \wt, c} \concat \ev{t_1, \wt, c} \concat \ev{t_1, \rd, y_{a_i}}  & \\
	\gamma_j &=& \ev{t_2, \wt, y_{b_j}} \concat \ev{t_1, \wt, y_{a_j}} \concat \ev{t_1, \rd, c} \concat \ev{t_2, \wt, c} \concat \ev{t_1, \wt, c} \concat \ev{t_1, \rd, y_{a_j}} & (j > i) \\
	\kappa' &=& \ev{t_2, \rd, u} \concat \ev{t_1, \wt, u} \concat \ev{t_1, \rd, c} \concat \ev{t_1, \rd, u} & \\
	\end{array}
	\end{align*}

	Here again, we observe that the new trace $\tr'$ is such that its thread-wise projections are the same as in $\tr$.
	Further, every read event reads-from the same corresponding events in both $\tr$ and $\tr'$.
	That is, $\tr \rfnovaleq \tr'$.
	Further, observe that the positions of $\theta_1$ and $\theta_2$ have been inverted.
	This trace $\tr'$ for a special case is shown in \figref{check-order-lb-reordering}.
}

This finishes our proof of \thmref{check-order-semantic-linear-space-lower-bound}.
\end{proof}

Let us now turn to the proof of \thmref{check-order-semantic-tradeoff}.
This time, we focus on the following language.
\begin{align*}
\langeq_{\text{pad},n} = \setpred{\vect{a}\#^n\vect{b}}{\vect{a}, \vect{b} \in \set{0,1}^n \text{ and } \vect{a} = \vect{b}}
\end{align*}

\begin{lemma}
For any streaming algorithm that recognizes $\langeq_{\text{pad},n}$ in time $T(n)$ and space $S(n)$, we have
$T(n) \cdot S(n) \in \Omega(n^2)$.
\end{lemma}

\begin{proof}
We first observe that the communication complexity of checking equality between two $n$-bit strings is 
$\Omega(n)$. 
Consider a Turing Machine $M$ that recognizes $\langeq_{\text{pad},n}$ in time $T(n)$ and space $S(n)$, by possibly going back and forth on the input tape.
Since $M$ takes $T(n)$ time, it must only traverse `across' 
the central padding `$\#^n$' atmost $\frac{T(n)}{n}$ times.
Since the space usage is $S(n)$, each time $M$ crosses the padding completely, it communicates
at most $S(n)$ bits across the padding.
Thus, the total number of bits that can be communicated is atmost $\frac{T(n) \cdot S(n)}{n}$
and thus we have $\frac{T(n) \cdot S(n)}{n} \in \Omega(n)$, giving us $T(n) \cdot S(n) \in \Omega(n^2)$.
\end{proof}

\begin{proof}[Proof Sketch for \thmref{check-order-semantic-tradeoff}]
The reduction in the proof of \thmref{check-order-semantic-linear-space-lower-bound} can be modified to prove
the time-space tradeoff.
This time, we can show a reduction from $T(n) \cdot S(n) \in \Omega(n^2)$ as we did in the proof of 
\thmref{check-order-semantic-linear-space-lower-bound}.
The only difference will be to add extra events in the run, corresponding to the padding string $\#^n$,
for which we can use  $n$ write events $\wt(d)$ in the fragment $\kappa$ in thread $t_1$ after the $\rd(u)$ event of $t_1$; here $d$ is a fresh memory location.
Observe that none of the $\wt(d)$ events are read by any read events.
Thus, the answer of $\checkOrder{\rfnovaleq}{\tr}{\theta_1}{\theta_2}$ does not get affected.
Besides the reduction itself is a streaming one pass  reduction that takes
$O(n)$ time and $O(1)$ space. 
\end{proof}

\subsection{Proof of \thmref{check-order-non-context-free}}

We will use the pumping lemma to establish \thmref{check-order-non-context-free}.
Recall that, a context-free language $L$ satisfies the following property:

\begin{align*}
\textsf{Pumpable}(L) \equiv & \text{ There is a } n \geq 1 \text{ s.t. forall } s \in L \text{ with } |s| > n, \text{ there are strings  } u,v,w,x,y \text{ s.t. } \\
& s = u{\concat}v{\concat}w{\concat}x{\concat}y, |v{\concat}w{\concat}x| \leq n, |v{\concat}x| \geq 1, \text{ and  for all } i\geq 0, u{\concat}v^i{\concat}w{\concat}x^i{\concat}y \in L
\end{align*}

We will show that $\textsf{Pumpable}(L^{c,d}_{\textsf{ordered}})$ does not hold:

\[
L^{c,d}_{\textsf{ordered}} = \setpred{\tr \in \alphabet^*c \alphabet^* d \alphabet^*}{c \text{ and } d \text{ are not causally concurrent in } w \text{ under } \rfnovaleq}
\]

To argue this, we will construct a class of strings $\set{s_n}_{n \geq 1}$
such that $s_n \in L^{c,d}_{\textsf{ordered}}$ (for all $p \geq 1$) and satisfies $|s_n| > n$,
but every way of splitting $s_n$ can be pumped down to a language that is not in 
$L^{c,d}_{\textsf{ordered}}$.

\myparagraph{Construction}{
Our construction of the string $s_n$ essentially mimics the construction 
for the proof of \thmref{check-order-semantic-linear-space-lower-bound}.
Thus, we will have 2 threads $t_1$ and $t_2$ in $s_n$, and the same set of variables as in this proof.
Without loss of generality, we will take $c = \ev{t_1, \rd, u}$ and $d = \ev{t_2, \wt, u}$.

More concretely, $s_n$ will be the run corresponding to the instance $0^n\#0^n \in \langeq_n$.
The idea behind the construction stems from the observation that the construction is tight,
and removing one or more events from the run results into one of: 
(a) non well-formed run, (b) the two focal events
$e_c = \ev{t_1, \rd, u}$ and $e_d = \ev{t_2, \wt, u}$ being unordered, or
(c) the focal events to completely vanish.
Any of these cases result in a string outside of $L^{c,d}_{\textsf{ordered}}$.

Formally, given $n \geq 1$, we construct the run $s_n = \pi \concat \kappa \concat \eta \concat \delta$
as shown in \figref{check-order-lb-trace},
where $\pi = \pi'_1  \pi_2 \ldots \pi_n$ and $\eta = \eta_1 \eta_2 \ldots \eta_n$ as described.
Here, $\pi'_1$ is the same string as $\pi_1$ but does not contain
the event $\ev{t_1, \wt, x_{\neg a_1}}$.
Now, we make the following observations;
we use the notation $\sigma \setminus e$ and $\sigma \setminus X$ to denote the
sequence obtained by removing from $\sigma$, the event $e$ and the set of events $X$ respectively.

\begin{claim}
\claimlabel{completely-in-one}
Let $X \subseteq \events{s_n}$ be a set of events such that $X \neq \emptyset$ and $|X| \leq n$
and $X$ contains events only from the segment $\pi \concat \kappa$.
The sequence $s_n \setminus X$ does not belong to $L^{c,d}_{\textsf{ordered}}$.
\end{claim}

\begin{proof}
First, if $X$ contains a write event but not an event that reads from it, then clearly $s_n \setminus X$ is not well-formed, and thus cannot belong to $L^{c,d}_{\textsf{ordered}}$.
In the rest of the proof, we assume that for every write event in $X$, the corresponding read event is also in $X$.

If $X$ contains the event $\ev{t_1, \wt, x_{a_1}}$, then $X$ must also contain $\ev{t_2, \rd, x_{b_1}}$
which is not in $\pi \concat \kappa$; so $X$ cannot contain this event.
If $X$ contains the event $\ev{t, \wt, y_{a_i}}$ for some $i \geq 2$, then it must also contain
$\ev{t_, \rd, y_{a_i}}$.
More importantly, the proof of \thmref{check-order-semantic-linear-space-lower-bound} argues that in fact
in this case the write events in $t_2$ can be carefully reordered so that the resulting reordering is $\rfnovaleq$ equivalent to $s_n \setminus X$ and thus $c$ and $d$ can be reordered.
For example, if $X$ contains $\ev{t_1, \wt, y_{a_1}}$, then the reordering that places $\ev{t_2, \rd, x_{b_1}}$ after $\ev{t_1, \wt, x_{a_1}}$
but places $\ev{t_2, \wt, c}\ev{t_2, \wt, y_{b_2}}\ev{t_2, \wt, c}$ after $\ev{t_2, \rd, x_{b_1}}$
but before $\ev{t_1, \wt, y_{a_3}}$ is the correct reordering that witnesses reordering of $c$ and $d$.
The same argument also shows why $X$ cannot contain any $\ev{t_1, \rd, y_{a_i}}$, $\wt{t_1, \wt, c}$ or $\wt{t_1,\rd, c}$.
Also, $\ev{t_1, \wt, u}$ cannot be in $X$; if so, then $\ev{t_1, \rd, u} \in X$ and thus $s_n \setminus X$ 
does not contain $c$.
\end{proof}

\begin{claim}
\claimlabel{completely-in-two}
Let $X \subseteq \events{s_n}$ be a set of events such that $X \neq \emptyset$ and $|X| \leq n$
and $X$ contains events only from the segment $\eta \concat \delta$.
The sequence $s_n \setminus X$ does not belong to $L^{c,d}_{\textsf{ordered}}$.
\end{claim}

\begin{proof}
If $X$ contains $\ev{t_2, \wt, u} = d$, then $s_n \setminus X$ trivially is not in $L^{c, d}_{\textsf{ordered}}$.
If $X$ contains $\ev{t_2, \rd, x_{b_1}}$, then there is no ordering from the event $\ev{t_1, \wt, x_{a_1}}$
to an event of thread $t_2$, as a result of which $d$ can be reordered before $c$ as show in  \figref{check-order-lb-reordering}.
If $X$ contains any $\ev{t_2, \wt, c}$ of $\kappa_i$, then we do not get an ordering from $\ev{t_1, \rd, y_{a_{i+1}}}$ to $\ev{t_2, \wt, y_{b_{i+1}}}$. As a result, we can reorder $c$ and $d$.
The same reasoning can be used to argue why $X$ cannot contain $\ev{t_2, \wt, y_{b_i}}$ for any $i\geq 2$.
\end{proof}

Now, we consider a case-by-case analysis of how the substrings $u, v, w, x, y$ of $s_n$ are picked subject to the constraints $s_n = u\concat v\concat w\concat x\concat y$, $|v\concat w\concat x| \leq n$, $|v \concat x| \geq 1$.

\begin{description}
	\item[Case $v \concat w \concat x \subseteq \pi\concat \kappa$ or $v \concat w \concat x \subseteq \eta\concat \delta$.]
	In this case, we can pump down (choose $i=0$) and the resulting string $s'_n = u\concat w\concat y \not\in L^{c, d}_{\textsf{ordered}}$ due to \claimref{completely-in-one} and \claimref{completely-in-two}.
	\item[Case $v\concat w\concat x$ spans $\pi \concat \kappa \concat \eta$.] 
	In this case, observe that there is no $j$ such that $v\concat w\concat x$ intersects with both $\pi_j$ and $\eta_j$ because of the restriction that $|v\concat w\concat x| \leq n$.
	Hence, choosing $i = 0$ (pumping down) again leads to a run $s'_n = u\concat w\concat y \not\in L^{c, d}_{\textsf{ordered}}$, and this can be established using arguments similar to the proofs of \claimref{completely-in-one} and \claimref{completely-in-two}.
\end{description}
}

\subsection{Proof of \thmref{quadratic-time-hardness}}
\applabel{ov-hardness}

The proof of this theorem can be proved in a similar manner
as an analogous result of~\cite{Mathur2020b} in the context of data race detection.
Given a set of events $X \subseteq \events{\tr}$,
a partial order $P \subseteq X \times X$ which totally orders
events of each thread, and a reads-from relation $RF : X \partialto X$
that maps each read event in $X$ to a write event in $X$ with the same variable,
the RF-Poset realizability problem for $(X, P, RF)$ asks
if there is a linearization of $P$ whose reads-from function matches $RF$.
The following is the statement of the analogous result in \cite{Mathur2020b}:

\begin{theorem}[Lemma 5.6 in \cite{Mathur2020b}]
\thmlabel{rf-poset}
Assuming SETH holds. RF-Poset realizability for posets with $n$ events 
cannot be solved in time $O(n^{2-\epsilon})$ for every $\epsilon > 0$,
even for inputs with $2$ threads and $7$ variables.
\end{theorem}

The proof of the above statement in fact constructs a simple RF-Poset
with only two threads, and is such that
the RF-Poset can be translated into a simple linearizartion (that first linearizes the events
of the first thread, followed by the events of the second thread).
This means that the RF-Poset realizability holds iff the linearization
can be reordered so that the two focal events (read and write of $z$) can be flipped.
Since our proof is not a direct reduction from RF-Poset realizability, we need
to prove our result separately.
However, since most parts of the proof are identical, we skip the entire construction and only outline the high level details, and highlight the low level details about where our construction differs.

Consider the sequence of threads $\tau_A$ and $\tau_B$ constructed
by \thmref{rf-poset} (dependin upon the two sequence of 
Boolean vectors $A$ and $B$ given as part of the OV instance).
We will use two fresh variables $z$ and $u$.
Let $\tau'_A$ be the sequence obtained by 
(a) inserting a 
$\wt(z)$ event after the event $\wt^{a_1}_1(x_1)$,
and 
(b) inserting a $\wt(u)$ event before the event
$\rd^{a_{n/2}}_1(x_2)$ in $\tau_A$.
Likewise, $\tau'_B$ be the sequence obtained by 
(a) inserting 
the event $\rd(z)$ before the event $\rd^{b_1}_1(x_1)$,
and
(b) inserting $\wt(u)$ event after $\wt^{b_{n/2}}_1(x_2)$ in $\tau_B$.
Observe that when the entire $\tau'_A$ appears after $\tau'_B$, we have
$(\wt^{a_1}_1(x_1), \rd^{b_1}_1(x_1)) \in \rdf{}$ and this imposes one of the desired orderings of the construction
in~\cite{Mathur2020b}.
However the other ordering is not imposed, but this can be imposed when we instead ask the check order question.

Formally, let $\sigma = \tau'_A \tau'_B$.
We claim that in $\sigma$, the two events $e_1 = \ev{\tau_A, \wt(u)}$
and $e_2 = \ev{\tau_B, \wt(u)}$ are reorderable under $\rfnovaleq$
iff the partial order $P$ constructed in \cite{Mathur2020b} is realizable.
For the forward direction, consider the witness reordering $\sigma'$ in which
$e_1$ and $e_2$ are reordered. 
Consider the run $\rho$ obtained by removing from $\sigma'$ the events of variables $z$ and $u$.
Observe that this is a witness to the realizability of the RF poset constructed in \cite{Mathur2020b} because $\sigma'$ is $\rdf{}$-equivalent to $\sigma$ and further $e_2$ appears before $e_1$ in $\sigma'$,
and thus $\sigma'$ satisfies all the constraints of the RF poset instance.
For the reverse direction, suppose that the RF poset instance of \cite{Mathur2020b} is realizable,
then, consider the trace that realizes the poset, say $\rho$.
In $\rho$, we can add the two write events on $u$ and the read-write pair on $z$
as described above.
We remark that the resulting trace $\sigma'$ is $\rdf{}$-equivalent to 
$\sigma$ --- the order between all thread-wise events is preserved as it was also a constraint in the poset.
Further, the reads-from of all the events on $x_1, x_2, \ldots, x_7$ is preserved since this $\rho$ is an instance of RF-poset and adding events on $z$ and $u$ doesnt change any other variables' reads-from.
Since the ordering on $f_1 = \wt^{a_1}_1(x_1)$ and $f_2 = \rd^{a_{n/2}}_1(x_1)$ was preserved (as $f_1 <^\rho f_2$) because of the constraint in 
the poset, it implies that the event immediately preceding $f_1$ (namely $g_1 = \ev{\tau_A, \wt(z)}$) and
immediately succeeding $f_2$ (namely $g_1 = \ev{\tau_B, \rd(z)})$) must also be preserved as 
$g_1 <^{\sigma'} g_2$, and since these two are the only two events on $z$, they are also a RF pair.
Finally, the two focal events get reversed because of the other ordering in the poset.
his means $\sigma'$ thus obtained is a witness of the reoderability of $e_1$ and $e_2$ in $\sigma$.

\subsection{Proof of \thmref{check-order-rf-time-upper-bound}}
\applabel{poly-algo}

\myparagraph{Overview}{
	Our proof is derived from similar statements in~\cite{Gibbons1994} and~\cite{Mathur2020b}.
	The idea behind the algorithm is to search for a path over an 
	abstract representation of the search space. 
	The search space will be represented as a graph, 
	also called the `frontier graph'~\cite{Gibbons1994}. 
	The nodes of this graph are subsets of $\events{\tr}$ 
	which are downward closed subsets of $(\po{\tr} \cup \rdf{\tr} \cup (e, f))^*$, 
	assuming this is a partial order (if not, we can directly say NO). 
	The edges represent `extensions' by one event. 
	The existence of a witnessing $\rdf{}$-equivalent run $\tr'$ 
	can then be checked by checking if there is a path from 
	the node representing the empty subset to the (unique) 
	node corresponding to the set $\events{\tr}$. 
	The proof of correctness will argue the correctness of this abstraction.
}

\begin{definition}[Ideal]
Let $P$ be a partial order on $\events{\tr}$ such that 
$(\po{\tr} \cup \rdf{\tr}) \subseteq P$.  
An ideal $X$ of $P$ is a subset of $\events{\tr}$ such that for every 
pair of events $e, f \in \events{\tr}$ such that $(e, f) \in P$, if $f \in X$, then $e \in X$.
\end{definition}

\begin{remark}
Observe that an ideal contains no information about relative ordering (linearization) of events, 
but just the set of events. 
Further, the empty set $\emptyset$ is also an ideal of $P$.
\end{remark}

\begin{definition}[Extension]
Given two ideals $X_1, X_2$ of a partial order $P$, 
we say that $X_2$ extends $X_1$ if there is an event $e$ such that
\begin{enumerate}
	\item $X_2 = X_1 \uplus \set{e}$,
	\item if $e$ is a write event (on variable $x$), and if there is 
	another write event $e' \neq e$ on variable $x$ in $X_1$, 
	then we must have $r \in X_1$ for every $r$ such that $\rdf{\tr}(r) = e'$.
\end{enumerate}
We also say that the event $e$ extends $X_1$ to $X_2$.
\end{definition}

\begin{proposition}
Observe that for an ideal $X$, there are at most $|\threads|$ ideals
 that are extensions of $X$.
\end{proposition}

\begin{remark}
The notion of extension above is `conservative' in the sense that it 
is aware of the future. 
In particular, our goal is to incrementally construct extensions, 
all the way until we reach the set $\events{\tr}$. 
As a result, some write $e'$ is not allowed to be added to an 
ideal because that will disable a future read $r$. 
If on the other hand, our goal was to not reach the ideal $\events{\tr}$ 
but only to reach a subset that does not contain the read $r$, then $e'$ could be added.
\end{remark}

\begin{definition}[Frontier Graph]
\deflabel{frontier-graph}
Let $P$ be a partial order on $\events{\tr}$. 
The frontier graph $G_P= (V_P, E_P)$ of $P$ is a directed graph such that:
\begin{enumerate}
\item $V_P$ is the set of ideals of $P$
\item $(X_1, X_2) \in E_P$ if $X_2$ is an extension of $X_1$. 
\end{enumerate}
Edges in this graph are implicitly labeled by events. 
For the edge $(X_1, X_2) \in E_P$ such that $\set{e} = X_2 - X_1$, 
the label is $e$ and we will denote this as $X_1 \rightarrow_{e} X_2$.
\end{definition}

\begin{remark}
The size of $G_P$ is $O(|\threads| \cdot |\tr|^{|\threads|})$ 
because there are at most $|\tr|^{|\threads|}$ nodes, each having at most 
$|\threads|$ outgoing edges. 
The time to construct it is
$O(|\threads| \cdot |\tr|^{|\threads|})$
because for every edge you spend time $O(|\tr|)$ to check if the edge corresponds to an extension.
\end{remark}

\begin{definition}[Respecting a partial order]
Let $S$ be a set and let $P$ be a partial order over $S$. 
Let $S' \subseteq S$ and let $\rho$ be a permutation of elements of $S'$.
We say that $\rho$ respects $P$ if for every $(e, f) \in P$, 
if $f \in S'$ then $e \in S'$ and further, $e$ appears before $f$ in $\rho$.
\end{definition}

\begin{definition}[$(\po{}, \rdf{})$-preserving subrun]
Given a run $\rho$, we say that $\rho$ is $(\po{}, \rdf{})$-preserving subrun of $\tr$ if
\begin{enumerate} 
\item $\events{\rho} \subseteq \events{\tr}$,
\item $\events{\rho}$ respects $\po{\tr} \cup \rdf{\tr}$
\item For every read event $r \in \events{\tr}$, with $w = \rdf{\tr}(r)$, 
there is no other write event $w' \neq w$ on the same variable as $r$ 
that appears between $w$ and $r$ in $\rho$.
\end{enumerate}
\end{definition}

\begin{remark}
If $\rho$ is a $(\po{}, \rdf{})$-preserving-subrun of $\tr$, and 
further if $\events{\rho} = \events{\tr}$, then $\rho$ is $\rdf{}$-equivalent to $\tr$.
\end{remark}
\begin{remark}
Not every $(\po{}, \rdf{})$-preserving-subrun of $\tr$ 
can be extended to an $\rdf{}$-equivalent run. 
In other words, there is a $\tr$ and a $(\po{}, \rdf{})$-preserving-subrun $\rho$ 
of $\tr$ for which every extension
$\rho \cdot \gamma$ is not $\rdf{}$-equivalent to $\tr$ 
\end{remark}

\begin{lemma}
\lemlabel{frontier-graph-reachability-preserves-po-rf}
Let $P$ be a partial order over $\events{\tr}$ and let 
$G_P = (V_P, E_P)$ be the frontier graph of $P$. 
For every ideal $X \in V_\tr$ such that $X$ is reachable from $\emptyset \in V_P$, 
there is a run $\rho$ such that $\rho$ is a $(\po{}, \rdf{})$-preserving-subrun of $\tr$ and $\rho$ 
respects $P$ (i.e., for every two events $e, f \in \events{\rho}$ such that $(e, f) \in P$, 
we also have that $e$ appears before $f$ in the run $\rho$).
\end{lemma}

\begin{proof}
We prove the following stronger statement that, 
in fact, every path $X_0, X_1, \ldots, X_n$ from $\emptyset$ to $X$ labeled 
$\sigma_X = e_1e_2\ldots e_n$ is indeed a $(\po{}, \rdf{})$-preserving-subrun of $\tr$ that respects $P$. 
Here, where $e_i$ is the $i^\text{th}$ event in the path, i.e., $X_i \rightarrow_{e_{i+1}} X_{i+1}$. 
Observe that $X = \{e_1, \ldots, e_n\}$.

We prove this by inducting on the length of $\tr_X$ (alternatively on the size of $X$).

{\bf Base case}: In this case $X = \emptyset$ and $\tr_X = \epsilon$ and the statement holds vacuously.

{\bf Inductive case}: Suppose we have that for every labeled path of length $\leq i$, the statement holds true. Consider now a labeled path $\tr_X$ of length $i+1$. Let $e$ be the last event in this path. We consider two cases based on the type of $e$.

\begin{description}
\item[$e = \ev{t, \wt(x)}$.] That is, $e$ is a write event. 
Suppose on the contrary that $\tr_X$ is not $(\po{}, \rdf{})$-preserving-subrun of
$\tr$ that respects $P$. 
Consider the penultimate prefix $\tr_{Y}$ such that $\tr_X = \tr_{Y} \circ e$; 
we will also use $Y = X \setminus \set{e}$. 
By inductive hypothesis, $\tr_Y$ is a $(\po{}, \rdf{})$-preserving-subrun of $\tr$ 
that also respects $P$. T
hen, there is an event $f \in Y$ such that $(e, f) \in \po{\tr} \cup \rdf{\tr} \cup P$ . 
But this means that $Y$ is not an ideal of $P$, a contradiction.
        
\item[$e = \ev{t, \rd(x)}$.] That is, $e$ is a read event. 
Suppose on the contrary that $\tr_X$ is not $(\po{}, \rdf{})$-preserving-subrun 
of $\tr$ that respects $P$. 
Then either we have that the penultimate set $Y = X \setminus \set{e}$ is not an ideal, 
as in the previous case giving a contradiction, 
or we have the more interesting case where $\tr_X$ is not a 
$(\po{}, \rdf{})$-preserving-subrun of $\tr$ becuase there is a write 
$f' \neq f$ (where $f = \rdf{\tr}(e)$) such that $f'$ appears later than 
$f$ in $\tr_Y$. 
Consider the largest prefix $\tr_Z$ of $\tr_Y$ that does not contain $f'$. 
Also, let $\tr_U$ be the immediately longer path, i.e., 
$\tr_U = \tr_Z \circ f'$. 
We will use $Z$  and  $U$ to denote the ideals r
eached after taking the paths $\tr_Z$ and $\tr_U$ respectively. 
Observe that $U = Z \uplus \set{f'}$, $e \not\in Z$, $f = \rdf{\tr}(e) \in Z$ and 
$Z \rightarrow_{f'} U$. This contradicts that $U$ is an extension of $Z$.   
\end{description}
 \end{proof}

\begin{remark} 
The graph $G_P$ defined above does not capture all $(\po{}, \rdf{})$-preserving-subruns of $\tr$. 
It only captures those that are prefixes of some run that is $\rdf{}$-equivalent to $\tr$. 
Thus the converse of \lemref{frontier-graph-reachability-preserves-po-rf}, 
as stated, is not true. The following lemma is true though (since it talks about $\rdf{}$-equivalent executions).
\end{remark}

\begin{lemma}
\lemlabel{converse-frontier-graph}
Let $P$ be a partial order over $\events{\tr}$ and let 
$G_P = (V_P, E_P)$ be the frontier graph of $P$. 
Let $\rho = f_1f_2\ldots f_N$ be a run that is $\rdf{}$-equivalent to $\tr$ such that $e_2$ 
appears before $e_1$ in $\rho$. Then there is a path in $G_P$ that is labeled with $\rho$.
\end{lemma}

\begin{proof}
We establish that $\rho$ is a path in the graph $G_P$ by inductively establishing that $X_i$ is a node of $G_P$ and $X_{i-1} \rightarrow_{f_i} X_i$ is an edge of $G_P$ for every $i \geq 1$.

{\bf Base Case (i=1).} Since $\rho$ is rf-equivalent to $\tr$, 
$f_1$ must be the first event of its own thread and further it must be a write event. 
Also, $f_1$ cannot be the focal event $e_1$. 
Thus, the set $X_1 = \set{f_1}$ is indeed an ideal of $P$ and thus $X_1 \in V_P$ 
is a node of the graph $G_P$. 
Next, since there is no other write event in $X_0 = \emptyset$, $f_1$ trivially extends $X_0$ to $X_1$. 
Thus, $X_0 \rightarrow_{f_1} X_1$ is an edge in $G_P$.

{\bf Inductive Case.} Suppose that $X_i$ is an ideal of $G_P$. 
Consider the set $X_i = X_{i-1} \uplus \set{f_i}$.  
Since $\rho$ is $\rdf{}$-equivalent to $\tr$, all events $\po{}$-ordered before 
$f_i$ must be in the prefix $f_1, \ldots, f_{i-1}$ and thus in the set $X_{i-1}$. 
Likewise, if $f_i$ is a read event, then $\rdf{\tr}(f_i)$ is also in $X_{i-1}$. 
Finally since $\rho$ does not execute $e_1$ before executing $e_1$, 
if $f_i = e_1$, then $e_1 \in X_{i-1}$. 
In other words, $X_i$ is an ideal of $P$ and thus a node in $G_P$.
Next, consider the case when $f_i$ and suppose on thae contrary that there is a write event 
$w \neq f_i$ and a read event $r$ with $\rdf{\tr}(r) = w$ 
such that $w \in X_{i-1}, r \notin X_{i-1}$. 
This will contradict that $\rho$ is rf-equivalent to $\tr$ as $f_i$ appears between 
$w$ and $r$ in $\rho$. 
Hence, $f_i$ extends $X_{i-1}$ to $X_i$. As a result, $X_{i-1} \rightarrow_{f_i} X_i$ is an edge in $G_P$  
\end{proof}      

\subsubsection{Algorithm for solving $\checkOrder{\rfnovaleq}{\cdot}{\cdot}{\cdot}$}

\begin{algorithm}[t]
    \KwIn{Run $\tr\in\alphabet^*$, events $e_1, e_2 \in \events{\tr}$ such that $e_1$ appears before $e_2$ in $\tr$}
    \KwOut{YES iff there is a run $\tr'$ which is $\rdf{}$-equivalent to  $\tr$ in which $e_2$ 
    appears before $e_1$.}
    
    Construct the transitive reduction of the quasi order $P = (\po{\tr} \cup \rdf{\tr} \cup \set{e_2, e_1})$\;
    \lIf{$P$ has a cycle}{\Return NO}
    Construct the frontier graph $G_P = (V_P, E_P)$ as in \defref{frontier-graph} \;
    \If{$(\emptyset, \events{\tr}) \in E_P^+$}{
    	\Return YES
    }
    \Else{
    	\Return NO
    }
    \caption{Polynomial time algorithm for Checking Causal Orderedness under $\rfnovaleq$}
    \algolabel{poly-check-order-rf}
\end{algorithm}

\begin{theorem}
\algoref{poly-check-order-rf} runs in time $O(|\threads|\cdot|\tr|^{|\threads| + 1})$ for an input run $\tr$.
Further, \algoref{poly-check-order-rf} returns YES iff $\neg \checkOrder{\rfnovaleq}{\tr}{e_1}{e_2}$
\end{theorem}

\begin{proof}
The time spent in constructing $P$ and checking cycles in it is $O(|\tr|)$.
 The time taken to build $G_P$
is $O(|\threads|\cdot|\tr|^{|\threads| + 1})$ and the time taken to check reachability in $G_P$
is $O(|\threads|\cdot|\tr|^{|\threads|})$

Let us now argue correcntess. \\
($\Rightarrow$) Suppose the algorithm says YES. Then $P$ is a partial order. 
Further the node $X = \events{\tr}$ is reachable from $\emptyset$ in $G_P$. 
Then by \lemref{frontier-graph-reachability-preserves-po-rf}, there is a run $\rho$ that is $\rdf{}$-equivalent to $\tr$ and flips $e_1$ and $e_2$.

($\Leftarrow$) Suppose there is an $\rdf{}$-equivalent run $\rho = f_1f_2\ldots f_{|\tr|}$ of 
$\tr$ such that $e_2$ appears before $e_1$ in $\rho$. First, observe that $\rho$ respects $\po{\tr} \cup \rdf{\tr} \cup (e_2, e_1)$ and thus $P$ must be acyclic. 
By \lemref{converse-frontier-graph}, 
we must have that $X = \events{\tr} = \set{f_1, \ldots, f_{|\tr|}}$ is reachable from $\emptyset$ in $G_P$. Thus the algorithm returns YES.
\end{proof}

\subsection{Proof of \thmref{check-order-rf-space-upper-bound}}
\applabel{linear-space-upper-bound}

\begin{proof}
At a high level, the algorithm enumerates permutations of
the given run $w$, one at a time, and checks if each such permutation is
$\rdf{}$-equivalent to $w$ and contains $e$ and $f$ appearing in the reverse order.
Further, this task can be performed in total additional space (besides the input $w$)
that is linear in $|w|$. 
In the following, we explain how these two tasks can be performed
by a deterministic Turing machine $M$ with a linearly bounded work tape.

\myparagraph{Generating permutations in linear space}{
	Given the run $w$, $M$ first stores the lexicographically minimum
	permutation $w_\textsf{min}$ of $w$ into a separate worktape.
	This can be done by maintaining a counter that counts the number of occurences
	of each symbol in $\alphabet$ and copying the required copies of each letter
	to the work tape.
	Next, given the current contents $u$ of the work tape,
	$M$ identifies the longest non-increasing  suffix $v$ of $u$,
	identifies the symbol $a$ right before it (and thus $u = u_1\cdot a \cdot v$ for some $u$).
	Next, $M$ identifies the rightmost symbol $b$ of $v$ that is lexicographically larger than $a$
	(and thus $v = v_1 \cdot f \cdot v_2$), swaps $a$ and $b$ to obtain
	the string $x = u_1 \cdot b \cdot v_1 \cdot a \cdot v_2$, and finally reverses the
	suffix $v_1 \cdot a \cdot v_2$ to obtain the next permutation $u' = u_1 \cdot b \cdot v_2^\textsf{rev} \cdot a \cdot v_1^\textsf{rev}$.
	Clearly, all this operations can be performed in place on the worktape.
}

\myparagraph{Checking each permutation}{
	Given a string $u$ on the worktape, $M$ also needs to check
	if $u \rfnovaleq w$ and whether the order of occurence of $e$ and $f$ 
	is different in $u$ and $w$.
	This can also be checked using no additional space as follows.
	First, $M$ checks if $u$ and $w$ have the same program order
	(i.e., $\po{w} = \po{u}$) by making $|\threads|$ many passes over $u$
	and $w$, one for each $t \in \threads$, checking if the projection
	of $u$ and $w$ to $t$ match.
	Next, $M$ checks if $\rdf{w} = \rdf{u}$ by making a separate pass for each read event $r$,
	and verifying if the corresponding last write event $w$ on the same variable as $r$
	and occuring before $r$ is performed by the same thread, and at the same local index in that thread.
	Finally, checking if $e$ and $f$ appear reversed corresponds to checking if the
	$i^\text{th}$ event of thread $t$ appears after the $j^\text{th}$ event of thread $t'$ in $u$,
	where $t, t'$ are respectively the threads performing $e, f$, and
	$i, j$ are their respective local indices in $t, t'$.
	The check of whether some two occurences of $c$ and $d$ can be reversed is analgous.
	$M$ performs each of these checks in place on the worktape, using only
	$\log |w|$ many extra bits to keep track of counters and local indices.
}
\end{proof}


\section{Proofs from Section \secref{geq}}
\applabel{geq}

\subsection{Proof of \thmref{wc-hardness}}

\begin{proof} 
Let us focus only on words $w$ where $w$ is of the form $a b^n c^m$. Given the relation $I$, we claim that $a$ and last $c$ are concurrent iff $n < m$, and ordered otherwise. It is clear that if $m \le n$, then $c$ can be reordered to before $a$ using a swap sequence like this:
\[a b^n c^m \to a bc b^{n-1} c^{m-1} \to \dots \to a (bc)^m b^{n-m} \to bc a (bc)^{m-1} b^{m-1} \to \dots \to (bc)^m a b^{m-1}\]
where each arrow in the first group corresponds to a $(b,bc)$ swap and each arrow in the second group corresponds to a an $(a, bc)$ swap. 
Let us argue why otherwise, the last $c$ is ordered wrt $a$. The only way that the last occurrence of $c$ can be reordered against $a$ is that all occurrences of $c$ have already been moved behind $a$; since occurrences of $c$ do not commute against each other. Every time an occurrence of $c$ is reordered with $a$, it must be through a swap $\dots a bc \to bc a$, because those are the only elements of $I$ that involve an $a$ and a $c$ on the two sides of a swap. This swap {\em consumes} one $b$, and this $b$ cannot move back, unless it moves back together with a $c$, which would be counterproductive. Therefore, we need at least as many $b$'s as $c$'s to swap all the $c$'s behind the $a$. 

A monitor, therefore, should be able to distinguish the two sets of strings $a b^n c^m$ where $n < m$ and where $n \ge m$ from each other. But, this involves counting and therefore is not regular. 
\end{proof}

\subsection{Proof of Theorem \ref{thm:sgrain}}
\begin{proof}
\noindent $\Leftarrow$ direction: If $gg' \not \equiv_{\rdf{}} g'g$ for some pair of grains, then it is straightforward to see that if they are consecutive and swapped, the soundness will be violated. For the additional conditions, the  grains from the run in Example \ref{ex:dreads} satisfy $gg' \equiv_{\rdf{}} g'g$ but violate the extra condition on the grains for the theorem and as argued the corresponding commutativity relation is not sound in the context of the run.

\noindent $\Rightarrow$ direction: We have to show that if $[w]_G \not \subseteq [w]_{\rdf{}}$, then at least one of the conditions of the theorem is violated. Assume there exists $u \in [w]_G$ where $u \not \in [w]_{\rdf{}}$. The latter can happen only if in $u$, either the program order or the reads-from relation are changed compared to $w$. Since $gg' \equiv_G g'g$ implies that program order is never changed (and any standard Mazurkiewicz swap preserves program order), the only point of change can be in the reads-from relation. 

We prove,  by contradiction, that the reads-from relation cannot change. Let us consider the sequence of swaps that would get us from $w$ to $u$:
\[w \stackrel{s_1}{\to} v_1 \stackrel{s_2}{\to} v_2 \stackrel{s_3}{\to} \dots  \stackrel{s_n}{\to} u\]

Let us assume $v_i$ is the first word in the sequence such that $v_i \not \equiv_{\rdf{}} w$, and as such $v_i \not \equiv_{\rdf{}} v_{i-1}$. Therefore, the swap $s_i$ is to blame. By definition, this cannot be a Mazurkiewicz swap. Therefore, it is a swap of two grains $g$ and $g'$ where $(g,g') \in I_G$. Note that any pair of $\wt(x)$ and $\rd(x)$ that are both on one side of this swap will be unaffected by the swap. Therefore, a change in the reads-from relation must be of one of the following forms, for a given variable $x$:
\begin{itemize}
\item $\wt(x)$ is before $gg'$ in $v_{i-1}$ and the corresponding $\rd(x)$ is somewhere in $gg'$. A swap can brake such a relation only if $\rd(x)$ is in $g$, there are no other writes to $x$ before $\rd(x)$ in $g$, and $g'$ contains a write to $x$. This is however in violation of the condition stated in the theorem. 
\item $\wt(x)$ is before $gg'$ in $v_{i-1}$ and the corresponding $\rd(x)$ is after $gg'$: This means that $gg'$ contains no write to variable $x$, and therefore, the swap cannot affect the fact that $\rd(x)$ reads from $\wt(x)$.    
\item $\wt(x)$ is in $gg'$ in $v_{i-1}$ and the corresponding $\rd(x)$ is also in $gg'$: This would contradict $gg' \equiv_{\rdf{}} g'g$. 
\item $\wt(x)$ is in $gg'$ in $v_{i-1}$ and the corresponding $\rd(x)$ is after $gg'$. This is in violation of the condition stated in the theorem. 
\end{itemize}
\end{proof}


\section{Proofs of Section \ref{sec:gmonitor}}

\subsection{Proof of Theorem \ref{thm:monitor-sound}}
\seclabel{app-transitions-grains}

The proof sketch we provided in the text, through a two-pass construction is formal enough to convince a reader with a good command of automata theory. Alternatively, one can give the full construction to this monitor as follows.


\begin{figure}[t]
\fbox{
\parbox{1.05\textwidth}{
\begin{center}
\begin{tabular}{lcl}
State & Event & State Update \\ \hline \hline
$\langle -, E, g_\emptyset, i, C \rangle$ & $e \in \Sigma$ & $\langle -, E,  g_\emptyset, i, C\rangle$ \\
$\langle -, E,  g_\emptyset, i, C \rangle$ & $\focalEv_1$ & $\langle \focalEv_1, E,  g_\emptyset, i, C\rangle$ \\
 \hline
$\langle \focalEv_1, E, g, i, C \rangle$ & $e= \evt{\wt(x)}{i}$ & $\langle \focalEv_1 -, \{\mathit{grain}(e)\}, g, \mathit{false}, C[x \mapsto \mathit{nondet}]\rangle$ \\
$\langle \focalEv_1, E, g, i, C \rangle$ & $e= \evt{\rd(x)}{i}$ & $\langle \focalEv_1 -, \{\mathit{grain}(e)\}, g, \mathit{false}, C\rangle$ \\ \hline
$\langle \focalEv_1, E, g, i, C \rangle$ & $\rhd$ & $\langle \focalEv_1 -,  \emptyset , g_\emptyset, \mathit{true}, C\rangle$ \\

$\langle \focalEv_1 -, E, g, \mathit{true}, C \rangle$ & $e=\evt{\wt(x)}{i}$ & $\langle \focalEv_1 -,  E , \mathit{update}(g,e), \mathit{true} , C \rangle$ \\
$\langle \focalEv_1 -, E, g, \mathit{true}, C \rangle$ & $e=\evt{\rd(x)}{i}$ & $\langle \focalEv_1 -,  E , \mathit{update}(g,e), \mathit{true}, C \rangle$ \\

$\langle \focalEv_1 -, E, g, \mathit{true}, C \rangle$ & $\lhd$ & $\langle \focalEv_1 -,  E \odot \mathit{ND}(g) , g_\emptyset, \mathit{false}, C \otimes \mathit{ND}(g)\rangle$ \\ \hline

$\langle \focalEv_1 -, E, g, \mathit{false}, C \rangle$ & $e=\evt{\wt(x)}{i}$ & $\langle \focalEv_1 -,  E \odot \mathit{grain}(e) , g, \mathit{false}, C[x \mapsto \mathit{nondet}] \rangle$ \\ 
$\langle \focalEv_1 -, E, g, \mathit{false}, C[x \mapsto \mathit{false}]  \rangle$ & $e=\evt{\rd(x)}{i}$ & $\langle \focalEv_1 -,  E \odot \mathit{grain}(e) , g, \mathit{false}, C  \rangle$ \\ 
$\langle \focalEv_1 -, E, g, i, C \rangle$ & $\focalEv_2$ & $\langle \focalEv_1 - \focalEv_2, E, g, i, C\rangle$ \\ 
\hline

$\langle \focalEv_1 - \focalEv_2, E, g, i, C \rangle$ & $e \in \Sigma$ & $\langle \mathit{Ord}(E,\mathit{ND}(\mathit{grain}(e))) 
, E, g, i, C \otimes \mathit{ND}(\mathit{grain}(e))\rangle$ \\ \hline
$\langle \focalEv_1 - \focalEv_2, E, g, i, C \rangle$ & $\rhd$ & $\langle \focalEv_1 - \focalEv_2 \rhd, E, g_\emptyset, \mathit{true}, C \rangle$ \\ 
$\langle \focalEv_1 - \focalEv_2 \rhd, E, g, i, C \rangle$ & $e=\evt{\wt(x)}{i}$ & $\langle \focalEv_1 - \focalEv_2 \rhd, E, \mathit{update}(g,e), \mathit{true}, C \rangle$ \\ 
$\langle \focalEv_1 - \focalEv_2 \rhd, E, g, i, C[x \mapsto \mathit{false}] \rangle$ & $e=\evt{\rd(x)}{i}$ & $\langle \focalEv_1 - \focalEv_2 \rhd, E, \mathit{update}(g,e), \mathit{true}, C \rangle$ \\ 
$\langle \focalEv_1 - \focalEv_2 \rhd, E, g, i, C \rangle$ & $\lhd$ & $\langle  \mathit{Ord}(E,\mathit{ND}(g)), E, g_\emptyset, \mathit{false}, C \otimes \mathit{ND}(g)\rangle$ \\ \hline
$\langle \mathit{false}, E, g, i, C \rangle$ & $e=\evt{\wt(x)}{i}$ &  $\langle \mathit{false}, E, g, i, C\rangle$\\
$\langle \mathit{false}, E, g, i, C[x \mapsto \mathit{false}] \rangle$ & $e=\evt{\rd(x)}{i}$ &  $\langle \mathit{false}, E, g, i, C\rangle$\\
\hline
\end{tabular}%
\begin{align*}
\mathit{update}(g,\evt{\opfont{op}(x)}{i}) &= \left\{ \begin{array}{ll}
                                      g.E = g.E \cup \{e\}, g.V = g.V \cup \{x\} & \text{if } (\opfont{op}=\rd \wedge \wt(x) \not \in g.E) \\
                                      g.E = g.E \cup \{e\} & \text{owise} 
                                    \end{array}\right. \\
\end{align*}%
\vspace{-30pt}%
\begin{align*}
\mathit{Ord}(E,g) &\iff \exists g' \in E: \ \depend{g,g'} & 
C \otimes g &= C - \{x| \evt{\wt(x)}{i} \in g.E\} \cup g.V  \\
E \odot g &= \left\{ \begin{array}{ll}
                                     E \cup g  & \text{if }  \mathit{Ord}(E,g) \\
                                     E & \text{owise} 
                                    \end{array}\right. & 
\mathit{ND}(g) &= \langle g.E, g.V \cup \{x|\ \evt{\wt(x)}{i} \in g.E \wedge \text{ nondet}\} \rangle 
\end{align*}%
\end{center}}}\vspace{-5pt}
\caption{Grain Concurrency Monitor: The monitor starts in $\langle -, \emptyset, g_\emptyset, \mathit{false}, \emptyset \rangle$ and accepts if it is in a state $\langle \mathit{false}, \dots \rangle$ once the input run is read. $g_\emptyset$ corresponds to an empty grain signature. $\mathit{grain}(e)$ is syntactic sugar for $\mathit{update}(g_\emptyset,e)$. With $\mathit{ND}$, the monitor nondeterministically decides for which of the writes that appear in the grain, all the read operations are also assumed to be in the grain.}\vspace{-10pt}
\label{fig:gcm}
\end{figure}

The state of the grain monitor is conceptually the same as the trace concurrency monitor, except that $E$ is no longer a set of {\em events}, but rather a set of {\em grain signatures}. 
Additionally, the state includes one grain signature $g$ that maintains information about the current grain,  a flag $i: Bool$ that is true whenever the monitor is in the middle of reading a grain, and a flag $C: \vars \to Bool$ which is the set of variables wrt which some previously closed grain was (nondeterministically) assumed to satisfy the part of condition \ref{eq:full} which assumes all the reads, that correspond to a write in some grain, also belong to the grain; therefore, the monitor does not expect to see any reads on these variables before it sees a write. 
The monitor makes a nondeterministic guess in the second component of a grain signature, for each variable on whether all the reads of a given write are included in the grain, and later checks  the (right) context of the grain to verify this guess. In effect, the monitor is able to make both choices about the grain and therefore, try its luck with both versions of the (sound) commutativity relations that are implied by the choice. A wrong choice will later result in a halt. 
\figref{gcm} lists the transitions of the monitor, and the definitions of the operators used.

\begin{proof}
The full proof of why the detailed monitor is correct is long and tedious, through case analysis. We sketch the high level idea behind the most interesting case of the proof. The proof is by induction on the length of the portion of the input run so far processed by the monitor. We assume that the monitor has processed the prefix $\sigma$ (which includes $e$ but not $e'$) and is about to read the next entity $a$, and assume that it satisfies the following induction hypothesis:

\begin{quote}
\em
There exists a resolution of nondeterministic choices such that: 
\begin{itemize}
\item For any entity $b$ (event or grain) in $\sigma$:  
\[ (\mathit{grain}(b) \subseteq E \iff b \text{ is ordered wrt } e)\]
\item $E \subseteq \{\mathit{grain}(d)|\ d \in \sigma\}$
\item For every entity in $E$, the monitored has correctly guessed the status of condition \ref{eq:full}.  
\end{itemize}
and any error in underestimating $E$ will eventually result in a halt.
\end{quote}

We split the proof on whether $a$ is concurrent or ordered wrt $e$: 

\begin{itemize}
\item Ordered: Consider the definition of  $[\sigma a]_G$. It is simpler to move to the corresponding grain monoid and consider the equivalence class of $[h_G(\sigma a)]_{\widehat{I}_G}$. By definition, $e$ and $a$ are always ordered the same way in $[h_G(\sigma a)]_{\widehat{I}_G}$. 
 
In the grain monoid, this implies the existence of a path in the partial order representation of the class. Let this path be:

\[e \to a_1 \to \dots \to a_m \to a\]  

where each $a_i$ is either in $\Sigma$ or in $\Sigma_G$. The induction hypothesis implies that:

\[\{e\} \cup \mathit{events}(a_1) \cup \dots \mathit{events}(a_m) \subseteq E\]

Since $a_{m} \to a$, we know the two items do not commute. The reasons for this could be that:
\begin{itemize}
\item They share a thread: then the monitor correctly decides that the new event $a$ is ordered and adds its signature to $E$.
\item They share a variable $x$, at least one writes to $x$, and condition $\ref{eq:full}$ does not hold. Since the grain summaries of $a_m$ in $E$ is correctly computed (and guessed) by the induction hypothesis, and the correct guess for $a$ is an option, we correctly see $a$ as ordered with respect to the summary $E$.
\end{itemize}
As such, the monitor adds $a$ to $E$ and re-established all the hypotheses. 

If as result of a wrong guess (with $\mathit{ND}$), we incorrectly establish that $a$ and $a_m$ commute, then we have under-estimated $E$ and have to argue that the computation will eventually halt.

The wrong guess implies that there exists a $\rd(x)$ in the remainder of the concurrent run that is matched with a $\wt(x)$ in $a$. The monitor made the mistake of not including $x$ in the $V$ component of he grain signature of $a$, and therefore $\mathit{commute}$ issued the wrong verdict. The same choice (with $\mathit{ND}$) is reflected in the $C$ component of the monitor's state as $C[x \mapsto \mathit{true}]$. Therefore, when the monitor eventually reaches this $\rd(x)$, it halts because it does not have any transitions defined for this configuration. 

\item Concurrent: By definition, $a$ cannot be ordered against any entity $b$ outlined in the induction hypothesis. Let us assume that the monitor correctly guesses all the relevant components of the grain for $a$. If the monitor incorrectly decides that $a$ is ordered wrt $E$, then it means that it does not commute with at least one grain signature in $E$. By induction hypothesis, this signature must belong to some entity (grain or event) that has already been observed in $\sigma$ and is ordered wrt $e$. Therefore, by definition, $a$ is also ordered wrt $e$. 

If the monitored correctly sees $a$ as concurrent, then it does not update $E$ and therefore maintains the invariants. If $a$ is concurrent, but the monitor incorrectly sees it dependent on $E$, then this is an over-estimation of $E$, which can be dismissed. There exists another nondeterministic choice (as outlined above) that will get $E$ right. 
\end{itemize}

\end{proof}

\subsection{Proof of~\thmref{monitor-sound2}}
\seclabel{grain-concurrency-constant-space}

Let $L_\textsf{WF}$ be the set of words described by the regular expression in \equref{wf}.
Conside the set:
\begin{align*}
\begin{array}{rl}
L_\textsf{concurrent-marked-grains} = \{w \in L_\textsf{WF} |& \textsf{The two focal events demarcated by }\focalEv_1 \text{ and } \focalEv_2 \\ &\text{are concurrent under the marked grains in } w\}.
\end{array}
\end{align*}
From \thmref{monitor-sound}, we have that $L_\textsf{concurrent-marked-grains}$ is a regular set.

We will now focus on the set of words of the form $L_\textsf{WF-unmarked-grains} = \alphabet^*\focalEv_1\alphabet^*\focalEv_2\alphabet^*$ and the following subset of it:
\begin{align*}
\begin{array}{rl}
L_\textsf{concurrent-unmarked-grains} = \{w \in L_\textsf{WF-unmarked-grains} |& \textsf{The two focal events demarcated by }\focalEv_1 \text{ and } \focalEv_2 \\ &\text{are concurrent under any choice of grains in } w\}.
\end{array}
\end{align*}


Observe that the set $L_\textsf{concurrent-unmarked-grains}$ is obtained by projecting the language
 $L_\textsf{concurrent-unmarked-grains}$ from the alphabet $\alphabet \uplus \set{\focalEv_1, \focalEv_2, \bgn, \egn}$ to the sub-alphabet $\alphabet \uplus \set{\focalEv_1, \focalEv_2}$.
Since $L_\textsf{concurrent-marked-grains}$ is regular, and since regular languages are closed under projection, we immediately get that  $L_\textsf{concurrent-unmarked-grains}$ is also regular and hence there is a constant space monitor that recognizes it.


\section{Proofs of Section \ref{sec:sgrains}}\label{app:sgrains}

\subsection{Proof of Theorem \ref{thm:swgc}}

\begin{proof}
The high level idea is to try to linearize the graph $\mathbb{G}_{w, G}$, with the addition of an extra directed edge from the node whose grain contains $e_2$ to the one whose node contains $e_1$. Since there is no path from $e_1$ to $e_2$ in the graph, the addition of the edge cannot put the pair of events in the same strongly connected component. Moreover, the addition of the edge ensures that in every linearization of the graph $e_2$ appears before $e_1$. 

The claim is that the step-by-step linearization succeeds and produces an $\rdf{}$-equivalent run. 

At each step, let $\sigma$ be the linearization so far and let $R$ be the set of residual vertices left in the graph. Let $\mathbb{G}|_R$ be the graph induced on the vertices in $R$ through the edges of $\mathbb{G}$. 

Induction Hypothesis:
\begin{itemize}
\item[(i)]  $\sigma$ contains everything in a maximal strongly connected component completely, or not at all, never partially. As such, the same is true about each grain, which is the smallest strongly connected component in the absence of a larger one. 
\item[(ii)] Every event from $w$ is either in $\sigma$ or in some grain in $R$. 
\item[(iii)]  $\sigma$ preserves ${\sf po}$ and $\rdf{}$ of $w$.
\item [(iv)] There does not exist an edge between a node in $R$ to a note containing any element in $\sigma$ in graph $\mathbb{G}$.
\end{itemize}

Base case: $\sigma = \epsilon$, and all invariants hold. 

Induction step: depending on the composition of $R$, we extend $\sigma$ and reprove the induction hypothesis. Consider $\mathbb{G}|_R$, and the condensation of $\widehat{\mathbb{G}}|_R$ in which every edge in every maximally strongly connected component has been contracted. Observe that $\widehat{\mathbb{G}}|_R$ is acyclic. Therefore, there exists a node in it with no predecessors. Let us call this node $v$: 

\begin{itemize}
\item Case 1: $v$ is not a contracted node, and it corresponds to a node in $\mathbb{G}|_R$. We let $\sigma = \sigma v$, with the understanding that $v$ is a word that represents the corresponding grain; the grain can correspond to a single letter and the word can be that letter. Remove the node $v$ from $\mathbb{G}|_R$ (and all its adjacent edges).

Let us reprove the induction hypothesis:

\begin{itemize}
\item[(i)]  By definition of  $\mathbb{G}|_R$, $v$ is in a strongly connected component of size 1, and therefore this assumption holds again. 
\item[(ii)] Trivially true, because we just shift events from one side to the other. 
\item[(iii)] Any predecessors (through ${\sf po} \cup \rdf{}$) of the events in $v$ would have a path to $v$. If $v$ is a minimal element, that means that all predecessors should  be in $\sigma$ or in $v$ . Inside the word in $v$, we do not change the order of any events, and therefore ${\sf po}$ and $\rdf{}$ cannot be broken inside $v$. By induction hypothesis, inside $\sigma$, we have maintained ${\sf po}$ and $\rdf{}$. Therefore, it remains to argue that ${\sf po}$ and $\rdf{}$ are not broken when we concatenate the two. 

If ${\sf po}$ is broken, this means that some element in $\sigma$ must have been ${\sf po}$ ordered after some element in $p$. But, this is a contradiction to induction hypothesis (iv).  

It remains to argue that any dangling reads in $v$ (i.e. reads whose matching writes do not belong to $v$)  are  matched with the correct write as a result of concatenating $\sigma$ and $v$. Assume that is not the case; dangling read $\rd(x)$ is matched with $\wt_1(x)$ in $\sigma v$, while it was matched with $\wt_2(x)$ in $w$. We argued that $\wt_2(x)$ is in $\sigma$.  Therefore, it must be the case that $\wt_2(x)$ appears before $\wt_1(x)$ in $\sigma$. There are two possibilities for the original arrangement of the three events in $w$:

\begin{itemize}
\item $\wt_1(x) \dots \wt_2(x) \dots \rd(x)$: This means that as a result of our linearization of $\sigma$ so far, we ended up reordering the two writes. This is only possible if their corresponding grains commute. Otherwise, there would be an edge between them that would prevent us from linearizing them in the new order. But, since the dangling $\rd(x)$ is not part of the grain of $\wt_2(x)$, by definition, it cannot commute against any other grain with which it shares the variable; specifically, not the grain that includes $\wt_1(x)$. Contradiction!
\item $\wt_2(x) \dots \rd(x) \dots \wt_1(x)$: This means that as a result of our linearization of $\sigma$ so far, we ended up reordering $\rd(x)$ and $\wt_1(x)$. This is only possible if their corresponding grains commute. Otherwise, there would be an edge between them that would prevent us from linearizing them in the new order. But, since the dangling $\rd(x)$ is not part of the grain of the write operation that it reads from (i.e. $\wt_2(x)$), by definition, it cannot soundly commute against another grain with which it shares $x$ and the other grain writes to $x$; specifically, not the grain that includes $\wt_1(x)$. Contradiction!
\end{itemize}

\item [(iv)] It is implied by the choice of $v$ as a node with no predecessors in the remaining graph.
\end{itemize}

\item Case 2: $v$ is a contracted node, and therefore corresponds to a set of nodes $V_v$ from $\mathbb{G}|_R$. We reference induction hypothesis (i) to state the fact that all nodes in $V_v$ belong in $\mathbb{G}|_R$.  Let $u$ be the subsequence of $w$ that includes precisely all the events from the grains in $V_v$. Let $\sigma=\sigma u$.  Remove the node $v$ from $\mathbb{G}|_R$.

Let us reprove the induction hypothesis:

\begin{itemize}
\item[(i)]  By definition of  $\mathbb{G}|_R$, this holds
\item[(ii)] Trivially true, because we just shift events from one side to the other. 
\item[(iii)] Any predecessors (through ${\sf po} \cup \rdf{}$) of the events in $u$ would have a path to $v$. If $v$ is a minimal element, then all such predecessors are either in $\sigma$ or inside $u$. 

Inside the word in $u$, we do not change the order of any events (compared to how they appear in $w$), and therefore ${\sf po}$ and $\rdf{}$ cannot be broken inside $u$.  

It remains to argue that ${\sf po}$ and $\rdf{}$ are not broken when we concatenate $\sigma$ and $u$. The argument for ${\sf po}$ is similar to the previous case. 

We must argue that any dangling reads in $u$ are matched to the correct write as a result of concatenating $\sigma$ and $u$. Assume that is not the case; dangling read $\rd(x)$ is matched with $\wt_1(x)$ in $\sigma u$, while it was matched with $\wt_2(x)$ in $w$. 

$\wt_1(x)$ and $\wt_2(x)$ must both be in $\sigma$, because we are not reordering anything in $u$ and $\rd(x)$ would not otherwise be a dangling read.  Therefore, it must be the case that $\wt_2(x)$ appears before $\wt_1(x)$ in $\sigma$ and neither belongs to $u$. There are two possibilities for the original arrangement of the three events in $w$:

\begin{itemize}
\item $\wt_1(x) \dots \wt_2(x) \dots \rd(x)$: This means that as a result of our linearization of $\sigma$ so far, we ended up reordering the two writes. This is only possible if their corresponding grains commute. Otherwise, there would be an edge between them that would prevent us from linearizing them in the new order. But, since the dangling $\rd(x)$ is not part of the grain of $\wt_2(x)$, by definition, it cannot commute against any other grain with which it shares the variable; specifically, not the grain that includes $\wt_1(x)$. Contradiction!
\item $\wt_2(x) \dots \rd(x) \dots \wt_1(x)$: This means that as a result of our linearization of $\sigma$ so far, we ended up reordering $\rd(x)$ and $\wt_1(x)$. This is only possible if their corresponding grains commute. Otherwise, there would be an edge between them that would prevent us from linearizing them in the new order. But, since the dangling $\rd(x)$ is not part of the grain of the write operation that it reads from (i.e. $\wt_2(x)$), by definition, it cannot soundly commute against another grain with which it shares $x$ and the other grain writes to $x$; specifically, not the grain that includes $\wt_1(x)$. Contradiction!
\end{itemize}
\item [(iv)] It is implied by the choice of $v$ as a node with no predecessors in the remaining graph.

\end{itemize}

\end{itemize}

\end{proof}

\subsection{Proof of Theorem \ref{thm:collapse}}

\begin{proof}
Observe that $\mathbb{G}_{w,G}$ is an acyclic graph for a valid set of {\em contiguous} grains $G$, and as such the condensed graph $\widehat{\mathbb{G}}$ and $\mathbb{G}_{w,G}$ coincide. Therefore, the proof of the theorem says that any linearization of $\mathbb{G}_{w,G}$, in which the grains appear as contiguous subwords is $\rdf{}$-equivalent to the original run $w$. It remains to argue that any such linearization can be acquired through a sequence of swaps. 

The graph $\mathbb{G}_{w,G}$ coincides with the partial order that represents the the equivalence class of $w$ under the grain monoid induced by $G$. More precisely, $[H_G(w)]_{\widehat{I}_G}$ is an equivalence class in a classic partially commutative monoid, and therefore precisely coincides with all the linearization of the partial order induced on the elements of $H_G(w)$ by $\widehat{I}_G$. Observe that the edges of $\mathbb{G}_{w,G}$ coincide with the constraints in this partial order, by construction. Therefore, any linearization $u \in (\Sigma \cup \Sigma_G)^*$ of the graph, with the view of noting down every grain as $a_g$ (its corresponding alphabet symbol in $\Sigma_G$) can be obtained from $H_G(w)$ through a sequence of swaps from $\widehat{I}_G$. Since every such swap has a corresponding swap at the level of $\Sigma^*$, one can mirror the same swap sequence and transform $H^-1_G(u)$ to $w$.  
\end{proof}


\section{Proofs from~\secref{beyond-contiguous-grains}}

\subsection{Proof of~\propref{split}}

\begin{proof}[Proof Sketch.]
A straightforward induction on $i$ can be used to establish that $\grain^{(i)}$ is minimal.
\end{proof}

\subsection{Proof of~\lemref{minimal-grains-sufficient}}

\begin{proof}
First, let's establish soundess of $\indrel_{\grains'}$ using the characterization of~\thmref{sgrain}.
Suppose on the contrary that $\indrel_{\grains'}$ is not sound.
This means, by \thmref{sgrain}, we have a pair of grains $(g'_1, g'_2) \in \indrel_{\grains'}$
that violate one of the conditions of \thmref{sgrain}; without loss of generality, we assume that $g'_1$ appears before $g'_2$ in $\tr$.
Let $g_1, g_2 \in $ be the larger grains in $\grains$ such that 
$g'_i \in \splitGrains(g_i)$ for $i \in \set{1, 2}$.
The violations witnessed by $(g'_1, g'_2)$ can be one of the following.
\begin{itemize}
	\item $g'_1g'_2 \not\rfnovaleq g'_2g'_1$. In this case, there must be a variable $x$ and events $e = \rd(x)$,
	$f = \rdf{e}$ and $f' = \wt(x)$ ($f' \neq f$)
	such that one of the following holds: 
	(a) $e \in g'_1$ but $f \not\in g_1$ and $f'\in g'_2$,
	(b) $e \in g'_1$ but $f \in g_1 - g'_1$ and $f'\in g'_2$,
	(c) $e \in g'_2$ but $f \not\in g_2$ and $f'\in g'_1$, or
	(d) $e \in g'_2$ but $f \in g_2 - g'_2$ and $f'\in g'_1$,.
	In cases (a) and (c), we have $g_1g_2 \not\rfnovaleq g_2g_2$, contradicting soundness of $\indrel_\grains$
	In case (b) (resp. case (d)), we have that $g'_1 \not\in \splitGrains(g_1)$
	(resp. $g'_2 \not\in \splitGrains(g_2)$) since the minimal grain containing 
	$f$ must also contain $e$, contradicting minimality of splity grains (\propref{split}).
	
	\item There is a variable $x \in \VariableOf{g'_1} \cap \VariableOf{g'_2}$ such that
	$\wt \in \OpOf{g'_1, x} \cup \OpOf{g'_2, x}$ and an event $e = \rd(x)$ (with $f = \rdf{e}$) such that one of the following holds:
	(a) $f \in g'_1$ and $e \in g_1 - g'_1$,
	(b) $f \in g'_1$ and $e \not\in g_1$,
	(c) $e \in g'_1$ and $f \not\in g_1 - g'_1$,
	(d) $e \in g'_1$ and $f \not\in g_1$,
	(e) $f \in g'_2$ and $e \in g_2 - g'_2$,
	(f) $f \in g'_2$ and $e \not\in g_2$,
	(g) $e \in g'_2$ and $f \not\in g_2 - g'_2$, or
	(h) $e \in g'_2$ and $f \not\in g_2$.
	In cases (a), (c), (e) and (g), minimality of either $g'_1$ or $g'_2$ is violated.
	In cases (b), (d), (f) and (h), soundness of $\indrel_\grains$ is violated.
\end{itemize}

Let us now prove that any two events 
that are declared grain graph concurrent using $\grains$ will also be declared so using $\grains'$.
Towards this, let $\GGraph{\tr, \grains} = (V, E)$ and
$\GGraph{\tr, \grains'} = (V', E')$ be the respective grain graphs.
Observe that for any two grains $g'_1, g'_2 \in \grains'$,
if $(g'_1, g'_2) \in E'$, then $(g'_1, g'_2) \not\in \widehat{\indrel}_{\grains'}$
and there exist events $e_1 \in g'_1$ and $e_2 \in g'_2$ such that $e_1$ appears before $e_2$
and $(e_1, e_2) \not\in \indrel_\maz$.
This means $(g'_1, g'_2) \not\in \indrel_{\grains'}$
and thus $(g_1, g_2) \not\in \indrel_\grains$, where $g_1$ and $g_2$ are the corresponding
grains of $\grains$ from which $g'_1\in g_1$ and $g'_2 \in g_2$ are obtained after splitting.
Since we have $e_1 \in g_1$ and $e_2 \in g_2$, we can conclude that $(g_1, g_2) \in E$.
As a result, if there is a path
$(h'_1, h'_2), (h'_2, h'_3) \ldots (h'_k, h'_{k+1})$ in $\GGraph{\tr, \grains'}$,
then the following is a path in  $\GGraph{\tr, \grains}$:
$(h_1, h_2), (h_2, h_3) \ldots (h_k, h_{k+1})$,
where $h'_i \in h_i$ is the corresponding larger grain.
Hence, if $e_1$ and $e_2$ are declared grain graph concurrent by $\grains$, they
will also be declared so using $\grains'$
\end{proof}

\subsection{Proof of~\lemref{minimal-grains-active-bound}}

\begin{proof}
Let $\pi$ be a prefix.
For a grain $g$ that is active at the end of $\pi$,
we let $\vars_g$ be the set of variables $x$ such that there is a read of $x$ which is not in $g$,
but its corresponding write event is in $g$.
Observe that since each of the grains in $\grains$ is minimal, we must have
$\vars_g \neq \emptyset$ for each $g \in \actv{\pi, \grains}$.
Now, consider two distinct grains $g, g' \in \actv{\pi, \grains}$.
We must have $\vars_g \cap \vars_{g'} = \emptyset$, as otherwise
two different write events on the same variable will be pending in $\pi$ which is impossible.
This gives us $|\actv{\pi, \grains}| \leq |\vars|$.
\end{proof}

\subsection{Proof of~\propref{summarized-reachability}}

$(\Leftarrow)$. It is easy to see that if there is a path from 
$\focalGrain{1}$ to $\focalGrain{2}$ in $\SGGraph{\pi, G}$, then there must also be
a path from $\focalGrain{1}$ to $\focalGrain{2}$ in $\GGraph{\pi, G}$
and thus a paths in $\GGraph{\tr, G}$.
This is because edges of  $\SGGraph{\pi, G}$ are paths of  $\GGraph{\pi, G}$.

$(\Rightarrow)$.
Consider a path $\rho$ from $\focalGrain{1}$ to $\focalGrain{2}$ in $\GGraph{\tr, G}$.
This path can be expressed as $\rho = \rho_1 \rho_2 \ldots \rho_k$, 
where for each sub-path $\rho_i$, the start and the end vertices are vertices of $\SGGraph{\pi, G}$,
and each of the intermediate vertices are not.
In other words, $\textsf{source}(\rho_i) \deadpath{\pi, \grains} \textsf{target}(\rho_i)$.
As a result, $(\textsf{source}(\rho_i), \textsf{target}(\rho_i))$ is an edge in $\SGGraph{\pi, G}$.
This gives a path from $\focalGrain{1}$ to $\focalGrain{2}$ in $\SGGraph{\tr, G}$.

\subsection{Proof of \thmref{weak-concurrency-fixed-grains-monitor}}


\newcommand{\stOf}[1]{\textsf{proj}(#1)}
\newcommand{\gIDOf}[1]{\textsf{ID}(#1)}
\newcommand{\grainOf}[1]{\grain_{#1}}

We will prove this theorem by proving that after having processed prefix $\pi$
of the input word $\tr$, the state $q$ of the automaton $\autsup{{\sf GG}}$
reflects the summarized graph $\SGGraph{\pi, G}$.
Towards this, let us define the notation $\stOf{\pi, G}$
to denote the finite projection of $\SGGraph{\pi, G} = (V_\pi, E_\pi)$,
i.e., $\stOf{\pi, G} = \stuple{V, E, C, P, SC, SP}$
such that the following holds.
In the following, we use $\gIDOf{g}$ to denote the grain identifer of a grain $g$.
Also, we use $\grainOf{\pi, u}$ to denote the grain identifier
of the latest grain in $\pi$ with identifier $u$. 
\begin{itemize}
	\item $V = \setpred{\gIDOf{g}}{g \in V_\pi}$.

	\item $E = \setpred{(\gIDOf{g_1}, \gIDOf{g_2})}{(g_1, g_2) \in E_\pi}$.

	\item For each $u \in V$, 
	$\contents(u) = \contents_\pi(\grainOf{u, \pi})$.

	\item For each $u \in V$, $\pvars(u) = \pvars_\pi(\grainOf{u, \pi})$.

	\item For each $u \in V$,
	$\summaries(u) = \summaries_\pi(\grainOf{u, \pi})$.

	\item For each $u \in V$,
	$\spvars(u) = \spvars_\pi(\grainOf{u, \pi})$.
\end{itemize}

The invariant that the automaton maintains is the following:
\begin{claim}
\claimlabel{automaton-invariant-scattered}
Let $\tr \in L_\textsf{VMG}$ and let $\grains$ be the corresponding grains.
Let $\pi$ be some prefix of $\tr$. 
Let $q$ be the state of $\autsup{{\sf GG}}$ after processing $\pi$.
We have, $q = \stOf{\pi, \grains}$.
\end{claim}

\begin{proof}
We establish this by inducting on the length of $\pi$.
In the base case, $\pi = \epsilon$ is the empty trace and thus $\SGGraph{\pi, \grains}$
is the empty graph and the corresponding state of the automaton reached,
namely $q_0$ also matches $q_0 = \stOf{\pi, \grains}$.
Let us now consider the run $\pi = \rho \cdot e$, where $e \in \gMarkAlphabet$
and let $q_\rho$ be the state of $\autsup{{\sf GG}}$ after having processed $\rho$.
By the inductive hypothesis, we have $q_\rho = \stOf{\rho, \grains}$
We can now establish the invariant about $q_\pi = \delta_{\sf GG}(q_\rho, e)$
by doing a case-by-case analysis on $e$.
In the following we will use the notation $q_\rho = \stuple{V_\rho, E_\rho, \contents_\rho, \pvars_\rho, \summaries_\rho, \spvars_\pi}$.

\begin{description}
	\item[Case $e = (i, (\bgn, Y))$]. Since $\tr$ belongs to $L_\textsf{VMG}$,
	we know that at the end of $\rho$, there is no active grain with identifier $i$,
	and hence $i \not\in V$.
	Clearly, $\SGGraph{\pi, \grains}.V = \SGGraph{\rho, \grains}.V \uplus \set{i}$.
	Also, no new edges must be added in $\SGGraph{\pi, \grains}.E$ over $\SGGraph{\rho, \grains}.E$.
	Likewise, since the only new event added is the begin event of a new transaction, the contents of
	each tracked grain stays the same. 
	The pending variables of the grain with ID $i$ is captured by the set $Y$.
	Finally, since there is no path in $\GGraph{\pi, \grains}$ that originate at $i$,
	neither the summaries nor the pending vars of reachable grains change.
	Thus, $q_\pi = \stOf{\pi, \grains}$.

	\item[Case $e = (i, \egn)$, $i\in\set{\focalEv_1, \focalEv_2}$]. No change in the summarized graph happens and this is also reflected in $q_\pi$ which is the same as $q_\rho$.

	\item[Case $e = (i, \egn)$, $i\not\in\set{\focalEv_1, \focalEv_2}$].
	In this case, the grain with identifier $i$ is not present in $\SGGraph{\pi, \grains}$ since it is no longer active.
	Indeed, we have $i \not\in V_\pi = V_\rho - \set{i}$,
	since $i$ is a dead grain at the end of $\pi$.
	For the same reason, $\pvars_\pi(i) = \emptyset$ and $\contents_\pi(i) = \emptyset$.
	Now, the set of dead paths are as follows: 
	$g \deadpath{\pi, \grains} g'$
	iff 
	$\grainOf{i, \pi} \not\in {g, g'}$ and either
	(a) $g \deadpath{\rho, \grains} g'$ or 
	(b) ($g \deadpath{\rho, \grains} \grainOf{i, \pi}$ and $\grainOf{i, \pi} \deadpath{\rho, \grains} g'$).
	Thus, $\SGGraph{\pi, \grains}.E = (\SGGraph{\rho, \grains}.E - \setpred{(\grainOf{i, \pi}, g), (g, \grainOf{i, \pi})}{g \neq \grainOf{i, \pi}}) \cup \setpred{(g, g')}{(g, \grainOf{i, \pi}) \in \SGGraph{\rho, \grains}.E, (\grainOf{i, \pi}, g') \in \SGGraph{\rho, \grains}.E}$.
 	Indeed, the function $\mergeEdge{E}{i}$ in the monitor captures this accurately.
	Finally, the summaries of a grain $g$ in the summarized graph change based on the new dead paths.
	In particular, for a grain $g$, the set of grains that are reachable from $g$ using a dead path
	are either those that were already reachable in $\rho$, or those that
	were reachable from $\grainOf{i, \pi}$, given $g \deadpath{\rho, \grains} \grainOf{i, \pi}$.
	This means that $\SGGraph{\pi, \grains}.\summaries(g) = \SGGraph{\rho, \grains}.\summaries(g) \cup \SGGraph{\rho, \grains}.\contents(\grainOf{\pi, i}) \cup \SGGraph{\rho, \grains}.\summaries(\grainOf{\pi, i})$
	if $(g, \grainOf{i, \pi}) \in \SGGraph{\rho, \grains}.E$,
	and $\SGGraph{\pi, \grains}.\summaries(g) = \SGGraph{\rho, \grains}.\summaries(g)$ for other active grains.
	Indeed this is accurately reflected in $\summaries_\pi$ using $\mergeSum{SM, M}{E}{i}$.
	The reasoning for $\spvars_\pi$ follows similar reasoning.

	\item[Case $e = (i, a)$, $a \in \alphabet$].
	In this case, no new active grain is seen.
	However, new edges may be formed between active grains.
	First note that no new edge between grains $g$ and $g'$ can be added at this point
	if $\grainOf{i, \pi} \not\in \set{g, g'}$ since no new dead paths can be formed at this point.
	Now consider a grain $g$ in $\SGGraph{\rho, \grains}.V = \SGGraph{\pi, \grains}.V$.
	An edge from $g$ to $\grainOf{i, \pi}$ may be inferred
	if a new dead path $g \deadpath{\pi, \grains} \grainOf{i, \pi}$ may be inferred.
	This can happen if there is an existing dead path $g \deadpath{\rho, \grains} g'$
	for some $g'$ (active or dead) such that there is an event $e \in g'$ (on variable $x$) 
	that is dependent with $a$ and $x$ is either pending in $g'$ or in $\grainOf{i, \pi}$.
	If $g'$ is dead, then inductively we have
	$e \in \SGGraph{\rho, \grains}.\summaries(g)$ and $x \in \SGGraph{\rho, \grains}.\spvars(g)$.
	Otherwise we would have 
	$g' = g$ and thus inductively  $e \in \SGGraph{\rho, \grains}.\contents(g)$ and $x \in \SGGraph{\rho, \grains}.\pvars(g)$.
	This is captured using $\depEdge{\contents(j){\cup}\summaries(j),\; a,\; \pvars(j){\cup}\spvars(j){\cup}\pvars(i)}$ in the monitor.
\end{description}
\end{proof}

As a corollary of \claimref{automaton-invariant-scattered}, it is easy to see that
for a run in $\tr \in L_\textsf{VMG}$ annotated with grains $\grains$, 
if $\tr$ is accepted, then there is a path from $\focalGrain{1}, \focalGrain{2}$
in $\SGGraph{\tr, \grains}$, and further
if $\focalGrain{1}, \focalGrain{2}$
in $\SGGraph{\tr, \grains}$, then the automaton reaches a state $q \in F_{\sf GG}$
where there is a path from node $\focalEv_1$ to node $\focalEv_2$.
This proves the theorem.

\subsection{Proof of \thmref{weak-grain-concurrency-regular}}


First, we outline the language $L_\textsf{VMG}$ and show that it is regular.
Observe that $L_\textsf{VMG} = L_\textsf{valid} \cap L_\textsf{minimal}$
where $L_\textsf{valid}$ is the set of annotated runs with correctly marked grains and pending variables
and $L_\textsf{minimal}$.
We will show that both of these languages are regular, and thus their intersection is regular.

First consider the language $L_\textsf{valid}$;
it can expressed as $L_\textsf{valid} = L_\textsf{valid-grains} \cap L_\mathsf{\focalEv}$,
where the first language is the collection of all words such that
for every grain identifier, if we project the word to that grain identifier,
then no two grains overlap, while the second language is the set of all
words that contain exactly one occurrence of the begin event 
corresponding to each of $\set{\focalEv_1, \focalEv_2}$.
For a grain identifier $i \in \grainIDs$,
let us use $L_\textsf{i}$ to denote the set of words that correspond
to the contents of a valid grain whose identifier is $i$.
Then, $L_\textsf{grains, i} = (\sum_{Y\subseteq \vars}(i, (\bgn, Y)) \cdot L_\textsf{i, Y} (i, \egn))^*$
is the set of words that correspond to valid sequences of grain $i$.
Then, consider the homomorphism $\pi_i$ that projects all letters whose grain id is $i$ to themselves,
and all other letters to $\epsilon$.
Then, $L_\textsf{valid-grains-non-pending} = \bigcap\limits_{i \in \grainIDs} (\pi_i^{-1} L_\textsf{grains, i})$  
is the set of words whose projection to any given grain ID is well formed;
this is obtained by intersecting the inverse homomorphic image of regular languages, hence it is regular.
Now, a simple automaton can also check if the pending variables are consistent with the annotation;
let $L_\textsf{pending}$ be the langauge of this automaton.
Hence, $L_\textsf{valid-grains} = L_\textsf{valid-grains-non-pending} \cap L_\textsf{pending}$
and thus regular.
Finally, $L_\mathsf{\focalEv}$ is the set of all words that contain exactly one occurence of
a letter from $\setpred{(\focalEv_1, (\bgn, Y))}{ Y \subseteq X}$
and one occurence of a letter from  $\setpred{(\focalEv_2, (\bgn, Y))}{ Y \subseteq X}$;
clearly this is regular.

Now, we consider the set of runs where the grains are minimal.
Minimality can also be checked using an automaton which tracks, for each active grain,
whether there is at least one pending read not yet seen, by guessing this read and later validating it (in a left to right pass).
Thus, $L_\textsf{VMG}$ is regular.

Now, the set of runs in which two given events (marked with $\focalEv_1$ and $\focalEv_2$)
are deemed concurrent are essentially those that are obtained by the homomorphic image of
the words in $L_\textsf{VMG}$, ensure that these two events are in the focal grains,
and are also accepted by $\autsup{\sf GG}$.
Here, the homomorphism we consider maps begin and end letters to $\epsilon$,
and for other letters it removes their grain identifiers.
This is clearly a regular language.
\end{document}